\documentclass[a4paper,USenglish, cleveref, autoref, thm-restate]{lipics-v2021}

\pdfoutput=1

\hideLIPIcs

\nolinenumbers

\def\cqedsymbol{\ifmmode$\lrcorner$\else{\unskip\nobreak\hfil
\penalty50\hskip1em\null\nobreak\hfil$\lrcorner$
\parfillskip=0pt\finalhyphendemerits=0\endgraf}\fi}

\crefname{lemma}{Lemma}{Lemmas}
\crefname{theorem}{Theorem}{Theorems}
\crefname{proposition}{Proposition}{Propositions}
\crefname{question}{Question}{Questions}
\crefname{definition}{Definition}{Definitions}
\crefname{conjecture}{Conjecture}{Conjectures}
\crefname{observation}{Observation}{Observations}
\crefname{corr}{Corollary}{Corollaries}
\crefformat{equation}{(#2#1#3)}
\Crefformat{equation}{Equation #2(#1)#3}
\crefname{remark}{Remark}{Remarks}
\crefname{scenario}{Scenario}{Scenarios}
\crefname{resultish}{Result-ish}{Results-ish}

\newenvironment{subproof}[1][\proofname]{%
  \begin{proof}[#1]%
}{%
  \end{proof}%
}

\newcommand{\cE}{\mathcal{{E}}}

\def\T{\mathcal{T}} % supertree
\let\ge\geqslant
\let\leq\leqslant
\let\geq\geqslant

\newcommand{\say}[1]{``#1''} %place quotes around something on korean keyboard
\definecolor{cobalt}{RGB}{0,71,171}
\definecolor{brightpink}{RGB}{255,0,204} 
\definecolor{bordeaux}{RGB}{100,0,50}
\definecolor{grun}{rgb}{0, 0.5, 0.5}
\definecolor{violet}{RGB}{177, 0.5, 255}
\definecolor{CornflowerBlue}{rgb}{0.39, 0.58, 0.93}
\definecolor{DarkGoldenrod}{rgb}{0.72, 0.53, 0.04}
\definecolor{BritishRacingGreen}{rgb}{0.0, 0.26, 0.15}
\definecolor{DarkMagenta}{rgb}{0.55, 0.0, 0.55}
\definecolor{AO}{rgb}{0.0, 0.5, 0.0}
\definecolor{BostonUniversityRed}{rgb}{0.8, 0.0, 0.0}
\definecolor{myRed}{rgb}{0.8, 0.0, 0.0}
\definecolor{DarkMidnightBlue}{rgb}{0.0, 0.2, 0.4}
\definecolor{DarkTangerine}{rgb}{1.0, 0.66, 0.07}
\definecolor{AppleGreen}{rgb}{0.55, 0.71, 0.0}
\definecolor{BrightUbe}{rgb}{0.82, 0.62, 0.91}
\definecolor{Amethyst}{rgb}{0.6, 0.4, 0.8}
\definecolor{DarkGray}{rgb}{0.52, 0.52, 0.51}
\definecolor{Gray}{rgb}{0.66, 0.66, 0.66}
\definecolor{BananaYellow}{rgb}{1.0, 0.88, 0.21}
\definecolor{Amber}{rgb}{1.0, 0.75, 0.0}
\definecolor{LightGray}{rgb}{0.83, 0.83, 0.83}
\definecolor{PrincetonOrange}{rgb}{1.0, 0.56, 0.0}
\definecolor{DeepCarrotOrange}{rgb}{0.91, 0.41, 0.17}
\definecolor{nblue}{rgb}{0.38, 0.51, 0.71} %glaucous, 97,130,181, #6182B5
\definecolor{darkblue}{RGB}{17, 42, 60} % 112A3C
\definecolor{nred}{RGB}{175, 49, 39} % AF3127
\definecolor{norange}{RGB}{217, 156, 55} % D99C37
\definecolor{ngreen}{RGB}{144, 169, 84} % 90A954
\definecolor{palegreen}{RGB}{197, 184, 104} % C5B868
\definecolor{nyellow}{RGB}{250, 199, 100} % FAC764
\definecolor{brokenwhite}{RGB}{218, 192, 166} % DAC0A6
\definecolor{brokengrey}{rgb}{0.77, 0.76, 0.82} % {196,194,209}, C4C2D1

\newcommand\dd{\hbox{-}} %short dash, eg. an edge

\newcommand{\spanningtree}{\mathsf{SpanningTree}}

\newcommand{\tab}{\mathsf{Table}}
\newcommand{\comp}{\mathsf{Components}}
\newcommand{\pieces}{\mathsf{Pieces}}

\newcommand{\htab}{\mathsf{HTable}}
\newcommand{\lp}{\mathsf{LongestPaths}}
\newcommand{\lcp}{\mathsf{LongestConstrainedPath}}

\DeclareMathOperator{\id}{Id}

\newcounter{regle}

\usepackage{tikz}

\declaretheorem[name=Question, sibling=theorem]{introquestion}

\title{Local certification of forbidden subgraphs}

\author{Nicolas Bousquet}{CNRS, INSA Lyon, UCBL, LIRIS, UMR5205, F-69622 Villeurbanne, France }{}{}{}{}
\author{Linda Cook}{Discrete Mathematics Group, Institute for Basic Science (IBS), Republic of Korea}{}{}{}{}
\author{Laurent Feuilloley}{CNRS, INSA Lyon, UCBL, LIRIS, UMR5205, F-69622 Villeurbanne, France }{}{}{}{}
\author{Théo Pierron}{CNRS, INSA Lyon, UCBL, LIRIS, UMR5205, F-69622 Villeurbanne, France }{}{}{}{}
\author{Sébastien Zeitoun}{CNRS, INSA Lyon, UCBL, LIRIS, UMR5205, F-69622 Villeurbanne, France }{}{}{}{}

\authorrunning{N. Bousquet, L. Cook, L. Feuilloley, T. Pierron, S. Zeitoun}



\acknowledgements{  L. Feuilloley thanks the participants of the PODC-DISC zulip chat for references about high-low degree techniques.
    This work was supported by ANR project GrR (ANR-18-CE40-0032). LC was supported by the Institute for Basic Science (IBS-R029-C1).}%optional

\begin{document}

\maketitle
\begin{abstract}

Detecting specific structures in a network has been a very active theme of research in distributed computing for at least a decade. In this paper, we investigate the topic of subgraph detection from the perspective of local certification.
Remember that a local certification is a distributed mechanism enabling the nodes of a network to check the correctness of the current configuration, thanks to small pieces of information called certificates. Our main question is: for a given graph $H$, what is the minimum certificate size that allows checking that the network does not contain $H$ as a (possibly induced) subgraph?

We show a variety of lower and upper bounds, uncovering an interesting interplay between the optimal certificate size, the size of the forbidden subgraph, and the locality of the verification. Along the way we introduce several new technical tools, in particular what we call the \emph{layered map}, which is not specific to forbidden subgraphs and that we expect to be useful for certifying other properties.  
\end{abstract}

 \thispagestyle{empty}

\maketitle

\newpage{}

\clearpage
\setcounter{page}{1}

%\tableofcontents

%\newpage{}

\section{Introduction}

\subsection{Context}

Finding some given small structures in a graph, triangles for example, has become a major theme in the area of distributed graph algorithms. 
A lot of effort has been put recently in understanding various versions of this problem (detection, listing, counting, and testing) in several congested models of distributed computing.\footnote{We refer to the recent survey by Censor-Hillel~\cite{Censor-Hillel22} 
for an introduction to the topic and a full bibliography.}
%This fundamental problem turned out to be very challenging and has been an incubator for new and important techniques, including expander decompositions~\cite{ChangPZ19} which has since been used in other contexts, for example for derandomizing minimum spanning tree construction~\cite{ChangS20}.  

Substructure detection has also become an important research topic in the more specific field of \emph{local certification}. 
In local certification, one is interested in how much information has to be given to each node to allow them to collectively verify some given property. 
More precisely, for a given property, a local certification is an algorithm performed independently by all the vertices of the graph. Every vertex $v$ takes as input its neighbors, each coming with a piece of information called \emph{certificate}, and outputs a binary decision, accept or reject. 
A certification scheme is correct if the following holds: there exists an assignment of certificates such that every node accepts when running the algorithm, if and only if, the graph satisfies the property. 
Actually, the notion of certification originates from self-stabilization~\cite{KormanKP10}, where one certifies the output of an algorithm, but in this paper we will focus on properties of the network itself.
We refer to~\cite{Feuilloley21} for an introduction to local certification and to Section~\ref{sec:model} for a formal definition of the model.

The classical measure of quality of a local certification is the size of the certificates, as a function of~$n$, the size of the graph. 
For example, certifying that a graph is acyclic can be done with certificates of size $O(\log n)$, and this is optimal. 
One often separates graph properties into two broad categories: the ones that have a compact certification, that is of (poly)logarithmic size, and the ones that do not, and typically polynomial size certificates.\footnote{Any property can be certified with a quadratic number of bits in this model, with certificates encoding the adjacency matrix, see \emph{e.g.}~\cite{Feuilloley21}.} 
Understanding which properties fall in each of these regimes has been a major direction in the area in recent years. While many tools have been designed towards certification schemes of (poly)logarithmic size, not so much is known for polynomial subquadratic schemes. One of the goal of this paper is to propose a very generic technique, called layered map partitions, that can be used as a blackbox in this regime.
%, see for example \cite{FeuilloleyBP22, FraigniaudMRT22, FraigniaudM0RT23}. 

When it comes to local certification of network substructures, the positive case, that is certifying that a given substructure exists in the network, is easy.
For example, if we wish to certify that a graph $G$ contains a graph $H$ as an (induced) subgraph, we can simply put the names of the vertices involved in the structure
%added this because of the minors discussion below
in the certificates of all vertices, and use a locally-encoded spanning tree to point to a node of the structure. In this case certificates are of size $\mathcal{O}_H(\log n)$.
%I thought it was confusing to only mention this giving all verticews when we talk about minors in the next paragraph
(Similar ideas allow us to certify that $G$ contains $H$ as a minor, induced subdivision, etc with $\mathcal{O}_H(\log n)$ bits). 
The real challenge is then to certify that some structure is \emph{absent} from the network, which is intuitively a much more global property.
In that direction, the community initiated the study of several substructures in the last years. 
For instance, one of the key conjectures of the area states that, for every graph $H$, certifying  $H$-minor-freeness can be done with certificates of size $O_H(\log n)$. This conjecture, which has been proved for many restricted graphs $H$~\cite{FeuilloleyFMRRT21, FeuilloleyF0RRT23, EsperetL22,BousquetFP21,FeuilloleyBP22,FraigniaudMRT22} and in a relaxed version~\cite{EsperetN22}, remains widely open. 

Classically in local certification, vertices can only see the certificates of their neighbors before taking their decision. However, in the last few years, many works studied larger views, where each vertex now has access to the certificates of nodes at some fixed distance. This distance is called the \emph{verification radius}. The goal is to describe the behavior of certificates sizes as a function of the verification radius. This has been investigated for constant radius (larger than 1) in ~\cite{GoosSuomela2016}, and non-constant radius in \cite{OstrovskyPR17, FeuilloleyFHPP21} (see also~\cite{FeuilloleyJKS24}), with recent new developments motivated by the so-called trade-off conjecture~\cite{FischerOS21, BousquetFZ24}.

\paragraph*{Certifying forbidden subgraphs}
In this paper we focus on subgraphs and in particular on the following question.
\begin{introquestion}
   What is the optimal certificate size for certifying that the network does not contain a fixed graph $H$ as an induced subgraph or as a non-necessarily induced subgraph, respectively?
\end{introquestion}
Remember that a graph $H$ is a (non-induced) \emph{subgraph} of $G$ if it can be obtained from $G$ by removing vertices and edges, and an induced subgraph of $G$ if it can be obtained from $G$ by removing only vertices (and the edges adjacent to these removed vertices).\footnote{In the literature, these two cases are simply called \emph{subgraph} and \emph{induced subgraph}, but to better differentiate we often add \emph{non-induced} or \emph{non-necessarily induced} to qualify the first case.}
%$^{,}$\footnote{Induced subgraphs detection have been less popular, for no obvious reason, except maybe that they are difficult to manage. A paper that focuses on induced subgraphs is~\cite{NikabadiK21} and a major part of \cite{FraigniaudM0RT23, p4-montealegre2021compact} is detecting whether a graph contains an induced four vertex path.}
(Induced) subgraphs, which are very localized in the graph, might look easier to manage than minors, which can span a large part of the graph (because of the contraction operation). 
And indeed, if one allows the local verification algorithm to look at large enough (constant) distance, every node can see whether the forbidden subgraph appears or not in its neighborhood. 
But otherwise, the problem seems to become more challenging than for minors. 
Intuitively, since minors can appear in many ways in a graph, forbidding one constrains the graph structure a lot, and one can use this structure for certification. 
For example, forbidding a triangle as a minor implies that the graph is a tree (which can be certified with $O(\log n)$ bits), while forbidding a triangle as an induced subgraph still leaves a complex graph class.
% Indeed, there are strong graph theoretic results on the structure of graphs with forbidden minors (see e.g. \cite{minors-NORIN2023109020, minors-robertson1985graph, minors-robertson1986graph}) for which there are no equivalents for forbidden induced subgraphs (or even for subgraphs). 

%Similarly, certifying graphs not containing $H$ as an induced subgraph seems to be much less challenging than the (not necessarily induced) subgraph case. 
%Classical results in extremal graph theory \cite{KST54,ES46} state that graphs which do not contain $H$ as a subgraph (and thus $H$-minor-free) cannot contain too many edges. 
%However, no such bound on the density of graphs not containing $H$ as an \emph{induced} subgraph exists, unless $H$ is a complete graph.

% \begin{figure}[!ht]
 %       \centering
 %       \includegraphics[width=.5\textwidth]{disc/non-containment.png}
 %       \caption{As illustrated by this Venn diagram, one of the key challenges of this paper is that forbidding a graph $H$ as a (not necessarily induced) subgraph may yield a much more complicated and larger class of graphs than forbidding $H$ as a minor.
  %      In general, the class of graphs obtained by forbidding $H$ as a minor is better understood than the class of graphs obtained by forbidding $H$ as a subgraph which in turn is better understood than the class of graphs obtained by forbidding $H$ as an induced subgraph. 
   %     Prior to this paper, the bulk of research on certifying forbidden substructures has been about forbidding minors, which we expect to be the simplest to certify.
%\end{figure}

\paragraph*{Adapting previous work for cliques and trees}

In related models (such as \textsc{Congest} or \textsc{Broadcast Congest}), the first studied cases of subgraph detection were cliques and trees. Some results can be directly adapted for local certification.
%In particular, it is known that results from the \textsc{Broadcast Congest} model are easy to translate to local certification (see Appendix~\ref{app:additional-related work} for more on \textsc{Broadcast Congest} and its links with local certification).
%The following is a consequence of the proof technique of a result giving a lower bound for detecting cliques in \textsc{Broadcast Congest}.
%\footnote{Actually, \cite{CrescenziFP19} only mentions the result for triangles, but it is easy to generalize it to cliques. Indeed, if there would exist a more compact certification of larger cliques, one could use it to certify triangles, by simulating $k-3$ virtual nodes adjacent to every node of the graph, \textit{i.e.} by giving their certificate in a $K_k$-freeness certification to every node of the original graph. As pointed out in private communication \cite{ChaniotisCHS24} this technique can be used to obtain lower bounds for certifying $H$-freeness in one round for many choices of a fixed graph $H$.}
%(Note that for cliques, subgraphs and induced subgraphs are the same.)
For cliques\footnote{Remember that $K_k$ denotes the clique on $k$ vertices, and we will also use $P_k$ for paths on $k$ vertices.} of size $k$, one can easily obtain a $O(n \log n)$ upper bound by writing in the certificate of each node $v$ the identifiers of every neighbor of $v$. The upper bound can then be reduced to $O(n)$ using a renaming argument~\cite{BousquetEFZ24}.
Combining an observation from~\cite{CrescenziFP19} with the known lower bound in the \textsc{Broadcast Congest} model from \cite{DruckerKO13} yields a matching linear lower bound on the certificate size when $k>3$, and an almost matching bound for triangles (see Subsection~\ref{app:additional-related work} for more on \textsc{Broadcast Congest} and its links with local certification). To summarize, the following holds:

\begin{restatable}[\cite{BousquetEFZ24}]{theorem}{ThmKkFree}
    \label{thm:Kk-free}
    For $k>3$, the optimal certificate size for $K_k$-free graphs with verification radius $1$ is $\Theta(n)$. For $k=3$, it lies between $\Omega(n/e^{O(\sqrt{\log n})})$ and $O(n)$ bits.
\end{restatable}

Note that for cliques being induced and non-necessarily induced does not make a difference. So Theorem~\ref{thm:Kk-free} suggests that certifying $H$-freeness for dense graphs $H$ is difficult. One can wonder what happens for very sparse graphs.
%Since the bound is polynomial in $n$, certifying forbidden subgraphs is difficult. \textcolor{red}{Did we already said that usually we want log n...etc... ?}
For trees, the situation is actually different (at least for non-induced case) because of the following result proved in Subsection~\ref{app:additional-related work} by combining techniques from previous work~\cite{KorhonenR17, EvenFFGLMMOORT17, FeuilloleyBP22}. 

\begin{restatable}{theorem}{ThmNonInducedTrees}
    \label{thm:non-induced-trees}
    For any tree $T$, certifying that a graph does not contain $T$ as a (non-necessarily induced) subgraph can be done with $O(\log n)$ bits, when verification radius is $1$, and this is tight for large enough trees.
\end{restatable}

% Actually, one can also derive this result from the meta-theorem of~\cite{FeuilloleyBP22} (see Corollary~2.7, noting that forbidding a path of length $t$ as a minor is equivalent to forbidding it as a (non-induced) subgraph). From the same paper, one can adapt a matching lower bound (see Section~7 in \cite{FeuilloleyBP22}).

One can then naturally wonder if the same holds for \emph{induced} trees. In particular:

\begin{introquestion}
    What is the optimal certificate size for induced trees? 
\end{introquestion}

For this question, the only non-trivial case known in the literature is for $P_4$-free graphs, which can be certified with $O(\log n)$ bits~\cite{FraigniaudM0RT23}. This proof is not helpful to understand the generic case, since it heavily uses the fact that $P_4$-free graphs, also known as co-graphs, have a very specific structure.
One of the main goals of this paper is to understand better the certification of $P_k$-free graphs, as a stepping stone to the more general case. 

\paragraph*{Trade-off between verification radius and certificate size}

In the case of cliques $K_k$, if we were allowed to have a verification radius $2$, we would not need any certificate since every vertex belonging to a $K_k$ subgraph (if one exists) would be able to see the whole clique. To get a similar certification scheme for paths, vertices require a verification radius larger than half the length of the forbidden path. A natural question is what happens when the verification radius is at least $1$, but not large enough to see the whole subgraph? Or in other words:

\begin{introquestion}
    What is the optimal size of a local certification for forbidden induced or non-induced subgraphs, as a function of the verification radius and of the subgraph size?
\end{introquestion}

%One expects that there is a trade-off between certification size and verification radius, but it is completely unclear whether it should be smooth, for example.

In terms of techniques, increasing the verification radius makes proving lower bounds more difficult. Indeed, a vertex sees at most $n$ edges at radius 1, but may see a quadratic number of edges at radius 2 or more, which stops easy counting arguments from being useful. 
On the other hand, increasing the radius, even just from $1$ to $2$, allows for more sophisticated upper bound techniques, as we will demonstrate with our novel technique of \say{layered maps}. 
Note that, unlike the papers previously mentioned, our paper concerns verifications radii that are a fraction of the size of the subgraph $H$ not increase with the size of the entire graph. So, for fixed $H$ we consider upper and lower bounds for certifying $H$-freeness where the verification radius is a \emph{constant} depending only on $H$. 

%\textcolor{red}{Finally note that most of the existing results so-far in the polynomial regime class problems in one of the following categories: certificates of size $\theta(n)$ or certificates of size $\theta(n^2)$ (up to subpolynomial factors) \textcolor{red}{discuss exception?}. We believe that subgraph detection actually has a complexity in between. Designing tools to prove "in-between" complexity is one of the objective of this paper.}

The main topic of this paper is resolving the two questions above and some of their variants. In the next section we will detail our results and sketch some proof techniques. For interested readers, additional background is given in Subsection~\ref{app:additional-related work}.

%%%%%%%%%%%%%%%%%%%%%%%%%%%%%%%%%%%%%%
\subsection{Our results and techniques}
\label{subsec:results-techniques}
%%%%%%%%%%%%%%%%%%%%%%%%%%%%%%%%%%%%%%

%The results of this paper are upper and lower bounds for certifying that a given subgraph does not appear in the network. The lower bound applies to any verification radius $k$,\footnote{Remember that the verification radius is the radius of the ball that the local verification algorithm gets as input.} while the upper bounds apply for $k\geq 2$.

\paragraph*{Paths as benchmark}

Before discussing our results and techniques in full generality, let us focus on our upper and lower bounds for paths, which appear to be amongst the hardest cases.\footnote{Intuitively because paths are connected graphs of maximum diameter given a fixed size.} See Table~\ref{table:paths-results}. 

\begin{table}[!h]
    \centering
    \begin{tabular}{c|c|c}
        Forbidden induced subgraph & Certificate size & Reference \\
        \hline{}
        
        $P_{2k-1}$ & 0 & Direct\\
        $P_{2k+1}$ & $O(n)$ & Observation~\ref{obs:distance-minus-1}\\
        $P_{3k-1}$ & $O(n\log^3n)$ & Theorem~\ref{thm:3k-quasilinear}    \\
        $P_{\left \lceil \frac{14}{3} k \right \rceil - 1}$ & $O(n^{3/2}\log^2n)$ & Theorem~\ref{thm:P_143k_free_subquadratic}\\
        $P_{4k+3}$ & $\Omega(\frac{n}{k})$ & Theorem~\ref{thm:LB-paths}
        \end{tabular}
    \caption{Bounds for induced paths verification when vertices have a verification radius $k\geq 2$.}
    \label{table:paths-results}
    \vspace{-0.5cm}
\end{table}

%We will see later in this section that several of our results hold for more general subgraphs (induced and non-necessarily induced), but paths are good as a benchmark, and all our techniques were first designed for paths and then generalized. 

%Let us make a few observations on this table before we review the theorems and techniques.
A first observation on Table~\ref{table:paths-results} is that, when paths are long enough, we get a polynomial lower bound. In particular, compact certifications are unlikely to exist for forbidden \emph{induced} subgraphs (contrarily to subgraphs), except for specific cases (\emph{e.g.} $P_4$-free graphs where the structure of the graph class is known) or for cases where we forbid several subgraphs (see the discussion in Appendix~\ref{app:additional-related work}). 
%lindachange
%Intuitively, this is because less is known about the structural consequences of forbidding a graph $H$ as an induced subgraph than forbidding a graph as an subgraph or minor. %to tell us less about the \emph{global} structure of the graph, than forbidding a graph  and relatively few graph classes closed under taking induced subgraphs are known to have complete structural descriptions. 
%add some citation to robertson seymour graph minor theory, turan numbers etc. 

A second observation is that, for our upper bounds, the size increases with the ratio between the path length and the radius. We believe that it is not an artifact of our proof techniques but that certifying $P_\ell$-freeness indeed becomes more and more difficult when this ratio increases. Given these bounds and inspired by the complexities observed for subgraph detection in congested models (see \cite{Censor-Hillel22}), we make the following conjecture.

\begin{conjecture}
    For all $\alpha > 0$, the optimal size for the local certification of $P_{\alpha k}$-free graphs with verification radius $k$ is of the form $\Theta(n^{2-1/f(\alpha)})$, for some unbounded increasing function~$f$. 
\end{conjecture}

%Note that although our results do not allow to rule out a complexity of $O(n)$ for any path length, we do not believe this to be true.
%Also, our upper bounds are generally for shorter paths than the lower bounds, but we do have non-trivial polynomial upper and lower bounds for lengths between $4k+3$ and $\lceil(14/3)k\rceil -1$. 
If true, this conjecture would give us the first known \say{natural} set of properties that have optimal size certificates super-linear but sub-quadratic. %a regime in which we only know artificial properties. 

\paragraph*{Lower bound} 
Let us now discuss our theorems one by one, starting with our lower bound.

\begin{restatable}{theorem}{ThmLBPaths}
    \label{thm:LB-paths}
    For every graph $T$ which is either $P_{4k+3}$, or a tree of diameter at least $4k+2$ without degree $2$ vertices, at least $\Omega\left(\frac{n}{k}\right)$ bits are needed to certify $T$-induced free graphs when the verification radius is $k$.
\end{restatable}

While we designed our lower bound arguments to handle paths, we believe they can be adapted to all trees of diameter at least $4k+2$ with much more case analysis. To avoid excessive technicalities, we extend it only to trees without degree $2$ vertices and leave the general case as an open problem. 
Our technique to establish this lower bound is inspired by the now classic reduction from non-deterministic communication complexity~\cite{NisanKushilevitz}.
However, there are several challenges to overcome to adapt this technique to induced structures which will be discussed in Section~\ref{sec:lower-bounds}.

Note that when restricting to verification radius 1, this theorem implies that certifying $P_7$-free graphs requires $\Omega( n)$ bits. For contrast, $P_4$-free graphs can be certified in $\mathcal{O}(\log n)$ bits \cite{p4-montealegre2021compact}.
We leave bounds for the remaining path lengths open.%(for both non-trivial upper and lower bounds).

\paragraph*{Overview of our upper bound framework and a first application}

We prove a series of upper bound results, all using the same tool as a black-box, that we call a \emph{layered map}, which can be seen as a refinement of the so-called \emph{universal scheme} (that we remind below).
Our layered maps are not specific to one type of task, and we believe they can be used in various contexts beyond subgraph finding. 
We complete this general tool with involved tools specific to subgraph detection.

In the universal scheme, the certificate of every node contains the full map of the graph, and the vertices check that they have been given the same map as their neighbors, and that the map locally coincides with their neighborhoods. This takes $O(n\log n + \min(m \log n, n^2))$ bits (where $m$ is the number of edges), using adjacency matrix or adjacency lists. 
Now, to make use of the larger verification radius, we use an idea from~\cite{FeuilloleyFHPP21}, which is to spread the map.
More precisely: on correct instances, we cut the map into pieces and distribute them to the nodes of the graph in such a way that all the vertices can see all the pieces in their neighborhood. 
(One can prove that such a distribution exists via the probabilistic method.) Every vertex then simply reconstructs the map, and resumes the universal scheme. A larger verification radius allows for smaller pieces, which means smaller certificates.

This technique works well if the neighborhood at distance $k$ is large enough to drastically decrease the certificate size. This is the case for example when the minimum degree of the graph is polynomial, but in general this assumption is not met.
Our layered map is a relaxed version of this spread universal scheme, where the nodes are given different certificates depending on their degree, and are able to check ``partial maps'' of the graph. In other words, each node gets trustable information about some part of the graph. Note that partitioning graphs into low-degree / high-degree vertices has already been used for distributed subgraph detection, see~\cite{FischerGKO18, Censor-HillelFG20, ChangPZ19}, but the arguments are very different from ours. See discussion in Subsection~\ref{app:additional-related work}.

The generic framework we describe below is independent of the graph class and the considered problem. We believe it is the first time universal schemes are generalized that way, and that its interest goes way beyond subgraph detection since it can be easily used as a subroutine in future certification algorithms in the polynomial regime.

In this introduction, we describe the technique in the simplest setting with two groups of vertices, separating high and low degrees vertices with a threshold of $\sqrt{n}$, and we show how this can be used to prove the following theorem (which does not appear in Table~\ref{table:paths-results} since it is superseded by better results). 

\begin{restatable}{theorem}{thmPFourkFreeSubquadratic}
    We can certify $P_{4k-1}$-free graphs with verification radius~$k$ and certificates of size $O(n^{3/2} \log^2 n)$.
\label{thm:p_4k_free_subquadratic}
\end{restatable}

Intuitively, %with this trade-off between verification radius and certificate size, 
the high-degree vertices have enough neighbors to spread the whole map in their neighborhood (if we cut the map into parts of size $O(n^{3/2})$, we may give a few well-chosen parts to each node ensuring that each part appears on one of $\sqrt{n}$ neighbors of each high-degree node) while the subgraph restricted to the low degree vertices is sparse enough to be given to all vertices (without spreading). 

The difficulty is that the maps reconstructed by two high-degree vertices might not be the same, and if these vertices are separated by enough low-degree vertices, no vertex will be able to detect the inconsistency.  
This cannot be worked around even with certificates of size $o(n^2)$. Imagine for instance two very dense graphs separated by a long path: to be able to agree on the map of the graphs, one would intuitively need to move $\Theta(n^2)$ bits of information through the path, and this means $\Theta(n^2)$-bit certificates.

Nevertheless, since nodes can check the certificates at radius $k$, if two high-degree vertices are at distance at most $2k-2$, and the certification is accepted, then their maps must be identical (the vertex in the middle of any shortest path joining them being able to check the consistency). Therefore, we can define so-called \emph{extended connected components}, or \emph{ECC} for short, which partition the sets of high-degree vertices into groups that must share the same map (see Section~\ref{sub:ecci-dfns-app} for a more formal definition of ECC as well as a proof of its certification).  
Our first technical work is to propose a clean and generic framework to certify the ECCs, independently of the graph class and the considered problem. More precisely, we prove that we can certify exactly the list of the ECCs, and how they partition the vertex set, with linear-size certificates.\footnote{By certifying, we mean here that all the vertices, even low degree vertices, have access to this list.} Since high degree vertices can recover which edges lie within their own ECC, we may summarize the gathered information as follows: together with the set of ECCs, (i) low degree vertices have access to (and can certify) the set of edges incident to all the low-degree vertices and, (ii) high degree vertices have access to (and can certify) the same information plus the edges within their own ECC.

\medskip

We then build on this generic framework to certify the absence of induced paths. Let us sketch how to prove Theorem~\ref{thm:p_4k_free_subquadratic}, the most basic use of the framework. Suppose that $G$ contains a path $P$ on $4k-1$ vertices. Because different ECCs are at distance at least $2k-1$, such a path can touch at most three different ECCs.
In the case of zero or one ECCs, the inconsistency issue raised earlier does not appear, and either the maps given are incorrect (and detected as incorrect in the ECC building part above), or they are correct and some vertices of $P$ can see in their map that they belong to a long path. The case of three ECCs is very specific (each ECC contains exactly one vertex among the endpoints and the middle vertex of the path) and can be handled easily.
The case of two ECCs is more interesting.
Essentially, we prove that we can certify, for every low-degree vertex, the length of the longest path starting from it and going \emph{stictly} towards the closest ECC (and we provide this information to all the nodes, which takes $O(n \log n)$ bits in total). Along with other arguments, this ensures that at least one vertex will detect a long path touching two ECCs.
The full proof is deferred to Section~\ref{app:advanced-upper-bounds}.
%The main technicalities being there to ensure that the ``extension'' of the path towards the closest ECC does not interact with the rest of the path when we are looking for induced subgraphs. This can be done by defining more carefully this length and by choosing wisely which vertex will detect the path $P$. 

%\textcolor{purple}{$\rightarrow$ Would be good to streamline that part + don't we give the length of all vertices to all vertices?}\textcolor{red}{T. I think it is ok now?}\textcolor{red}{N. No we don't give the length all to all that would be too big.}\textcolor{blue}{S.: The information we give to all vertices is the following : for each vertex $u$, the length of the longest path $P_u$ starting at $u$ and then being strictly closer from the nearest ECC (and it costs only $O(n \log n)$)} \textcolor{red}{N. C'est ce que j'essayais de dire (sans definir proprement ECC), don't hesitate to update if you have a better formulation :) .}

\paragraph*{Advanced upper bounds}

Using layered maps and encoding additional information at the nodes for path detection as we did for Theorem~\ref{thm:p_4k_free_subquadratic} is a general canvas that we improve in several ways to get our three upper bound results. Namely, we generalize it to all graphs $H$ in Theorem~\ref{thm:H_4k_free_subquadratic} (with the same certificate size and radius), we reduce the certificate size to quasi-linear at the expense of shortening the path (or equivalently, increasing the radius) in Theorem~\ref{thm:3k-quasilinear} and finally, we certify longer paths, with same certificate size in Theorem~\ref{thm:P_143k_free_subquadratic}. We only very briefly describe the adaptations needed to prove them (with some imprecisions). The complete proofs of these results can be found in Section~\ref{app:advanced-upper-bounds}.

\begin{restatable}{theorem}{thmHFourkFreeSubquadratic}
\label{thm:H_4k_free_subquadratic}
    For every $k \ge 2$, we can certify $H$-free graphs with verification radius $k$ and certificates of size $O(n^{3/2} \log^2 n)$ for every $H$ which has at most $4k-1$ vertices.
\end{restatable}

The first part of the certification consists in constructing the ECCs as in Theorem~\ref{thm:p_4k_free_subquadratic}. But to detect a general graph $H$ if it exists, we need to store more information than for paths. Indeed, for paths, we simply stored the longest path starting from each vertex $v$ going strictly towards the closest ECC. That information basically allowed us to know on which position of a potential path $v$ can be. For subgraphs $H$, in order to determine if a vertex $v$ can be at some position of a copy of $H$, we need to store a table saying which subgraphs $H'$ of $H$ can contain $v$ and vertices only strictly towards the closest ECC (the description here is coarse and incomplete, see Section~\ref{app:advanced-upper-bounds} for more details).

\begin{restatable}{theorem}{thmThreekQuasilinear}
\label{thm:3k-quasilinear}
For every $k \ge 2$, we can certify $P_{3k-1}$-free graphs with verification radius $k$ and certificates of size $O(n \log^3 n)$.
\end{restatable}

In the proof of this statement, we use the full power of our ECC machinery by cutting the graph into a logarithmic number of types of degrees instead of $2$ (high degree / low degree). This more refined partition allows us to get almost linear certificates, by better adapting to the density of each part of the graph. The downside is that it makes the path detection much more involved, and reduces the length of the paths that can be detected. 

\begin{restatable}{theorem}{thmPImprovedPathsFreeSubquadratic}
\label{thm:P_143k_free_subquadratic}
    For every $k \ge 2$, we can certify $P_{\left \lceil \frac{14}{3} k \right \rceil - 1}$-free graphs with verification radius~$k$ and certificates of size $O(n^{3/2}\log^2 n)$.
\end{restatable}

The first part of the proof consists in constructing the ECCs as in Theorem~\ref{thm:p_4k_free_subquadratic}. But the information of the length of the longest path starting from a low-degree vertex and going strictly towards the closest ECC is not sufficient anymore. We prove that we can also certify the lengths of less constrained paths, and that this can be used to push a bit further the upper bound on the length of the paths detected. Essentially, we will be able to detect some type of paths that are not going strictly closer to the ECC, by storing at every vertex $v$, the length of four types of paths to every other vertex $w$, \emph{e.g.} one avoiding some neighborhood, one having its first step going strictly closer to the ECC etc. Note that since we now need some specific integers for each pair of vertices, to avoid a blow-up of the size of the certificates, we must restrict a node to know only the lengths that are relevant to it (whereas in the proof Theorem~\ref{thm:p_4k_free_subquadratic}, all vertices of the graph share the same table which contains, for every vertex, some information relevant to it).
%some special types of paths starting from it and going (non strictly) towards the closest ECC using certificates of size $O(n^{3/2})$ thanks to the structure of these paths and the fact that low-degree vertices have degree at most $\sqrt{n}$.\footnote{We will essentially allow paths to come back at most once not strictly closer to the ECC. The amount of information to remember will consist of a neighbor and non-neighbor of $v$ (which size at most $O(n^{3/2})$).} 
%However, now, since the size of this certificate is larger, it cannot be shared by all the vertices of the graph and can only be put in the certificate of the vertex itself. 
The structure of the certified paths plus the locality of this certification makes the proof much more technical.
%We prove that we can actually give an additional piece of information to the each low-degree node whose size is $O(n^{3/2})$ because of the structure of the graph which allows to cover more cases. Proving that we can certify this property, together with the fact that every path of length at least $\left \lceil \frac{14}{3} k \right \rceil - 1$ has a vertex with that property needs a careful analysis.

%\textcolor{purple}{$\rightarrow$ A bit too mysterious...} \textcolor{red}{N. J'ai édité, est ce que c'est mieux?}
%\linda{I think the we can also certify some \say{types} is too mysterious. I try to think of a better short description } \textcolor{red}{N. Thanks ! I agree that the redaction is not great here and I'm happy with any reformulation !}

\medskip

%\textcolor{red}{Finally note that the extensions of Theorem~\ref{thm:p_4k_free_subquadratic} described above are presented independently but they can be used together to improve further the length of the paths to certify, the types of graphs we want to certify or the certificate. For instance, we can use together the ideas of the proof of Theorems~\ref{thm:H_4k_free_subquadratic} and~\ref{thm:3k-quasilinear} to get quasilinear bounds for some graphs $H$. In order to keep the proof techniques as readable as possible and since combining the tools did not allow us to close gaps between lower and upper bounds we have decided to present them independently for readability. However, the interested reader can indeed combine them to obtain stronger results.}

A priori there are no obstacles to combining the three results above, to get \emph{e.g.} quasilinear bounds for more general subgraphs. But for the sake of readability, we explain the three extensions seperately.
The formal proofs can be found in Section~\ref{app:advanced-upper-bounds}. 

We note that our upper bounds theorems and technique also holds for non-necessarily induced graphs. This is not relevant for paths, since they are covered by Theorem~\ref{thm:non-induced-trees}, but it is relevant for Theorem~\ref{thm:H_4k_free_subquadratic} which applies to every $H$.

%%%%%%%%%%%%%%%%%%%%%%%%%%%%%%%%%%%%
\subsection{Additional related work and discussions}
\label{app:additional-related work}
%%%%%%%%%%%%%%%%%%%%%%%%%%%%%%%%%%%%%

\paragraph*{Adapting known results, in particular from the Broadcast congest model}

%\begin{note*}
%    Since the topic of certifying forbidden subgraphs is tightly linked to other areas and questions of distributed computing, there are a lot interesting discussions about \emph{e.g.} how previous work can (and cannot) be adapted to forbidden subgraphs certification (in particular works in \textsc{Broadcast Congest}). This is the focus of this subsection.
%Readers who are more interested in the new results can safely skip this section, and jump to Section~\ref{subsec:results-techniques}, simply reading the  statements in the Question and Observation environments, for context.  
%\end{note*}

%\paragraph*{Relationship between local certification and other distributed models}
There is a plethora of literature on subgraph detection in various distributed models. 
In general, results in these models are difficult to adapt to local certification.
Indeed, on the one hand, the non-constructive aspect of certificates makes the latter a stronger model, but on the other hand, local certification is a broadcast-type model, in the sense that all neighbors see the same information from a specific node, which implies that comparisons with unicast models such as \textsc{Congest} are rarely useful.\footnote{The topic of broadcast versus unicast in local certification has been explored in \cite{Patt-ShamirP22}.} 

The only model that is clearly useful for us is \textsc{Broadcast Congest} where, at each round, each node sends the same $O(\log n)$-bit message to all its neighbors.
In particular, as described in the observation below, an upper bound in \textsc{Broadcast Congest} implies an upper bound for local certification. 

\begin{observation}[Folklore]
    \label{obs:relationship-to-broadcast-congest}
    Let $\mathcal{P}$ be a fixed graph property.
    If there is a \textsc{Broadcast Congest}  algorithm deciding $\mathcal{P}$ running in $f(n)$ rounds, then there is a local certification for $\mathcal{P}$ using certificates of size $O(f(n)\log n )$. 
\end{observation}

Here, by \emph{deciding a property}, we mean the same decision mechanism as in local certification: the graph is accepted if and only if all nodes accept. This works well with subgraph detection in \textsc{Broadcast Congest}, where, if the graph contains one or several copies of the subgraph, then at least one node will detect at least one copy.

\begin{proof}
Suppose that there is an upper bound of $f(n)$ rounds in this model for a deciding a property $\mathcal{P}$. Then we can derive an upper bound of $O(f(n)\log n )$ for certification, by encoding in the certificate of each node all the messages that it sends during the run of the algorithm.
Indeed, given this information, every node can check that the run is correct, and if so a node rejects this certification in the end, if and only if, the output of the original algorithm is a reject. 
\end{proof}

Let us use this to get the upper bound of the theorem of the introduction that we restate here. 

\ThmNonInducedTrees*

It is proved in~\cite{KorhonenR17, EvenFFGLMMOORT17} is that any tree on $k$ vertices can be detected in $O(k2^k)$ rounds, which directly translates to a $O(\log n)$ upper bound for certifying that the graph does not contain a tree $T$ as a (non-necessarily induced) subgraph. 

The matching lower bound can be proved by adapting the treedepth lower bound in~\cite{FeuilloleyBP22}. In the proof mentioned, the core result is that it requires $\Omega(\log n)$ bits to accept graphs of type (1) and reject those of type (2), where (1) is a collection of cycles of length at most 8, plus a universal vertex, and (2) is a collection of cycles one of which has size at least 16, and the others at least 8, plus a universal vertex. Note that graphs of type (1) are $P_{20}$ free while the ones of the type (2) do contain a $P_{20}$, which proves the bounds for such paths. And the same proof can be adapted to any larger path/tree.

For lower bounds, results in \textsc{Broadcast Congest} cannot be adapted directly, but the techniques can sometimes be borrowed. For example, it is proved in~\cite{CrescenziFP19}, that the lower bound from~\cite{DruckerKO13} for the broadcast congested clique can be adapted to get an $\Omega(n/e^{O(\sqrt{\log n})})$ bound for the certification of triangle-free graphs. Actually, we can get from another theorem of~\cite{DruckerKO13}, that $\Omega(n)$ bits are necessary for $K_k$-free when $k>3$. This gives the lower bound of the following theorem cited in the introduction. 

\ThmKkFree*

The upper bound is proved in~\cite{BousquetEFZ24} using a renaming theorem. Namely: at a cost of an additional $O(\log n)$ bits in the certificates, we can assume that the identifiers are in the range $[1,n]$ instead of $[1,n^c]$. Given that, it is sufficient for the prover to give to each vertex $v$ a bit vector $T[v]$ of size $n$, where $T[v]_i$ is equal to $1$ if and only if $v$ is a neighbor of the vertex of (new) identifier~$i$. With this information in the certificates, every vertex knows the adjacency of its neighbors, and can thus determine the size of the largest clique to which it belongs, and rejects if it is at least~$k$.

Let us mention a last result that uses the renaming technique of~\cite{BousquetEFZ24}, and which generalizes the upper bound technique of the previous Theorem~\ref{thm:Kk-free}.

\begin{observation}
    \label{obs:distance-minus-1}
    Consider a graph $H$ with a central vertex $w$ such that all vertices of $H$ are at distance at most $d$ from $w$. Then we can certify that a graph is $H$-free with a verification radius $d$ and certificates of size $O(n)$. 
\end{observation}

Note that with verification radius $d+1$, one does not need certificates, since the whole subgraph would be in the view of the copy of $w$ in the graph.  

The proof consists again in using the identifiers from range $[1,n]$ given by~\cite{BousquetEFZ24}, and encoding the neighborhood of a vertex~$v$ in its certificate using the same bit vector~$T[v]$. With this information, every node~$v$ knows all the edges incident to the vertices at distance~$d$ from itself (whereas without certificates, $v$ does not see the edges between two vertices at distance exactly~$d$). Hence, if $H$ is present, the copy of $w$ in $G$ will detect it and reject.

\paragraph*{Certification size landscape}
A fruitful line of work in the LOCAL model consists in establishing the landscape of complexities for the classic family of problems called locally checkable languages, LCLs for short, see \emph{e.g.}~\cite{Suomela20}. 
In this perspective, one aims at characterizing the functions $f(n)$ for which there exists a problem whose optimal complexity is $f(n)$.
Developing the same kind of theory for local certification size is an exciting research direction. 
So far, the only sizes for which we have natural problems with tight bounds are: $O(1)$, $\Theta(\log n)$, $\tilde{\Theta}(n)$ and $\tilde{\Theta}(m)$, where $\tilde{\Theta}$ means up to subpolynomial factors. (One can actually build a problem for any size $f(n,m)$ above $\log n$, but it is artificial in the sense that the definition of the problem refers to the function $f$.)

The constant certification size regime contains local properties such as coloring (intuitively the LCLs), and the dependency in other parameters such as the maximum degree has been explored very recently~\cite{BousquetFZ24}.  
The size $\Theta(\log n)$ is very common, with the archetypal problem being acylicity. See for example the list of problems in~\cite{GoosSuomela2016}, and the recent series of papers establishing meta-theorems for this regime~\cite{FeuilloleyBP22, FraigniaudMRT22, FraigniaudM0RT23}. 
Actually, in some of these papers, the upper bounds are polylogs and not logs, and they are not matching the lower bounds, leaving open whether some of these properties have optimal certificate size $\Theta(\log^{c}n)$, with $c>1$, or not.\footnote{Minimum spanning trees with polynomial weights  have optimal certificate size $\Theta(\log^2 n)$~\cite{KormanKP10}, but one of the $\log n$'s originates from the encoding of the weights, hence this is not very satisfactory.}
Finally, in the polynomial regime, we know problems with complexity $\tilde{\Theta}(n)$ (\emph{e.g.} diameter $\leq 3$~\cite{Censor-HillelPP20}) and problems with complexity $\tilde{\Theta}(m)$ (\emph{e.g.} symmetric graphs~\cite{GoosSuomela2016} and non-3-colorable graphs~\cite{GoosSuomela2016}). 
Our paper explores the polynomial regime, but unfortunately since the bounds do not match we cannot conclude for sizes between $n$ and $m$. Nevertheless, we believe that the case of induced paths of various sizes is a promising direction in this landscape perspective. 

\paragraph*{Certification of graph classes and forbidden subgraphs}

A recent trend in local certification has been to focus on certifying graph properties (\emph{e.g.} planarity) instead of certifying the output of an algorithm (\emph{e.g.} that a set of pointers distributed on the nodes collectively form a spanning tree). Two motivations behind this focus are that the optimal certification size can be seen as a measure the locality of a graph property (raising interest from the graph theory community) and that algorithms tailored to work on specific graph classes make more sense if one can ensure that the graph indeed belongs to the class.  

Among the classes studied, many are defined by forbidden minors, as discussed previously. But some of these classes are not closed under minors, and are better described by families of forbidden induced subgraphs. 
For example, the authors of \cite{JaureguiMRR23} tackle the case of chordal graphs (that are characterized by forbidding all induced cycles of length at least 4) and several other classes whose subgraph characterizations are more cumbersome: interval graphs, circular arc graphs, trapezoid graphs and permutation graphs. For all these classes, \cite{JaureguiMRR23} establishes a $O(\log n)$-bit certification. 
This might come as a surprise, since in this paper all the certifications are in the polynomial regime. The reason for this contrast is that all the classes we have just mentioned are very structured, and in particular have geometric representations, which can be used for certification. 
Very recently, it was proved that some classes are hard to certify even though they do have geometric representation: 1-planar graphs, unit disk graphs and other related classes require (quasi)-linear-in-$n$ certificates~\cite{DefrainELMR24}.

\paragraph*{Techniques based on bucketing by degree}

Our layered maps are based on bucketing the vertices by degree, which is a classic strategy in distributed subgraph detection.
A classic canvas for detection algorithm is to first process the high-degree nodes, either arguing that there are few of them (see \emph{e.g.}~\cite{FischerGKO18}) or that many nodes are close to high degree nodes (see \emph{e.g.}~\cite{Censor-HillelFG20}), and then to process low-degree nodes, often using color coding (a large part of these algorithms is randomized). 
Another technique consists in computing an expander (or conductance) decomposition, to separate the graph into parts that are well-connected (where one can basically compute as if the cluster would be a clique) and parts of low degree (see \emph{e.g.}~\cite{ChangPZ19, EdenFFKO22}). 

Our understanding is that only general intuitions can be transferred from these algorithms to our setting, namely:  inside a cluster one should use the fact that it is easy to move/distribute information, whereas in low-degree parts one should enjoy the fact that there are fewer edges, hence fewer information to spread.
%%%%%%%%%%%%%%%%%%%%%%%%%%%%%%%%
\section{Model and definitions}
\label{sec:model}
%%%%%%%%%%%%%%%%%%%%%%%%%%%%%%%%

\subsection{Graph theory notions}

In this paper, the network is modeled by an undirected graph without loops or parallel edges. The number of nodes is denoted by $n$ and the number of edges is denoted by $m$. A graph is \emph{$d$-regular} if all its vertices have degree $d$. A graph is \emph{regular} if it is $d$-regular for some $d$.

We call $H$ an \emph{induced subgraph} of $G$ if $H$ is obtained from $G$ by deleting a subset of the vertices of $G$ and the edges incident to them.  We call a graph $H$-free if it does not contain $H$ as an induced subgraph.
We call $H$ a \emph{subgraph} (sometimes specified as \emph{non-necessarily induced}) of $G$ if $H$ is obtained from $G$ by deleting a subset of the vertices of $G$, the edges incident to them, and an arbitrary subset of edges of $G$.  
We let $P_k, C_k$ denote the path and cycle on $k$ vertices respectively. The length of a path or cycle is the number of its edges.

We say two disjoint sets of vertices $X, Y$ are \emph{complete} to each other if all possible edges between $X$ and $Y$ are present. If there are no edges between $X$ and $Y$ we say they are \emph{anticomplete}. 
A set of edges $M$ of a bipartite graph $G$ is a perfect matching if all the vertices of the graph are adjacent to exactly one edge of $M$. A set of edges in a bipartite graph is an \emph{antimatching}, if they form a complete bipartite graph without a perfect matching. 

\subsection{Local certification}
\label{subsec:model-certification}

In the networks we consider, the vertices are equipped with unique identifiers on $O(\log n)$ bits. Certificates are labels attached to the vertices. The view at distance $d$ of a vertex $v$ consists of: (1) the vertices at distance at most $d$ from $v$, (2) the identifiers and certificates of these vertices, and (3) the edges between these vertices, except the ones between two vertices at distance exactly $d$.

\begin{definition}
    We say that there exists a \emph{local certification} of size $s$ with verification radius $d$ for a property $P$ if there exists a local algorithm (called the \emph{verification algorithm}) taking as input on every node $v$ the view at distance $d$, and outputting accept/reject such that:
\begin{itemize}
    \item For every $n$-vertex graph that satisfies the property $P$, there exists a certificate assignment, with certificates of size at most $s(n)$, such that the verification algorithm accepts at every node. 
    \item For every $n$-vertex graph that does not satisfy the property $P$, for all certificate assignments, there exists at least one node where the verification algorithm rejects.
\end{itemize}
\end{definition}

In order to facilitate the writing, we say that the certificates are given by a \emph{prover}. 
We can specify how the certificates are assigned by the prover on correct instances (\emph{i.e.} graphs satisfying the property $P$), but we cannot control what happens on incorrect instances. 

As an example, let us describe a local certification with verification radius~$1$ and certificate size $O(\log n)$, for checking that the graph is acyclic. 
On a correct instance, the prover chooses a node to be the root, and gives as a certificate to every node its distance to the root. 
The verification algorithm checks that the distances are consistent (typically that one neighbor has been assigned a strictly smaller distance, and the others a strictly larger distance). Now if the graph has a cycle, for any certificate assignment, the vertex with the largest assigned distance in the cycle has at least two vertices with distance smaller than or equal to its distance, hence it rejects.

%%%%%%%%%%%%%%%%%%%%%%%%%%%%%%%%%%%%%%%%%%%%%%%%%%%
\section{Lower bounds for paths (and extensions)}
\label{sec:lower-bounds}
%%%%%%%%%%%%%%%%%%%%%%%%%%%%%%%%%%%%%%%%%%%%%%%%%%%

Our lower bound result is the following.

\ThmLBPaths*

In this section we give an overview of the proof techniques to design our lower bounds as well as the construction we use. 

\paragraph*{Discussion of the challenges of the proof, based on the communication complexity intuition}
At first glance, the proof of Theorem~\ref{thm:LB-paths} is a reduction from the problem of non-disjointness in  non-deterministic communication complexity, which is now a classic tool in the area (see for example \cite{Censor-HillelPP20, FeuilloleyFHPP21}). 
Let us quickly remind ourselves what such reductions look like, to be able to discuss the challenges. Readers who do not know this technique can safely skip this discussion, since in our formal proof we do not use communication complexity. 

In such a reduction, we use two graphs $G_A$ and $G_B$ having order of $n$ vertices, encoding the inputs $A$ and $B$ of Alice and Bob, and some ``middle part'' connecting them, in such a way that (1) our graph has the desired property if and only if $A \cap B = \emptyset$, and (2) the middle part is such that the nodes in $G_A$ cannot see $G_B$, and conversely.
Intuitively, since set disjointness for sets of size~$t$ requires Alice and Bob to communicate order of $t$ bits, and only the certificates can be used to transfer this information, any local certification will require certificates of order $t/c$, where $c$ is the size of the minimum cut in the middle part. Here, $t$ will be of order~$n^2$ (because $A$ and $B$ will be sets of edges of some $n$-vertices bipartite graphs), and $c$ will be of order $kn$, leading to the $\Omega\left(\frac{n}{k}\right)$ lower bound.

Adapting this framework to forbidden induced subgraphs certification is quite challenging.
In classic proofs of the same flavor, one can go from checkability radius 1 to checkability radius~$k$ in a very simple way: every edge in the middle part is replaced by a path of length~$k$, to separate the nodes controlled by Alice and Bob even further. This is not good for us, since we want to control very precisely the path structure of the graph.
Also, we need to care about non-edges everywhere in the construction and not only in the part ``encoding the disjointness'', because we care about induced subgraphs, and not subgraphs. (Remember that we know from Theorem~\ref{thm:non-induced-trees} that a high lower bound cannot hold for tree subgraphs.) 
To overcome these difficulties we use a construction that blends cliques, matchings and antimatchings, and whose analysis requires a lot of care to rule out the existence of unwanted paths in our construction.

\subsection{Our lower bound construction}

As said earlier, in our proofs we do not make the communication complexity appear explicitly, instead we use a simple counting argument (similar to fooling sets in the communication complexity proof, and to the arguments used~\cite{GoosSuomela2016}). This is more self-contained, and in our opinion, more elegant, and also avoids formal definition of non-deterministic communication complexity. 

Let $n \geq 1$ be an integer, and assume that the elements of $A,B$ are subsets of $[1,n]$ of size $2$. Let $G_A$ be the bipartite graph with two parts of $n$ vertices, labelled $\{1,\ldots,n\}$ and $\{1',\ldots,n'\}$ 
and the following edges:
\begin{itemize}
\item For each $\{i,j\} \notin A$, $ij', ji' \in E(G_A)$, and
\item for each $i \in \{1,2,3, \dots, n\}$, $ii' \in E(G_A)$.
\end{itemize}
We define $G_B$ similarly from $B$.

We now connect $G_A$ and $G_B$, by adding paths on $2k$ vertices from vertices labeled with $[1,n]$ (resp. $[1',n']$), called \emph{top vertices} (resp. \emph{bottom vertices}) in $G_A$ and $G_B$. This is done via adding a collection of cliques linked by matchings and antimatchings. This is the key feature that distinguishes our proof technique from those of pre-existing lower bounds for substructure detection.

  \begin{figure}[!h]
        \centering
        \scalebox{0.85}{
        \input{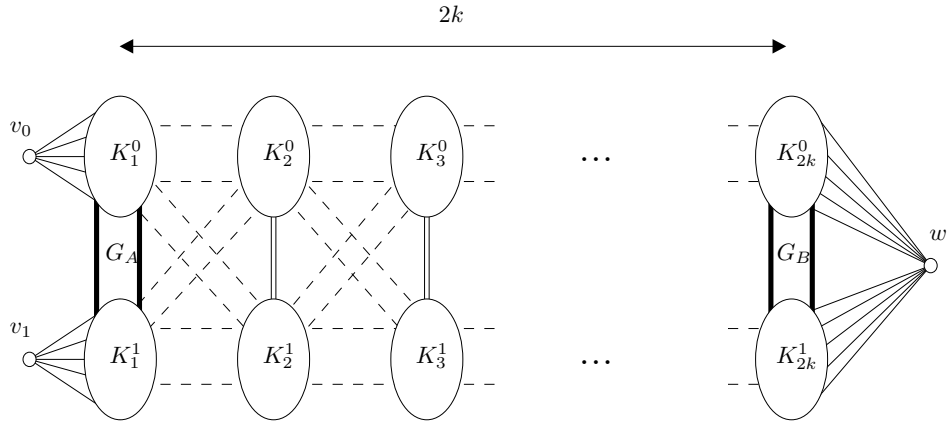}
        }
        \caption{The graph $G_{k,n}(A,B)$. Each blob is a clique on $n$ vertices. The double thick edges represent the graphs $G_A$ and $G_B$. The other double edges represent the matchings. The double dashed edges represent the antimatchings.}
        \label{fig:GHH'}
    \end{figure}

Let $T$ be path on at least $4k+3$ vertices or a tree of diameter at least $4k+2$ without degree two vertices. We construct a graph $G_{k,n}(A,B)$ on $4k(n-1)+|T|$ vertices, which will be $T$-free if and only if $G_A$ and $G_B$ do not have a common non-edge (and so if and only if $A$ and $B$ are disjoint). The construction is illustrated on Figure~\ref{fig:GHH'}.
    
We construct $4k$ cliques of size $n$, denoted by $K^0_1, \ldots K^0_{2k}, K^1_1, \ldots K^1_{2k}$. In each clique, we number the vertices from $1$ to $n$.\footnote{We leave the precise unique identifier assignment implicit, for simplicity, but it is easy to create one, having the vertices of the first clique named 1 to $n$, then the second $n+1$ to $2n+1$, etc.} For every clique $K$ and $i \in \{1, \ldots, n\}$, we denote by $K[i]$ the vertex numbered by $i$ in $K$. We add the following edges in $G_{k,n}(A,B)$:
    \begin{itemize}
        \item We put the edges of $G_A$ between $K^0_1$ and $K^1_1$, and of $G_B$ between $K^0_{2k}$ and $K^1_{2k}$.\footnote{That is, we identify the vertices of each side of the bipartite graph with the vertices of each blob, \emph{e.g.} in increasing order of identifiers.}
        \item For every $j \in \{2, \ldots, 2k-1\}$, we put the perfect matching $\{(K^0_j[i],K^1_j[i])\}_{1 \leqslant i \leqslant n}$.
        \item For every $j \in \{1, \ldots, 2k-1\}$ and for every pair of cliques $K, K'$ with $K \in \{K^0_j,K^1_j\}$, $K' \in \{K^0_{j+1},K^1_{j+1}\}$, we put the antimatching defined by the complement of $\{(K[i],K'[i])\}_{1 \leqslant i \leqslant n}$ between $K$ and $K'$ (that is, all the edges but the matching $\{(K[i],K'[i])\}_{1 \leqslant i \leqslant n}$).
    \end{itemize}
    
Finally, if $T=P_{4k+3}$, we add 3 more vertices $v_0,v_1,w$ in the following way. The vertex $v_0$ (resp. $v_1$) is complete to $K^0_1$ (resp. $K_1^1$). The vertex $w$ is complete to $K^0_{2k}$ and $K^1_{2k}$. Otherwise, we choose a path $P$ of length $4k$ in $T$ containing no leaf, label its vertices by $v_1^0,\ldots, v_{2k}^0,w,v_{2k}^1,\ldots,v_1^1$, and for each connected component $C$ of $T\setminus P$ adjacent with vertex $v_a^b$, we create a copy of $C$ in $G_{k,n}(A,B)$ where the vertex adjacent with $v_a^b$ becomes complete to $K_a^b$. Vertices in this copy are called \emph{pending vertices} of $K_a^b$. We also add a copy of the connected component of $w$ in $T\setminus E(P)$, and add all edges between $w$ and $K_{2k}^0\cup K_{2k}^1$. 

Note that applying the latter construction to $P_{4k+3}$ yields the former graph, hence we may use $G_{k,n}(A,B)$ in both cases.

We state the main property satisfied by this construction in \cref{prop:common non-edge} below. We prove \cref{prop:common non-edge} by looking carefully at the structure of $G_{k,n}(A,B)$ 
%%%%%%%%%%%%% SHORT VERSION vs FULL VERSION %%%%%%%%%%
%in \cref{app:lower-bounds}.
%
at the end of this section.

    \begin{proposition}
        \label{prop:common non-edge}
        The graph $G_{k,n}(A,B)$ is $T$-free if and only if $A$ and $B$ are disjoint.
        %Added for clarification, since obviously they have common non-edges as they are both bipartite
    \end{proposition}

  We prove \cref{prop:common non-edge} by carefully analyzing the structure of $G_{k,n}(A,B)$.
  Since its proof is quite long and technical, we postpone it to \cref{sec:common_non_edge} and we conclude this section by proving Theorem~\ref{thm:LB-paths} using Proposition~\ref{prop:common non-edge}. 

  \begin{proof}[Proof of \cref{thm:LB-paths}]
  Assume that $m$ bits are sufficient to certify $T$-free graphs with verification radius~$k$.
    For any $A,B$, let $V$ denote the vertex set of $G_{k,n}(A,B)$. 

    We first define a family of positive instances. 
    Let $\overline{H}$ denote the complement of $H$. By Proposition~\ref{prop:common non-edge}, for every set $A$, the graph $G_{k,n}(A,\overline{A})$ is $T$-free.
    %where $\overline{H}$ is the bipartite-complement of $H$ (that is, the bipartite obtained from $H$ by replacing each plus the edges $\{(i,i)\}_{1\leqslant i\leqslant n}$. 
    %LC: found it slightly confusing to use \overline and call it the complement since it is not the complement graph in the graph theoretic sense
    Therefore, by hypothesis, there exists a certificate function $c_A : V \rightarrow \{0, \ldots, 2^m-1\}$ such that all the vertices in $G_{k,n}(A,\overline{A})$ accept. 
    \begin{claim}\label{4k-distinct}
       For any two distinct sets $A,B$, $c_A\neq c_B$.
    \end{claim}
    \begin{subproof}
    Suppose that there are two different sets $A$ and $B$ such that $c_A = c_B$.
    Since $A\neq B$, we have either %$H \nsubseteq H'$ or $\overline{H} \nsubseteq \overline{H'}$ or both (where the inclusion is on the edge set).
    $E(G_A) \nsubseteq E(G_B)$ or $E(G_{\bar{A}}) \nsubseteq E(G_{\bar{B}})$ or both.
    By symmetry, assume that $E(G_A) \nsubseteq E(G_B)$. Then, the graph $G_{k,n}(B,\overline{A})$ is accepted with $c_A$. 
    Indeed, the vertices in the cliques $K^i_j$ for $i \in \{0,1\}$ and $j \in \{1, \ldots, k\}$, or pending from these cliques have the same view as their view in $G_{k,n}(B,\overline{B})$ since the only part of the graph which has changed between $G_{k,n}(B,\overline{A})$ and $G_{k,n}(B, \overline{B})$ is the bipartite graph between $K^0_{2k}$ and $K^1_{2k}$, which is at distance at least $k$ from them. Thus, these vertices accept in $G_{k,n}(B,\overline{A})$. Similarly, $w$ and the vertices in the cliques $K^i_j$ for $i \in \{0,1\}$ and $j \in \{k+1, \ldots, 2k\}$ or pending from them have the same view in $G_{k,n}(B,\overline{A})$ and $G_{k,n}(A,\overline{A})$, so they accept in the graph $G_{k,n}(B,\overline{A})$ as well. 
    However, since $A\neq B$, the graphs $G_B$ and $G_{\overline{A}}$ have a common non-edge, hence $G_{k,n}(B,\overline{A})$ is not $T$-free (by Proposition~\ref{prop:common non-edge}), which is a contradiction.
    \end{subproof}
    
    %Hence, all the certificate functions $c_H$ are different, for the bipartite graphs $H$ satisfying the two conditions above. 
    There are $2^{\frac{n(n-1)}{2}}$ such sets $A,B$, and there are at most $2^{m(4kn+|T|)}$ functions $V \rightarrow \{0, \ldots, 2^m-1\}$. 
    Thus, by \cref{4k-distinct}, we get $2^{\frac{n(n-1)}{2}} \leqslant 2^{m(4kn+|T|)}$, which finally gives us $m = \Omega(\frac{n}{k})$.
    \end{proof}

\subsection{Proof of Proposition~\ref{prop:common non-edge}}
\label{sec:common_non_edge}

    Let us now turn to the proof of Proposition~\ref{prop:common non-edge}. 

    Recall that for each set $A$, the graph $G_A$ satisfies the two following properties:
    \begin{enumerate}[(i)]
        \item For all $(i,j) \in \{1, \ldots, n\}^2$, $ij' \in E(G_A)$ if and only if $ji'\in E(G_A)$ \quad ($G_A$ is symmetric).
        \item For all $i \in \{1, \ldots, n\}$, $ii'\in E(G_A)$ \quad ($G_A$ is reflexive).
    \end{enumerate}

    We will need the following result about the shape of long paths in $G_{k,n}(A,B)$.
    For succinctness, we will call $K^0_1, \ldots K^0_{2k}, K^1_1, \ldots K^1_{2k}$ the \emph{cliques (of $G_{k,n}(A,B)$)} throughout this section. 
        \begin{lemma}
            \label{lem:<4k vertices in the cliques}
            Let $T$ be an induced tree in $G_{k,n}(A,B)$. Then, $T$ has at most $4k$ vertices in the cliques.
        \end{lemma}
        
        \begin{proof}[Proof of Lemma~\ref{lem:<4k vertices in the cliques}.]
            % Trivially $P$ has at most two vertices in each cliques.
            The main ingredient of the proof of Lemma~\ref{lem:<4k vertices in the cliques} is the following Claim~\ref{claim:properties of P}. 
            In the following, two cliques are said to be \emph{antimatched} if there is an antimatching between them (that is, between the first and the second clique, there is a complete bipartite graph except for $n$ independent non-edges).
            \begin{claim}
            \label{claim:properties of P}
            $T$ has at most two vertices in each clique and
            \begin{enumerate}[(i)]
                \item If $T$ has two vertices in a clique $K$, then for all the cliques $K'$ antimatched with $K$, $T$ has at most one vertex in $K'$, and 
                %\marginpar{\linda{we probably should define antimatched at its first occurance}} 
                \label{pathclaim:antimatched}
                \item If $T$ has two vertices in a clique $K^i_j$, then $T$ has at most one vertex in $K^{1-i}_j$.\label{pathclaim:before}
            \end{enumerate}
            \end{claim}
            \begin{proof}
                Since $T$ is triangle-free, it contains at most two vertices in each clique.
                Suppose there is a clique $K:=K^i_j$ and two distinct integers $x,y \in \{1,2, \dots, n\}$ such that the vertices $K[x]$ and $K[y]$ are both in $T$.
                We prove \ref{pathclaim:antimatched} and \ref{pathclaim:before} separately:                
               \begin{itemize}
               % \item For (i), it is straightforward.
                \item For (\ref{pathclaim:antimatched}), %let $K$ be a clique in which $P$ has two vertices. Let $K[x],K[y]$ be the two vertices in $P \cap K$.
                let $K'$ be a clique antimatched to $K$.
                For each $z \in \{1,2, \dots, n\}$, $K'[z]$ is adjacent to $K[x]$ (if $z \neq x$) or to $K[y]$ (if $z \neq y$).
                Thus, if there were at least two vertices in $T \cap K'$, $T$ would have either a $C_3$ or a $C_4$ as a subgraph, which is a contradiction with $T$ being a tree.
                \item We now prove property (\ref{pathclaim:before}).
                Let $K=K^i_j$ and $K'=K^{1-i}_j$. 
                Assume for a contradiction that there exist two distinct $s,t  \in \{1,2, \dots, n\}$ such that $K'[s]$ and $K'[t]$ are both in $T$. Note that since $T$ is triangle-free, if $j=1$ then $T$ cannot contain $v_0$ nor $v_1$, and if $j=2k$, then $T$ cannot contain $w$. Therefore, since $T$ has at least $5$ vertices,  
                %Let $C$ denote the component containing $K' \cup K$ after deleting all the vertices of $K^0_{j-1}, K^1_{j-1}, K^0_{j+1}, K^1_{j+1}$.
                %Then $V(C) = K' \cup K, K' \cup K \cup \{u_0, v_0, u_1, v_1\}$ (if $j = 1$)  or $K \cup K' \cup \{w\}$ (if $j = 2k$).
                %So in particular $C$ can contain at most TODO vertices of $T$.
               % contradiction that there exist $i,j$ such that $P$ has at least two vertices in both $K:=K^i_j$ and $K':=K^{1-i}_j$.
                %Since $T$ has $p \geqslant TODO$ vertices,                
                it follows that there is some $K'' \in \{K^0_{j-1}, K^1_{j-1}, K^0_{j+1}, K^1_{j+1}\}$ such that $T$ contains $K''[z]$ for some $z \in \{1,2, \dots, n\}$.
                %Let $K[x],K[y]$ (resp. $K'[x'],K'[y']$) denote the two vertices in $P \cap K$ (resp. $P \cap K'$). 
               % Let $K''[z]$ denote the vertex in $P \cap K''$.
                We have $x = z$ or $y = z$ otherwise, $K[x],K[y],K''[z]$ would be a triangle in $T$. Similarly, we have $s = z$ or $t = z$.
                Without loss of generality, we may assume that $x =s = z$.
                But then $K[z] \dd K[y] \dd K''[z] \dd K'[t] \dd K'[z] \dd K[z]$ is a $C_5$ subgraph with vertex set in $V(P)$, see Figure~\ref{fig:(iii)} for an illustration. (Note that the edge $K[z]\dd K'[z]$ arises even when $j\in\{1,2k\}$ since $H$ and $H'$ are reflexive.)
                This is a contradiction with the fact that $T$ is an induced tree. \qedhere
            \end{itemize}
            \begin{figure}[!ht]
                    \centering
                    \scalebox{0.8}{
                    \begin{tikzpicture}[x=0.75pt,y=0.75pt,yscale=-1,xscale=1]
%uncomment if require: \path (0,460); %set diagram left start at 0, and has height of 460

%Straight Lines [id:da008245061354414829] 
\draw [color={rgb, 255:red, 155; green, 155; blue, 155 }  ,draw opacity=1 ] [dash pattern={on 4.5pt off 4.5pt}]  (228,231) -- (319,132) ;
%Straight Lines [id:da4776395844641276] 
\draw [color={rgb, 255:red, 155; green, 155; blue, 155 }  ,draw opacity=1 ] [dash pattern={on 4.5pt off 4.5pt}]  (229,253) -- (318,159) ;
%Straight Lines [id:da6279203021599922] 
\draw [color={rgb, 255:red, 155; green, 155; blue, 155 }  ,draw opacity=1 ] [dash pattern={on 4.5pt off 4.5pt}]  (222,149) -- (321,149) ;
%Straight Lines [id:da5085177060690573] 
\draw [color={rgb, 255:red, 155; green, 155; blue, 155 }  ,draw opacity=1 ] [dash pattern={on 4.5pt off 4.5pt}]  (222,116) -- (321,116) ;
%Shape: Ellipse [id:dp8949099487853657] 
\draw  [fill={rgb, 255:red, 255; green, 255; blue, 255 }  ,fill opacity=1 ] (200,134.25) .. controls (200,114.51) and (210.52,98.5) .. (223.5,98.5) .. controls (236.48,98.5) and (247,114.51) .. (247,134.25) .. controls (247,153.99) and (236.48,170) .. (223.5,170) .. controls (210.52,170) and (200,153.99) .. (200,134.25) -- cycle ;
%Shape: Ellipse [id:dp6797793509035407] 
\draw  [fill={rgb, 255:red, 255; green, 255; blue, 255 }  ,fill opacity=1 ] (200,254.25) .. controls (200,234.51) and (210.52,218.5) .. (223.5,218.5) .. controls (236.48,218.5) and (247,234.51) .. (247,254.25) .. controls (247,273.99) and (236.48,290) .. (223.5,290) .. controls (210.52,290) and (200,273.99) .. (200,254.25) -- cycle ;
%Shape: Ellipse [id:dp5520817055617177] 
\draw  [fill={rgb, 255:red, 255; green, 255; blue, 255 }  ,fill opacity=1 ] (300,134.25) .. controls (300,114.51) and (310.52,98.5) .. (323.5,98.5) .. controls (336.48,98.5) and (347,114.51) .. (347,134.25) .. controls (347,153.99) and (336.48,170) .. (323.5,170) .. controls (310.52,170) and (300,153.99) .. (300,134.25) -- cycle ;
%Shape: Circle [id:dp4479748104409951] 
\draw  [color={rgb, 255:red, 208; green, 2; blue, 27 }  ,draw opacity=1 ][fill={rgb, 255:red, 208; green, 2; blue, 27 }  ,fill opacity=1 ] (217.75,120.25) .. controls (217.75,117.9) and (219.65,116) .. (222,116) .. controls (224.35,116) and (226.25,117.9) .. (226.25,120.25) .. controls (226.25,122.6) and (224.35,124.5) .. (222,124.5) .. controls (219.65,124.5) and (217.75,122.6) .. (217.75,120.25) -- cycle ;
%Straight Lines [id:da25413611608081643] 
\draw [color={rgb, 255:red, 155; green, 155; blue, 155 }  ,draw opacity=1 ]   (225,170) -- (225,218.5)(222,170) -- (222,218.5) ;
%Straight Lines [id:da3219689013632552] 
\draw [color={rgb, 255:red, 208; green, 2; blue, 27 }  ,draw opacity=1 ]   (222,120.25) -- (323.5,130) ;
%Straight Lines [id:da758322079946804] 
\draw [color={rgb, 255:red, 208; green, 2; blue, 27 }  ,draw opacity=1 ]   (223.5,268.25) -- (323.5,130) ;
%Straight Lines [id:da17871513576629738] 
\draw [color={rgb, 255:red, 208; green, 2; blue, 27 }  ,draw opacity=1 ]   (222,120.25) -- (222,149) ;
%Curve Lines [id:da9856959620943758] 
\draw [color={rgb, 255:red, 208; green, 2; blue, 27 }  ,draw opacity=1 ]   (222.5,240) .. controls (203.5,215) and (205.5,175) .. (222,149) ;
%Straight Lines [id:da939020383069048] 
\draw [color={rgb, 255:red, 208; green, 2; blue, 27 }  ,draw opacity=1 ]   (223.5,268.25) -- (222.5,240) ;
%Shape: Circle [id:dp21679357873076244] 
\draw  [color={rgb, 255:red, 208; green, 2; blue, 27 }  ,draw opacity=1 ][fill={rgb, 255:red, 208; green, 2; blue, 27 }  ,fill opacity=1 ] (217.75,149) .. controls (217.75,146.65) and (219.65,144.75) .. (222,144.75) .. controls (224.35,144.75) and (226.25,146.65) .. (226.25,149) .. controls (226.25,151.35) and (224.35,153.25) .. (222,153.25) .. controls (219.65,153.25) and (217.75,151.35) .. (217.75,149) -- cycle ;
%Shape: Circle [id:dp4575543116857793] 
\draw  [color={rgb, 255:red, 208; green, 2; blue, 27 }  ,draw opacity=1 ][fill={rgb, 255:red, 208; green, 2; blue, 27 }  ,fill opacity=1 ] (319.25,130) .. controls (319.25,127.65) and (321.15,125.75) .. (323.5,125.75) .. controls (325.85,125.75) and (327.75,127.65) .. (327.75,130) .. controls (327.75,132.35) and (325.85,134.25) .. (323.5,134.25) .. controls (321.15,134.25) and (319.25,132.35) .. (319.25,130) -- cycle ;
%Shape: Circle [id:dp6710803343430719] 
\draw  [color={rgb, 255:red, 208; green, 2; blue, 27 }  ,draw opacity=1 ][fill={rgb, 255:red, 208; green, 2; blue, 27 }  ,fill opacity=1 ] (219.25,268.25) .. controls (219.25,265.9) and (221.15,264) .. (223.5,264) .. controls (225.85,264) and (227.75,265.9) .. (227.75,268.25) .. controls (227.75,270.6) and (225.85,272.5) .. (223.5,272.5) .. controls (221.15,272.5) and (219.25,270.6) .. (219.25,268.25) -- cycle ;
%Shape: Circle [id:dp5244489024217467] 
\draw  [color={rgb, 255:red, 208; green, 2; blue, 27 }  ,draw opacity=1 ][fill={rgb, 255:red, 208; green, 2; blue, 27 }  ,fill opacity=1 ] (218.25,240) .. controls (218.25,237.65) and (220.15,235.75) .. (222.5,235.75) .. controls (224.85,235.75) and (226.75,237.65) .. (226.75,240) .. controls (226.75,242.35) and (224.85,244.25) .. (222.5,244.25) .. controls (220.15,244.25) and (218.25,242.35) .. (218.25,240) -- cycle ;

% Text Node
\draw (231,142) node [anchor=north west][inner sep=0.75pt]   [align=left] {$z$};
% Text Node
\draw (329.75,133) node [anchor=north west][inner sep=0.75pt]   [align=left] {$z$};
% Text Node
\draw (227,229) node [anchor=north west][inner sep=0.75pt]   [align=left] {$z$};
% Text Node
\draw (215,71) node [anchor=north west][inner sep=0.75pt]   [align=left] {{\large $K$}};
% Text Node
\draw (215,297) node [anchor=north west][inner sep=0.75pt]   [align=left] {{\large $K'$}};
% Text Node
\draw (357,119) node [anchor=north west][inner sep=0.75pt]   [align=left] {{\large $K''$}};

\end{tikzpicture}
                    }
                    \caption{The case described in \ref{pathclaim:before} of Claim~\ref{claim:properties of P}.}
                    \label{fig:(iii)}
                \end{figure}
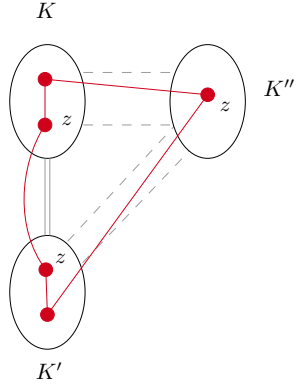
            \end{proof}

            We are now ready to conclude the proof of Lemma~\ref{lem:<4k vertices in the cliques}.
            Assume by contradiction that $T$ has at least $4k+1$ vertices in the cliques. 
            By the pigeonhole principle, there exists $j \in \{1, \ldots, k\}$ such that $T$ has at least 5 vertices in the cliques $K:=K^0_{2j-1},K':=K^1_{2j-1},K'':=K^0_{2j},K''':=K^1_{2j}$. 
            Again by the pigeonhole principle, $T$ has at least 2 vertices in one of these four cliques, and by Claim~\ref{claim:properties of P}, exactly 2.
            By symmetry, assume that $T$ has two vertices in $K$ (see Figure~\ref{fig:cases (iv)} for an illustration).
            Then, by \ref{pathclaim:antimatched} and \ref{pathclaim:before} of Claim~\ref{claim:properties of P}, $T$ has at most one vertex in each of the cliques $K',K'',K'''$, so exactly one in each (because $T$ contains five vertices in the union of these four cliques).
            Let us denote by $K[x],K[y]$ these two vertices in $K \cap T$, and by $K'[x'],K''[x''],K'''[x''']$ the vertices of $T$ in $K',K'',K'''$ respectively.
            Since $T$ is an induced tree, we have $x''=x$ or $x''=y$ (else, $\{x,y,x''\}$ would be a triangle, since $K''$ is antimatched with $K$).
            By symmetry, we can assume that $x''=x$.
            Similarly, we have $x''' = x$ or $x''' = y$. In the former case, $K[y]\dd K''[x]\dd K'''[x] \dd K[y]$ is a triangle in $T$. Hence, we have $x'''=y$.
                
                Finally, there are two cases, depicted in Figure~\ref{fig:cases (iv)}, which are the following:
                \begin{itemize}
                    \item If $x'\in\{x,y\}$, assume by symmetry that $x'=x$. Then, $K[x]\dd K'[x] \dd K'''[y] \dd K[x]$ is a triangle in $T$, a contradiction.
                    \item If $x'\notin\{x,y\}$, then $T$ has a $C_5$ as a subgraph. This is a contradiction.\qedhere
                \end{itemize}
                \begin{figure}[!ht]
                    \centering

                    \begin{subfigure}[t]{0.2\textwidth}
                        \scalebox{0.8}{
                        \begin{tikzpicture}[x=0.75pt,y=0.75pt,yscale=-1,xscale=1]
%uncomment if require: \path (0,460); %set diagram left start at 0, and has height of 460

%Straight Lines [id:da5041397694166209] 
\draw [color={rgb, 255:red, 155; green, 155; blue, 155 }  ,draw opacity=1 ] [dash pattern={on 4.5pt off 4.5pt}]  (228,231) -- (319,132) ;
%Straight Lines [id:da8918883606346641] 
\draw [color={rgb, 255:red, 155; green, 155; blue, 155 }  ,draw opacity=1 ] [dash pattern={on 4.5pt off 4.5pt}]  (229,253) -- (318,159) ;
%Straight Lines [id:da2361628106817677] 
\draw [color={rgb, 255:red, 155; green, 155; blue, 155 }  ,draw opacity=1 ] [dash pattern={on 4.5pt off 4.5pt}]  (229,134) -- (328,235) ;
%Straight Lines [id:da5478143004117517] 
\draw [color={rgb, 255:red, 155; green, 155; blue, 155 }  ,draw opacity=1 ] [dash pattern={on 4.5pt off 4.5pt}]  (222,149) -- (323.5,254.25) ;
%Straight Lines [id:da7655290915868957] 
\draw [color={rgb, 255:red, 155; green, 155; blue, 155 }  ,draw opacity=1 ] [dash pattern={on 4.5pt off 4.5pt}]  (222,236) -- (321,236) ;
%Straight Lines [id:da3245524275629026] 
\draw [color={rgb, 255:red, 155; green, 155; blue, 155 }  ,draw opacity=1 ] [dash pattern={on 4.5pt off 4.5pt}]  (222,269) -- (321,269) ;
%Straight Lines [id:da4376069564703454] 
\draw [color={rgb, 255:red, 155; green, 155; blue, 155 }  ,draw opacity=1 ] [dash pattern={on 4.5pt off 4.5pt}]  (222,149) -- (321,149) ;
%Straight Lines [id:da08514380027164803] 
\draw [color={rgb, 255:red, 155; green, 155; blue, 155 }  ,draw opacity=1 ] [dash pattern={on 4.5pt off 4.5pt}]  (222,116) -- (321,116) ;
%Shape: Ellipse [id:dp056769532364885356] 
\draw  [fill={rgb, 255:red, 255; green, 255; blue, 255 }  ,fill opacity=1 ] (200,134.25) .. controls (200,114.51) and (210.52,98.5) .. (223.5,98.5) .. controls (236.48,98.5) and (247,114.51) .. (247,134.25) .. controls (247,153.99) and (236.48,170) .. (223.5,170) .. controls (210.52,170) and (200,153.99) .. (200,134.25) -- cycle ;
%Shape: Ellipse [id:dp4705942230817538] 
\draw  [fill={rgb, 255:red, 255; green, 255; blue, 255 }  ,fill opacity=1 ] (200,254.25) .. controls (200,234.51) and (210.52,218.5) .. (223.5,218.5) .. controls (236.48,218.5) and (247,234.51) .. (247,254.25) .. controls (247,273.99) and (236.48,290) .. (223.5,290) .. controls (210.52,290) and (200,273.99) .. (200,254.25) -- cycle ;
%Shape: Ellipse [id:dp0032983445765769615] 
\draw  [fill={rgb, 255:red, 255; green, 255; blue, 255 }  ,fill opacity=1 ] (300,134.25) .. controls (300,114.51) and (310.52,98.5) .. (323.5,98.5) .. controls (336.48,98.5) and (347,114.51) .. (347,134.25) .. controls (347,153.99) and (336.48,170) .. (323.5,170) .. controls (310.52,170) and (300,153.99) .. (300,134.25) -- cycle ;
%Shape: Ellipse [id:dp05982308277934445] 
\draw  [fill={rgb, 255:red, 255; green, 255; blue, 255 }  ,fill opacity=1 ] (300,254.25) .. controls (300,234.51) and (310.52,218.5) .. (323.5,218.5) .. controls (336.48,218.5) and (347,234.51) .. (347,254.25) .. controls (347,273.99) and (336.48,290) .. (323.5,290) .. controls (310.52,290) and (300,273.99) .. (300,254.25) -- cycle ;
%Shape: Circle [id:dp23766479808869478] 
\draw  [color={rgb, 255:red, 208; green, 2; blue, 27 }  ,draw opacity=1 ][fill={rgb, 255:red, 208; green, 2; blue, 27 }  ,fill opacity=1 ] (227,118) .. controls (227,115.65) and (228.9,113.75) .. (231.25,113.75) .. controls (233.6,113.75) and (235.5,115.65) .. (235.5,118) .. controls (235.5,120.35) and (233.6,122.25) .. (231.25,122.25) .. controls (228.9,122.25) and (227,120.35) .. (227,118) -- cycle ;
%Shape: Circle [id:dp8824734742338014] 
\draw  [color={rgb, 255:red, 208; green, 2; blue, 27 }  ,draw opacity=1 ][fill={rgb, 255:red, 208; green, 2; blue, 27 }  ,fill opacity=1 ] (207.25,147.25) .. controls (207.25,144.9) and (209.15,143) .. (211.5,143) .. controls (213.85,143) and (215.75,144.9) .. (215.75,147.25) .. controls (215.75,149.6) and (213.85,151.5) .. (211.5,151.5) .. controls (209.15,151.5) and (207.25,149.6) .. (207.25,147.25) -- cycle ;
%Shape: Circle [id:dp4096661026279709] 
\draw  [color={rgb, 255:red, 208; green, 2; blue, 27 }  ,draw opacity=1 ][fill={rgb, 255:red, 208; green, 2; blue, 27 }  ,fill opacity=1 ] (319.25,254.25) .. controls (319.25,251.9) and (321.15,250) .. (323.5,250) .. controls (325.85,250) and (327.75,251.9) .. (327.75,254.25) .. controls (327.75,256.6) and (325.85,258.5) .. (323.5,258.5) .. controls (321.15,258.5) and (319.25,256.6) .. (319.25,254.25) -- cycle ;
%Shape: Circle [id:dp54180729926354] 
\draw  [color={rgb, 255:red, 208; green, 2; blue, 27 }  ,draw opacity=1 ][fill={rgb, 255:red, 208; green, 2; blue, 27 }  ,fill opacity=1 ] (219.25,254.25) .. controls (219.25,251.9) and (221.15,250) .. (223.5,250) .. controls (225.85,250) and (227.75,251.9) .. (227.75,254.25) .. controls (227.75,256.6) and (225.85,258.5) .. (223.5,258.5) .. controls (221.15,258.5) and (219.25,256.6) .. (219.25,254.25) -- cycle ;
%Straight Lines [id:da6923956682565864] 
\draw [color={rgb, 255:red, 155; green, 155; blue, 155 }  ,draw opacity=1 ]   (225,170) -- (225,218.5)(222,170) -- (222,218.5) ;
%Straight Lines [id:da4728504596550097] 
\draw [color={rgb, 255:red, 155; green, 155; blue, 155 }  ,draw opacity=1 ]   (325,170) -- (325,218.5)(322,170) -- (322,218.5) ;
%Shape: Circle [id:dp63204247686383] 
\draw  [color={rgb, 255:red, 208; green, 2; blue, 27 }  ,draw opacity=1 ][fill={rgb, 255:red, 208; green, 2; blue, 27 }  ,fill opacity=1 ] (319.25,134.25) .. controls (319.25,131.9) and (321.15,130) .. (323.5,130) .. controls (325.85,130) and (327.75,131.9) .. (327.75,134.25) .. controls (327.75,136.6) and (325.85,138.5) .. (323.5,138.5) .. controls (321.15,138.5) and (319.25,136.6) .. (319.25,134.25) -- cycle ;
%Straight Lines [id:da38157152210437806] 
\draw [color={rgb, 255:red, 208; green, 2; blue, 27 }  ,draw opacity=1 ]   (211.5,147.25) -- (231.25,118) ;
%Straight Lines [id:da9363338321777648] 
\draw [color={rgb, 255:red, 208; green, 2; blue, 27 }  ,draw opacity=1 ]   (231.25,118) -- (323.5,134.25) ;
%Straight Lines [id:da842506602198799] 
\draw [color={rgb, 255:red, 208; green, 2; blue, 27 }  ,draw opacity=1 ]   (211.5,147.25) -- (323.5,254.25) ;
%Straight Lines [id:da25679372664656874] 
\draw [color={rgb, 255:red, 208; green, 2; blue, 27 }  ,draw opacity=1 ]   (211.5,147.25) -- (223.5,254.25) ;
%Straight Lines [id:da14935521758198256] 
\draw [color={rgb, 255:red, 208; green, 2; blue, 27 }  ,draw opacity=1 ]   (323.5,254.25) -- (223.5,254.25) ;

% Text Node
\draw (323,113) node [anchor=north west][inner sep=0.75pt]   [align=left] {\textcolor[rgb]{0.82,0.01,0.11}{$x$}};
% Text Node
\draw (329.75,257.25) node [anchor=north west][inner sep=0.75pt]   [align=left] {\textcolor[rgb]{0.82,0.01,0.11}{$y$}};
% Text Node
\draw (223,136.5) node [anchor=north west][inner sep=0.75pt]   [align=left] {\textcolor[rgb]{0.82,0.01,0.11}{$x$}};
% Text Node
\draw (213,104) node [anchor=north west][inner sep=0.75pt]   [align=left] {\textcolor[rgb]{0.82,0.01,0.11}{$y$}};
% Text Node
\draw (215,261) node [anchor=north west][inner sep=0.75pt]   [align=left] {\textcolor[rgb]{0.82,0.01,0.11}{$x$}};
% Text Node
\draw (216,70) node [anchor=north west][inner sep=0.75pt]   [align=left] {{\large $K$}};
% Text Node
\draw (215,296) node [anchor=north west][inner sep=0.75pt]   [align=left] {{\large $K'$}};
% Text Node
\draw (317,71) node [anchor=north west][inner sep=0.75pt]   [align=left] {{\large $K''$}};
% Text Node
\draw (315,296) node [anchor=north west][inner sep=0.75pt]   [align=left] {{\large $K'''$}};

\end{tikzpicture}
                        }
                        \caption{The first case\\ ($x' \in \{x,y\}$).}
                    \end{subfigure}
                    \hspace{2cm}
                    \begin{subfigure}[t]{0.2\textwidth}
                        \scalebox{0.8}{
                        \begin{tikzpicture}[x=0.75pt,y=0.75pt,yscale=-1,xscale=1]
%uncomment if require: \path (0,460); %set diagram left start at 0, and has height of 460

%Straight Lines [id:da5041397694166209] 
\draw [color={rgb, 255:red, 155; green, 155; blue, 155 }  ,draw opacity=1 ] [dash pattern={on 4.5pt off 4.5pt}]  (228,231) -- (319,132) ;
%Straight Lines [id:da8918883606346641] 
\draw [color={rgb, 255:red, 155; green, 155; blue, 155 }  ,draw opacity=1 ] [dash pattern={on 4.5pt off 4.5pt}]  (229,253) -- (318,159) ;
%Straight Lines [id:da2361628106817677] 
\draw [color={rgb, 255:red, 155; green, 155; blue, 155 }  ,draw opacity=1 ] [dash pattern={on 4.5pt off 4.5pt}]  (229,134) -- (328,235) ;
%Straight Lines [id:da5478143004117517] 
\draw [color={rgb, 255:red, 155; green, 155; blue, 155 }  ,draw opacity=1 ] [dash pattern={on 4.5pt off 4.5pt}]  (222,149) -- (323.5,254.25) ;
%Straight Lines [id:da7655290915868957] 
\draw [color={rgb, 255:red, 155; green, 155; blue, 155 }  ,draw opacity=1 ] [dash pattern={on 4.5pt off 4.5pt}]  (222,236) -- (321,236) ;
%Straight Lines [id:da3245524275629026] 
\draw [color={rgb, 255:red, 155; green, 155; blue, 155 }  ,draw opacity=1 ] [dash pattern={on 4.5pt off 4.5pt}]  (222,269) -- (321,269) ;
%Straight Lines [id:da4376069564703454] 
\draw [color={rgb, 255:red, 155; green, 155; blue, 155 }  ,draw opacity=1 ] [dash pattern={on 4.5pt off 4.5pt}]  (222,149) -- (321,149) ;
%Straight Lines [id:da08514380027164803] 
\draw [color={rgb, 255:red, 155; green, 155; blue, 155 }  ,draw opacity=1 ] [dash pattern={on 4.5pt off 4.5pt}]  (222,116) -- (321,116) ;
%Shape: Ellipse [id:dp056769532364885356] 
\draw  [fill={rgb, 255:red, 255; green, 255; blue, 255 }  ,fill opacity=1 ] (200,134.25) .. controls (200,114.51) and (210.52,98.5) .. (223.5,98.5) .. controls (236.48,98.5) and (247,114.51) .. (247,134.25) .. controls (247,153.99) and (236.48,170) .. (223.5,170) .. controls (210.52,170) and (200,153.99) .. (200,134.25) -- cycle ;
%Shape: Ellipse [id:dp4705942230817538] 
\draw  [fill={rgb, 255:red, 255; green, 255; blue, 255 }  ,fill opacity=1 ] (200,254.25) .. controls (200,234.51) and (210.52,218.5) .. (223.5,218.5) .. controls (236.48,218.5) and (247,234.51) .. (247,254.25) .. controls (247,273.99) and (236.48,290) .. (223.5,290) .. controls (210.52,290) and (200,273.99) .. (200,254.25) -- cycle ;
%Shape: Ellipse [id:dp0032983445765769615] 
\draw  [fill={rgb, 255:red, 255; green, 255; blue, 255 }  ,fill opacity=1 ] (300,134.25) .. controls (300,114.51) and (310.52,98.5) .. (323.5,98.5) .. controls (336.48,98.5) and (347,114.51) .. (347,134.25) .. controls (347,153.99) and (336.48,170) .. (323.5,170) .. controls (310.52,170) and (300,153.99) .. (300,134.25) -- cycle ;
%Shape: Ellipse [id:dp05982308277934445] 
\draw  [fill={rgb, 255:red, 255; green, 255; blue, 255 }  ,fill opacity=1 ] (300,254.25) .. controls (300,234.51) and (310.52,218.5) .. (323.5,218.5) .. controls (336.48,218.5) and (347,234.51) .. (347,254.25) .. controls (347,273.99) and (336.48,290) .. (323.5,290) .. controls (310.52,290) and (300,273.99) .. (300,254.25) -- cycle ;
%Shape: Circle [id:dp23766479808869478] 
\draw  [color={rgb, 255:red, 208; green, 2; blue, 27 }  ,draw opacity=1 ][fill={rgb, 255:red, 208; green, 2; blue, 27 }  ,fill opacity=1 ] (227,118) .. controls (227,115.65) and (228.9,113.75) .. (231.25,113.75) .. controls (233.6,113.75) and (235.5,115.65) .. (235.5,118) .. controls (235.5,120.35) and (233.6,122.25) .. (231.25,122.25) .. controls (228.9,122.25) and (227,120.35) .. (227,118) -- cycle ;
%Shape: Circle [id:dp8824734742338014] 
\draw  [color={rgb, 255:red, 208; green, 2; blue, 27 }  ,draw opacity=1 ][fill={rgb, 255:red, 208; green, 2; blue, 27 }  ,fill opacity=1 ] (207.25,147.25) .. controls (207.25,144.9) and (209.15,143) .. (211.5,143) .. controls (213.85,143) and (215.75,144.9) .. (215.75,147.25) .. controls (215.75,149.6) and (213.85,151.5) .. (211.5,151.5) .. controls (209.15,151.5) and (207.25,149.6) .. (207.25,147.25) -- cycle ;
%Shape: Circle [id:dp4096661026279709] 
\draw  [color={rgb, 255:red, 208; green, 2; blue, 27 }  ,draw opacity=1 ][fill={rgb, 255:red, 208; green, 2; blue, 27 }  ,fill opacity=1 ] (319.25,254.25) .. controls (319.25,251.9) and (321.15,250) .. (323.5,250) .. controls (325.85,250) and (327.75,251.9) .. (327.75,254.25) .. controls (327.75,256.6) and (325.85,258.5) .. (323.5,258.5) .. controls (321.15,258.5) and (319.25,256.6) .. (319.25,254.25) -- cycle ;
%Shape: Circle [id:dp54180729926354] 
\draw  [color={rgb, 255:red, 208; green, 2; blue, 27 }  ,draw opacity=1 ][fill={rgb, 255:red, 208; green, 2; blue, 27 }  ,fill opacity=1 ] (219.25,254.25) .. controls (219.25,251.9) and (221.15,250) .. (223.5,250) .. controls (225.85,250) and (227.75,251.9) .. (227.75,254.25) .. controls (227.75,256.6) and (225.85,258.5) .. (223.5,258.5) .. controls (221.15,258.5) and (219.25,256.6) .. (219.25,254.25) -- cycle ;
%Straight Lines [id:da6923956682565864] 
\draw [color={rgb, 255:red, 155; green, 155; blue, 155 }  ,draw opacity=1 ]   (225,170) -- (225,218.5)(222,170) -- (222,218.5) ;
%Straight Lines [id:da4728504596550097] 
\draw [color={rgb, 255:red, 155; green, 155; blue, 155 }  ,draw opacity=1 ]   (325,170) -- (325,218.5)(322,170) -- (322,218.5) ;
%Shape: Circle [id:dp63204247686383] 
\draw  [color={rgb, 255:red, 208; green, 2; blue, 27 }  ,draw opacity=1 ][fill={rgb, 255:red, 208; green, 2; blue, 27 }  ,fill opacity=1 ] (319.25,134.25) .. controls (319.25,131.9) and (321.15,130) .. (323.5,130) .. controls (325.85,130) and (327.75,131.9) .. (327.75,134.25) .. controls (327.75,136.6) and (325.85,138.5) .. (323.5,138.5) .. controls (321.15,138.5) and (319.25,136.6) .. (319.25,134.25) -- cycle ;
%Straight Lines [id:da38157152210437806] 
\draw [color={rgb, 255:red, 208; green, 2; blue, 27 }  ,draw opacity=1 ]   (211.5,147.25) -- (231.25,118) ;
%Straight Lines [id:da9363338321777648] 
\draw [color={rgb, 255:red, 208; green, 2; blue, 27 }  ,draw opacity=1 ]   (231.25,118) -- (323.5,134.25) ;
%Straight Lines [id:da842506602198799] 
\draw [color={rgb, 255:red, 208; green, 2; blue, 27 }  ,draw opacity=1 ]   (211.5,147.25) -- (323.5,254.25) ;
%Straight Lines [id:da21067272227585065] 
\draw [color={rgb, 255:red, 208; green, 2; blue, 27 }  ,draw opacity=1 ]   (323.5,134.25) -- (223.5,254.25) ;
%Straight Lines [id:da14935521758198256] 
\draw [color={rgb, 255:red, 208; green, 2; blue, 27 }  ,draw opacity=1 ]   (323.5,254.25) -- (223.5,254.25) ;

% Text Node
\draw (323,113) node [anchor=north west][inner sep=0.75pt]   [align=left] {\textcolor[rgb]{0.82,0.01,0.11}{$x$}};
% Text Node
\draw (329.75,257.25) node [anchor=north west][inner sep=0.75pt]   [align=left] {\textcolor[rgb]{0.82,0.01,0.11}{$y$}};
% Text Node
\draw (223,136.5) node [anchor=north west][inner sep=0.75pt]   [align=left] {\textcolor[rgb]{0.82,0.01,0.11}{$x$}};
% Text Node
\draw (213,104) node [anchor=north west][inner sep=0.75pt]   [align=left] {\textcolor[rgb]{0.82,0.01,0.11}{$y$}};
% Text Node
\draw (215,261) node [anchor=north west][inner sep=0.75pt]   [align=left] {\textcolor[rgb]{0.82,0.01,0.11}{$x'$}};
% Text Node
\draw (216,70) node [anchor=north west][inner sep=0.75pt]   [align=left] {{\large $K$}};
% Text Node
\draw (215,296) node [anchor=north west][inner sep=0.75pt]   [align=left] {{\large $K'$}};
% Text Node
\draw (317,71) node [anchor=north west][inner sep=0.75pt]   [align=left] {{\large $K''$}};
% Text Node
\draw (315,296) node [anchor=north west][inner sep=0.75pt]   [align=left] {{\large $K'''$}};

\end{tikzpicture}
                        }
                        \caption{The second case ($x' \notin \{x,y\}$).}
                    \end{subfigure}
                    \caption{The two cases described at the end of the proof of Lemma~\ref{lem:<4k vertices in the cliques}.}
                    \label{fig:cases (iv)}
                \end{figure}
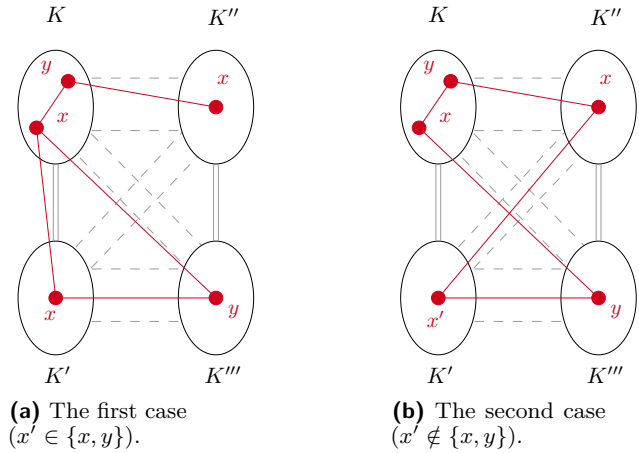
                
        \end{proof}

    We are finally able to prove Proposition~\ref{prop:common non-edge}. We start by the case of paths, and then extend the proof to handle the other case.
    \begin{proof}[Proof of Proposition~\ref{prop:common non-edge} for $P_{4k+3}$.]
        Assume first that $A$ and $B$ are not disjoint and contain a pair $\{i,j\}$. In particular, $G_A$ and $G_B$ have the common non-edges $ij'$ and $i'j$. Then, we can construct an induced path of length $4k+3$ as follows: $v_0,K^0_1[i],K^0_2[j],K^0_3[i], \ldots, K^0_{2k}[j]$ (alternating between $i$ and~$j$), $w$, $K^1_{2k}[i],K^1_{2k-1}[j],K^1_{2k-2}[i], \ldots, K^1_1[j],v_1$. This path is depicted on Figure~\ref{fig:P4k in G if common non edge}. 
        
    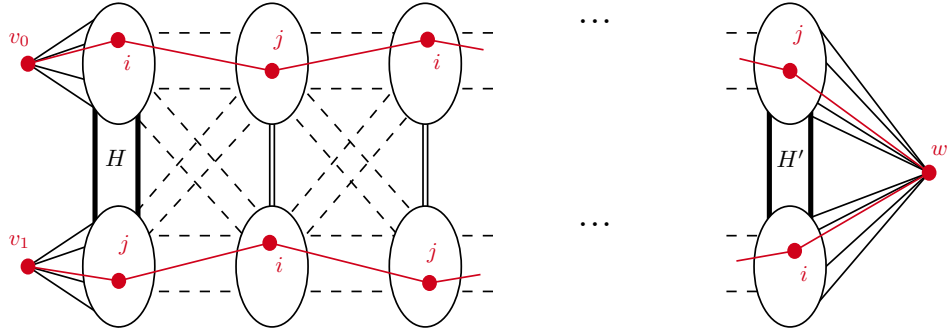
\begin{figure}
        \centering

        \scalebox{0.85}{
        \tikzset{every picture/.style={line width=0.75pt}} %set default line width to 0.75pt        

\begin{tikzpicture}[x=0.75pt,y=0.75pt,yscale=-1,xscale=.9]
%uncomment if require: \path (0,460); %set diagram left start at 0, and has height of 460

%Straight Lines [id:da05797722985575193] 
\draw  [dash pattern={on 4.5pt off 4.5pt}]  (520,269) -- (567,269) ;
%Straight Lines [id:da047592037490973516] 
\draw  [dash pattern={on 4.5pt off 4.5pt}]  (520,236) -- (567,236) ;
%Straight Lines [id:da24377939266462723] 
\draw  [dash pattern={on 4.5pt off 4.5pt}]  (520,149) -- (567,149) ;
%Straight Lines [id:da7287440613839526] 
\draw  [dash pattern={on 4.5pt off 4.5pt}]  (520,116) -- (567,116) ;
%Flowchart: Process [id:dp23916351351808063] 
\draw  [line width=2.25]  (548,152.5) -- (575,152.5) -- (575,239) -- (548,239) -- cycle ;
%Flowchart: Process [id:dp7411082105683671] 
\draw  [line width=2.25]  (108,154.5) -- (136,154.5) -- (136,241) -- (108,241) -- cycle ;
%Straight Lines [id:da0312638000155494] 
\draw    (652.25,198.75) -- (572,284) ;
%Straight Lines [id:da5213636217619234] 
\draw    (567,269) -- (652.25,198.75) ;
%Straight Lines [id:da6885712757040078] 
\draw    (652.25,198.75) -- (557.5,243) ;
%Straight Lines [id:da5119964432287542] 
\draw    (564,229) -- (652.25,198.75) ;
%Straight Lines [id:da3368516387266882] 
\draw    (575,107) -- (652.25,198.75) ;
%Straight Lines [id:da2889034475233626] 
\draw    (568,162) -- (652.25,198.75) ;
%Straight Lines [id:da2555163125023011] 
\draw    (573,122) -- (652.25,198.75) ;
%Straight Lines [id:da3404422996325329] 
\draw    (555.5,144) -- (652.25,198.75) ;
%Straight Lines [id:da23670221460583962] 
\draw  [dash pattern={on 4.5pt off 4.5pt}]  (228,231) -- (319,132) ;
%Straight Lines [id:da9368816458264855] 
\draw  [dash pattern={on 4.5pt off 4.5pt}]  (229,253) -- (318,159) ;
%Straight Lines [id:da7383385529700018] 
\draw  [dash pattern={on 4.5pt off 4.5pt}]  (229,134) -- (328,235) ;
%Straight Lines [id:da4543807443879616] 
\draw  [dash pattern={on 4.5pt off 4.5pt}]  (222,149) -- (323.5,254.25) ;
%Straight Lines [id:da14794431934881425] 
\draw  [dash pattern={on 4.5pt off 4.5pt}]  (129,133) -- (223,230) ;
%Straight Lines [id:da5121497523242154] 
\draw  [dash pattern={on 4.5pt off 4.5pt}]  (137,168) -- (223.5,254.25) ;
%Straight Lines [id:da5291845379714084] 
\draw  [dash pattern={on 4.5pt off 4.5pt}]  (135,222) -- (215,138) ;
%Straight Lines [id:da5021404891539996] 
\draw  [dash pattern={on 4.5pt off 4.5pt}]  (123.5,254.25) -- (221,158) ;
%Straight Lines [id:da3127983684958735] 
\draw  [dash pattern={on 4.5pt off 4.5pt}]  (123,236) -- (222,236) ;
%Straight Lines [id:da5922236117899703] 
\draw  [dash pattern={on 4.5pt off 4.5pt}]  (123,269) -- (222,269) ;
%Straight Lines [id:da642046433576111] 
\draw  [dash pattern={on 4.5pt off 4.5pt}]  (222,236) -- (321,236) ;
%Straight Lines [id:da6558568483371225] 
\draw  [dash pattern={on 4.5pt off 4.5pt}]  (222,269) -- (321,269) ;
%Straight Lines [id:da2188570533358517] 
\draw  [dash pattern={on 4.5pt off 4.5pt}]  (321,236) -- (368,236) ;
%Straight Lines [id:da5457868421768921] 
\draw  [dash pattern={on 4.5pt off 4.5pt}]  (321,269) -- (368,269) ;
%Straight Lines [id:da42553171555526303] 
\draw  [dash pattern={on 4.5pt off 4.5pt}]  (321,149) -- (368,149) ;
%Straight Lines [id:da9841613706685727] 
\draw  [dash pattern={on 4.5pt off 4.5pt}]  (321,116) -- (368,116) ;
%Straight Lines [id:da8090236964916413] 
\draw  [dash pattern={on 4.5pt off 4.5pt}]  (222,149) -- (321,149) ;
%Straight Lines [id:da7535597348412088] 
\draw  [dash pattern={on 4.5pt off 4.5pt}]  (222,116) -- (321,116) ;
%Straight Lines [id:da27559355488691106] 
\draw  [dash pattern={on 4.5pt off 4.5pt}]  (123,149) -- (222,149) ;
%Straight Lines [id:da7462747137112902] 
\draw  [dash pattern={on 4.5pt off 4.5pt}]  (123,116) -- (222,116) ;
%Straight Lines [id:da6790842557036169] 
%\draw [color={rgb, 255:red, 208; green, 2; blue, 27 }  ,draw opacity=1 ]   (14.25,254.25) -- (64.25,254.25) ;
%Straight Lines [id:da8910974183655151] 
%\draw [color={rgb, 255:red, 208; green, 2; blue, 27 }  ,draw opacity=1 ][fill={rgb, 255:red, 208; green, 2; blue, 27 }  ,fill opacity=1 ]   (14.25,134.25) -- (64.25,134.25) ;
%Straight Lines [id:da8265296630887734] 
\draw    (64.25,254.25) -- (123.5,218.5) ;
%Straight Lines [id:da9846097368103056] 
\draw    (64.25,254.25) -- (123.5,290) ;
%Straight Lines [id:da19518671476961014] 
\draw    (64.25,254.25) -- (123,236) ;
%Straight Lines [id:da025048109802588092] 
\draw    (64.25,254.25) -- (123,269) ;
%Straight Lines [id:da9197811377788397] 
\draw    (64.25,134.25) -- (123.5,98.5) ;
%Straight Lines [id:da2414804493959254] 
\draw    (64.25,134.25) -- (123.5,170) ;
%Straight Lines [id:da33842802932042026] 
\draw    (64.25,134.25) -- (123,116) ;
%Straight Lines [id:da7396133730856398] 
\draw    (64.25,134.25) -- (123,149) ;
%Shape: Ellipse [id:dp3342368755790637] 
\draw  [fill={rgb, 255:red, 255; green, 255; blue, 255 }  ,fill opacity=1 ] (100,254.25) .. controls (100,234.51) and (110.52,218.5) .. (123.5,218.5) .. controls (136.48,218.5) and (147,234.51) .. (147,254.25) .. controls (147,273.99) and (136.48,290) .. (123.5,290) .. controls (110.52,290) and (100,273.99) .. (100,254.25) -- cycle ;
%Shape: Ellipse [id:dp2575557827050784] 
\draw  [fill={rgb, 255:red, 255; green, 255; blue, 255 }  ,fill opacity=1 ] (100,134.25) .. controls (100,114.51) and (110.52,98.5) .. (123.5,98.5) .. controls (136.48,98.5) and (147,114.51) .. (147,134.25) .. controls (147,153.99) and (136.48,170) .. (123.5,170) .. controls (110.52,170) and (100,153.99) .. (100,134.25) -- cycle ;
%Shape: Ellipse [id:dp3002522670541836] 
\draw  [fill={rgb, 255:red, 255; green, 255; blue, 255 }  ,fill opacity=1 ] (200,134.25) .. controls (200,114.51) and (210.52,98.5) .. (223.5,98.5) .. controls (236.48,98.5) and (247,114.51) .. (247,134.25) .. controls (247,153.99) and (236.48,170) .. (223.5,170) .. controls (210.52,170) and (200,153.99) .. (200,134.25) -- cycle ;
%Shape: Ellipse [id:dp5431926424730037] 
\draw  [fill={rgb, 255:red, 255; green, 255; blue, 255 }  ,fill opacity=1 ] (200,254.25) .. controls (200,234.51) and (210.52,218.5) .. (223.5,218.5) .. controls (236.48,218.5) and (247,234.51) .. (247,254.25) .. controls (247,273.99) and (236.48,290) .. (223.5,290) .. controls (210.52,290) and (200,273.99) .. (200,254.25) -- cycle ;
%Shape: Ellipse [id:dp7023744884352279] 
\draw  [fill={rgb, 255:red, 255; green, 255; blue, 255 }  ,fill opacity=1 ] (300,134.25) .. controls (300,114.51) and (310.52,98.5) .. (323.5,98.5) .. controls (336.48,98.5) and (347,114.51) .. (347,134.25) .. controls (347,153.99) and (336.48,170) .. (323.5,170) .. controls (310.52,170) and (300,153.99) .. (300,134.25) -- cycle ;
%Shape: Ellipse [id:dp3793544431681709] 
\draw  [fill={rgb, 255:red, 255; green, 255; blue, 255 }  ,fill opacity=1 ] (300,254.25) .. controls (300,234.51) and (310.52,218.5) .. (323.5,218.5) .. controls (336.48,218.5) and (347,234.51) .. (347,254.25) .. controls (347,273.99) and (336.48,290) .. (323.5,290) .. controls (310.52,290) and (300,273.99) .. (300,254.25) -- cycle ;
%Shape: Ellipse [id:dp16370403693952185] 
\draw  [fill={rgb, 255:red, 255; green, 255; blue, 255 }  ,fill opacity=1 ] (538,134.25) .. controls (538,114.51) and (548.52,98.5) .. (561.5,98.5) .. controls (574.48,98.5) and (585,114.51) .. (585,134.25) .. controls (585,153.99) and (574.48,170) .. (561.5,170) .. controls (548.52,170) and (538,153.99) .. (538,134.25) -- cycle ;
%Shape: Ellipse [id:dp5085634851243372] 
\draw  [fill={rgb, 255:red, 255; green, 255; blue, 255 }  ,fill opacity=1 ] (538,254.25) .. controls (538,234.51) and (548.52,218.5) .. (561.5,218.5) .. controls (574.48,218.5) and (585,234.51) .. (585,254.25) .. controls (585,273.99) and (574.48,290) .. (561.5,290) .. controls (548.52,290) and (538,273.99) .. (538,254.25) -- cycle ;
%Shape: Circle [id:dp3703385320050828] 
\draw  [color={rgb, 255:red, 208; green, 2; blue, 27 }  ,draw opacity=1 ][fill={rgb, 255:red, 208; green, 2; blue, 27 }  ,fill opacity=1 ] (60,134.25) .. controls (60,131.9) and (61.9,130) .. (64.25,130) .. controls (66.6,130) and (68.5,131.9) .. (68.5,134.25) .. controls (68.5,136.6) and (66.6,138.5) .. (64.25,138.5) .. controls (61.9,138.5) and (60,136.6) .. (60,134.25) -- cycle ;
%Shape: Circle [id:dp6470369345738788] 
\draw  [color={rgb, 255:red, 208; green, 2; blue, 27 }  ,draw opacity=1 ][fill={rgb, 255:red, 208; green, 2; blue, 27 }  ,fill opacity=1 ] (60,254.25) .. controls (60,251.9) and (61.9,250) .. (64.25,250) .. controls (66.6,250) and (68.5,251.9) .. (68.5,254.25) .. controls (68.5,256.6) and (66.6,258.5) .. (64.25,258.5) .. controls (61.9,258.5) and (60,256.6) .. (60,254.25) -- cycle ;
%Shape: Circle [id:dp43176860521502736] 
%\draw  [color={rgb, 255:red, 208; green, 2; blue, 27 }  ,draw opacity=1 ][fill={rgb, 255:red, 208; green, 2; blue, 27 }  ,fill opacity=1 ] (10,134.25) .. controls (10,131.9) and (11.9,130) .. (14.25,130) .. controls (16.6,130) and (18.5,131.9) .. (18.5,134.25) .. controls (18.5,136.6) and (16.6,138.5) .. (14.25,138.5) .. controls (11.9,138.5) and (10,136.6) .. (10,134.25) -- cycle ;
%Shape: Circle [id:dp7157049975812985] 
%\draw  [color={rgb, 255:red, 208; green, 2; blue, 27 }  ,draw opacity=1 ][fill={rgb, 255:red, 208; green, 2; blue, 27 }  ,fill opacity=1 ] (10,254.25) .. controls (10,251.9) and (11.9,250) .. (14.25,250) .. controls (16.6,250) and (18.5,251.9) .. (18.5,254.25) .. controls (18.5,256.6) and (16.6,258.5) .. (14.25,258.5) .. controls (11.9,258.5) and (10,256.6) .. (10,254.25) -- cycle ;
%Shape: Circle [id:dp7289686668664345] 
\draw  [color={rgb, 255:red, 208; green, 2; blue, 27 }  ,draw opacity=1 ][fill={rgb, 255:red, 208; green, 2; blue, 27 }  ,fill opacity=1 ] (648,198.75) .. controls (648,196.4) and (649.9,194.5) .. (652.25,194.5) .. controls (654.6,194.5) and (656.5,196.4) .. (656.5,198.75) .. controls (656.5,201.1) and (654.6,203) .. (652.25,203) .. controls (649.9,203) and (648,201.1) .. (648,198.75) -- cycle ;
%Shape: Circle [id:dp13580028589905802] 
\draw  [color={rgb, 255:red, 208; green, 2; blue, 27 }  ,draw opacity=1 ][fill={rgb, 255:red, 208; green, 2; blue, 27 }  ,fill opacity=1 ] (118.75,120.25) .. controls (118.75,117.9) and (120.65,116) .. (123,116) .. controls (125.35,116) and (127.25,117.9) .. (127.25,120.25) .. controls (127.25,122.6) and (125.35,124.5) .. (123,124.5) .. controls (120.65,124.5) and (118.75,122.6) .. (118.75,120.25) -- cycle ;
%Straight Lines [id:da4050807157335564] 
\draw [color={rgb, 255:red, 208; green, 2; blue, 27 }  ,draw opacity=1 ]   (64.25,134.25) -- (123,120.25) ;
%Straight Lines [id:da7326605370747843] 
\draw [color={rgb, 255:red, 208; green, 2; blue, 27 }  ,draw opacity=1 ]   (123,120.25) -- (223.5,138.5) ;
%Shape: Circle [id:dp7730248601285922] 
\draw  [color={rgb, 255:red, 208; green, 2; blue, 27 }  ,draw opacity=1 ][fill={rgb, 255:red, 208; green, 2; blue, 27 }  ,fill opacity=1 ] (219.25,138.5) .. controls (219.25,136.15) and (221.15,134.25) .. (223.5,134.25) .. controls (225.85,134.25) and (227.75,136.15) .. (227.75,138.5) .. controls (227.75,140.85) and (225.85,142.75) .. (223.5,142.75) .. controls (221.15,142.75) and (219.25,140.85) .. (219.25,138.5) -- cycle ;
%Shape: Circle [id:dp6933564370001895] 
\draw  [color={rgb, 255:red, 208; green, 2; blue, 27 }  ,draw opacity=1 ][fill={rgb, 255:red, 208; green, 2; blue, 27 }  ,fill opacity=1 ] (320.75,120) .. controls (320.75,117.65) and (322.65,115.75) .. (325,115.75) .. controls (327.35,115.75) and (329.25,117.65) .. (329.25,120) .. controls (329.25,122.35) and (327.35,124.25) .. (325,124.25) .. controls (322.65,124.25) and (320.75,122.35) .. (320.75,120) -- cycle ;
%Shape: Circle [id:dp8737155238410298] 
\draw  [color={rgb, 255:red, 208; green, 2; blue, 27 }  ,draw opacity=1 ][fill={rgb, 255:red, 208; green, 2; blue, 27 }  ,fill opacity=1 ] (557.25,138.5) .. controls (557.25,136.15) and (559.15,134.25) .. (561.5,134.25) .. controls (563.85,134.25) and (565.75,136.15) .. (565.75,138.5) .. controls (565.75,140.85) and (563.85,142.75) .. (561.5,142.75) .. controls (559.15,142.75) and (557.25,140.85) .. (557.25,138.5) -- cycle ;
%Straight Lines [id:da06344262242764442] 
\draw [color={rgb, 255:red, 208; green, 2; blue, 27 }  ,draw opacity=1 ]   (223.5,138.5) -- (325,120) ;
%Straight Lines [id:da6302001403964619] 
\draw [color={rgb, 255:red, 208; green, 2; blue, 27 }  ,draw opacity=1 ]   (325,120) -- (361.5,126) ;
%Straight Lines [id:da8833208103537111] 
\draw [color={rgb, 255:red, 208; green, 2; blue, 27 }  ,draw opacity=1 ]   (528.5,131.25) -- (561.5,138.5) ;
%Straight Lines [id:da3183451419295862] 
\draw [color={rgb, 255:red, 208; green, 2; blue, 27 }  ,draw opacity=1 ][fill={rgb, 255:red, 208; green, 2; blue, 27 }  ,fill opacity=1 ]   (561.5,138.5) -- (652.25,198.75) ;
%Shape: Circle [id:dp6599306473423164] 
\draw  [color={rgb, 255:red, 208; green, 2; blue, 27 }  ,draw opacity=1 ][fill={rgb, 255:red, 208; green, 2; blue, 27 }  ,fill opacity=1 ] (560.25,245) .. controls (560.25,242.65) and (562.15,240.75) .. (564.5,240.75) .. controls (566.85,240.75) and (568.75,242.65) .. (568.75,245) .. controls (568.75,247.35) and (566.85,249.25) .. (564.5,249.25) .. controls (562.15,249.25) and (560.25,247.35) .. (560.25,245) -- cycle ;
%Straight Lines [id:da25553007981426135] 
\draw [color={rgb, 255:red, 208; green, 2; blue, 27 }  ,draw opacity=1 ]   (564.5,245) -- (652.25,198.75) ;
%Straight Lines [id:da5647479768176181] 
\draw [color={rgb, 255:red, 208; green, 2; blue, 27 }  ,draw opacity=1 ]   (526.5,251) -- (564.5,245) ;
%Shape: Circle [id:dp33566138751178054] 
\draw  [color={rgb, 255:red, 208; green, 2; blue, 27 }  ,draw opacity=1 ][fill={rgb, 255:red, 208; green, 2; blue, 27 }  ,fill opacity=1 ] (322,264) .. controls (322,261.65) and (323.9,259.75) .. (326.25,259.75) .. controls (328.6,259.75) and (330.5,261.65) .. (330.5,264) .. controls (330.5,266.35) and (328.6,268.25) .. (326.25,268.25) .. controls (323.9,268.25) and (322,266.35) .. (322,264) -- cycle ;
%Shape: Circle [id:dp9444688695761553] 
\draw  [color={rgb, 255:red, 208; green, 2; blue, 27 }  ,draw opacity=1 ][fill={rgb, 255:red, 208; green, 2; blue, 27 }  ,fill opacity=1 ] (217.75,240.25) .. controls (217.75,237.9) and (219.65,236) .. (222,236) .. controls (224.35,236) and (226.25,237.9) .. (226.25,240.25) .. controls (226.25,242.6) and (224.35,244.5) .. (222,244.5) .. controls (219.65,244.5) and (217.75,242.6) .. (217.75,240.25) -- cycle ;
%Shape: Circle [id:dp6013383836057612] 
\draw  [color={rgb, 255:red, 208; green, 2; blue, 27 }  ,draw opacity=1 ][fill={rgb, 255:red, 208; green, 2; blue, 27 }  ,fill opacity=1 ] (119.25,262.75) .. controls (119.25,260.4) and (121.15,258.5) .. (123.5,258.5) .. controls (125.85,258.5) and (127.75,260.4) .. (127.75,262.75) .. controls (127.75,265.1) and (125.85,267) .. (123.5,267) .. controls (121.15,267) and (119.25,265.1) .. (119.25,262.75) -- cycle ;
%Straight Lines [id:da8019270981879828] 
\draw [color={rgb, 255:red, 208; green, 2; blue, 27 }  ,draw opacity=1 ]   (64.25,254.25) -- (123.5,262.75) ;
%Straight Lines [id:da23023859392090995] 
\draw [color={rgb, 255:red, 208; green, 2; blue, 27 }  ,draw opacity=1 ]   (326.25,264) -- (359.5,259) ;
%Straight Lines [id:da19184307489766095] 
\draw [color={rgb, 255:red, 208; green, 2; blue, 27 }  ,draw opacity=1 ]   (222,240.25) -- (326.25,264) ;
%Straight Lines [id:da6860298516582647] 
\draw [color={rgb, 255:red, 208; green, 2; blue, 27 }  ,draw opacity=1 ]   (123.5,262.75) -- (222,240.25) ;
%Straight Lines [id:da9623296726929863] 
\draw    (225,170) -- (225,218.5)(222,170) -- (222,218.5) ;
%Straight Lines [id:da49909981976553264] 
\draw    (325,170) -- (325,218.5)(322,170) -- (322,218.5) ;

% Text Node
\draw (421,106.5) node [anchor=north west][inner sep=0.75pt]   [align=left] {{\huge ...}};
% Text Node
\draw (421,226.5) node [anchor=north west][inner sep=0.75pt]   [align=left] {{\huge ...}};
% Text Node
\draw (113,183.5) node [anchor=north west][inner sep=0.75pt]  [font=\large] [align=left] {$H$};
% Text Node
\draw (551,183.5) node [anchor=north west][inner sep=0.75pt]  [font=\large] [align=left] {$H'$};
% Text Node
\draw (125,127.5) node [anchor=north west][inner sep=0.75pt]   [align=left] {\textcolor[rgb]{0.82,0.01,0.11}{$i$}};
% Text Node
\draw (327,127.25) node [anchor=north west][inner sep=0.75pt]   [align=left] {\textcolor[rgb]{0.82,0.01,0.11}{$i$}};
% Text Node
\draw (224,247.5) node [anchor=north west][inner sep=0.75pt]   [align=left] {\textcolor[rgb]{0.82,0.01,0.11}{$i$}};
% Text Node
\draw (223,113) node [anchor=north west][inner sep=0.75pt]   [align=left] {\textcolor[rgb]{0.82,0.01,0.11}{$j$}};
% Text Node
\draw (122,234) node [anchor=north west][inner sep=0.75pt]   [align=left] {\textcolor[rgb]{0.82,0.01,0.11}{$j$}};
% Text Node
\draw (323,239) node [anchor=north west][inner sep=0.75pt]   [align=left] {\textcolor[rgb]{0.82,0.01,0.11}{$j$}};
% Text Node
\draw (562,110) node [anchor=north west][inner sep=0.75pt]   [align=left] {\textcolor[rgb]{0.82,0.01,0.11}{$j$}};
% Text Node
\draw (566.5,252.25) node [anchor=north west][inner sep=0.75pt]   [align=left] {\textcolor[rgb]{0.82,0.01,0.11}{$i$}};
%\draw (0,114.5) node [anchor=north west][inner sep=0.75pt]   [align=left] {\textcolor[rgb]{0.82,0.01,0.11}{$u_0$}};
%\draw (0,234) node [anchor=north west][inner sep=0.75pt]   [align=left] {\textcolor[rgb]{0.82,0.01,0.11}{$u_1$}};
\draw (50,114.5) node [anchor=north west][inner sep=0.75pt]   [align=left] {\textcolor[rgb]{0.82,0.01,0.11}{$v_0$}};
\draw (50,234) node [anchor=north west][inner sep=0.75pt]   [align=left] {\textcolor[rgb]{0.82,0.01,0.11}{$v_1$}};
\draw (652,181) node [anchor=north west][inner sep=0.75pt]   [align=left] {\textcolor[rgb]{0.82,0.01,0.11}{$w$}};

\end{tikzpicture}

}
        
        \caption{The path on $4k+3$ vertices in $G_{k,n}(A,B)$ if $A$ and $B$ both contain $\{i,j\}$.}
        \label{fig:P4k in G if common non edge}
    \end{figure}

    Conversely, assume that there is an induced path $P$ on $4k+3$ vertices in $G_{n,k}(A,B)$, and let us show that $A,B$ are not disjoint.
    By Lemma~\ref{lem:<4k vertices in the cliques}, $P$ has at most $4k$ vertices in the cliques, and since there are exactly $3$ vertices outside the cliques, $P$ has exactly $4k$ vertices in the cliques, and it uses all the vertices outside of the cliques, namely $v_0,v_1,w$. Note that the neighborhoods of $v_0$ and $v_1$ in $G_{n,k}(H,H')$ are cliques, hence they must have degree $1$ in $P$ so they are the two endpoints of $P$. So $P$ starts with $v_0, x_1 \in K^0_1$, and ends with $v_1,y_1 \in K^1_1$. 
    
    %Between its start and its end, $P$ passes through $w$, so $P$ must contain two vertices $K^0[x_\ell]$ and $K^1_\ell[y_\ell]$ for all $1 \leqslant \ell \leqslant 2k$ ($K^0_\ell[x_\ell]$ before passing through 
    Since $w \in V(P)$ is not an end of $P$ it must have degree two in $P$.
    Hence, for some $i,j \in \{1,2, \dots, n\}$, $K^{0}_{2k}[j], K^{1}_{2k}[i]$ are the neighbors of $w$ in $V(P)$.
    $K^{0}_{2k}[i], K^{1}_{2k}[j]$ must be non-adjacent, so $ji'$ is not an edge of $G_B$.
    Since $G_A$ is reflexive it follows that $i \neq j$ and since it is symmetric $ij'$ is not an edge of $G_B$.
    \begin{claim}\label{inner-claim}
    For each $\ell \in \{1,2, \dots, 2k\}$ the vertices of $P$ in $K^0_\ell \cup K^1_\ell$ are exactly $K^{0}_{\ell}[i], K^{1}_{\ell}[j]$ or they are exactly $K^{0}_{\ell}[j], K^{1}_{\ell}[i]$.
    \end{claim}
    \begin{subproof}
        Let $\ell \in \{1,2, \dots, 2k-1\}$ be maximum such that the claim does not hold for $\ell$.
        By symmetry, we may assume that both $x:= K^{0}_{\ell+1}[j], y:= K^{1}_{\ell+1}[i]$ are elements of $V(P)$.
        Since $v_0, v_1$ are the ends of $P$, both $x,y$ must have a neighbor in $K^{0}_{\ell} \cup K^{1}_{\ell}$.
        In particular, there must be vertices $x', y' \in  K^{0}_{\ell} \cup K^{1}_{\ell}$ so that $xx', yy'$ are edges and there are no further edges between vertices in $\{x,x', y,y'\}$.    
        By construction for each $K \in \{K^{0}_{\ell}, K^{1}_{\ell} \}$ and $K' \in \{K^{0}_{\ell+1}, K^{1}_{\ell+1} \}$, there is an antimatching between $K$ and $K'$.
        So in particular, either $x' = K^{0}_{\ell}[i]$ and $y' =  K^{1}_{\ell}[j]$ or 
        $x' = K^{1}_{\ell}[i]$ and $y' = K^{0}_{\ell}[j]$. 
    \end{subproof}
   % Since $x,y$ are non-adjacent, we may assume that $x=K^{0}_{k}[i]$ and $y=$x=K^{0}_{k}[i]$
   % $w$, and $K^1[y_\ell]$ after passing through~$w$).
    %Moreover, these two vertices are not consecutive in $P$, so they are not adjacent.
    %Thus, they are in different cliques: one is in $K^0_\ell$, the other one in $K^1_\ell$.
    %In particular, 
    %Hence, $x_\ell, y_\ell$ are non-adjacent for each $\ell$
    %Finally, the only possible shape for $P$ is the following one:
    %$$P = (u_0,v_0,K^{i_1}_1[x_1], K^{i_1}_2[x_2], \ldots, K^{i_{2k}}_{2k}[x_{2k}],w,K^{1-i_{2k}}_{2k}[y_{2k}], K^{1-i_{2k-1}}_{2k-1}[y_{2k-1}], \ldots, K^{1-i_1}_1[y_1],v_1,u_1)$$
    Thus by \cref{inner-claim} and since $G_A$ is symmetric, $ij'$ is a non-edge of $G_A$.
    So $G_A$ and $G_B$ have a common non-edge $ij'$, and $A,B$ both contain $\{i,j\}$, as desired. 
\end{proof}

We may now adapt this proof to handle the case of trees without vertices of degree $2$.
\begin{proof}[Proof of the remaining case of Proposition~\ref{prop:common non-edge}.]
If $A$ and $B$ are not disjoint, then by the previous proof, there is an induced path on $4k+3$ vertices in $G_{k,n}(A,B)$ as depicted on Figure~\ref{fig:P4k in G if common non edge}. Adding all the vertices outside of the cliques to this path yields a copy of $T$. 

Conversely, assume that there is an induced copy of $T$ in $G_{n,k}(A,B)$, and let us show that $G_A,G_B$ have a common non-edge. By Lemma~\ref{lem:<4k vertices in the cliques}, $T$ has at most $4k$ vertices in the cliques, and since there are exactly $|V(T)|-4k$ vertices outside the cliques, $T$ has exactly $4k$ vertices in the cliques, and it uses all the vertices outside of the cliques. In particular, by connectivity, $T$ must contain at least a vertex in each clique (hence exactly once). Using this, we can show that Claim~\ref{inner-claim} still holds in this case, which concludes.
\end{proof}
\section{Upper bound core technique: layered maps}
\label{app:eccs}

%%%%%%%%%%%%%%%%%%%%%%%%%%%%%%%%%%%%%%%%%%%%%%
\subsection{Spread universal certification and graphs of large minimum degree}
\label{subsec:warm-up-app}

We have described in Subsection~\ref{subsec:results-techniques} the general idea of layered maps, which are a refinement of the technique of spreading the certificate of the universal scheme. 
We formalize the spread universal scheme (introduced in~\cite{FeuilloleyFHPP21}) in the following theorem. 

\begin{restatable}{theorem}{thmupperwarmup}
\label{thm:warm up} The following statements both hold:
\begin{enumerate}[(i)]
    \item Let $\delta < 1$. Any property can be certified with verification radius~$2$ and certificates of size $O(n^{2-\delta} \log n)$ on graphs of minimum degree $n^\delta$.
    \item Any property can be certified with verification radius~$2$ and certificates of size $O(n \log^2 n)$ on regular graphs.
\end{enumerate}
\end{restatable}

%%%%%%%%%%%%%%%%%% SHORT vs LONG VERSION %%%%%%%%%%%

%The proof is deferred to Appendix~\ref{app:warm-up-upper-bound}, and we sketch it here for the first item. 

Before giving the proof, we give a short sketch for the case of the first item.
%
%%%%%%%%%%%%%%%%%%%%%%%%%%%%%%%%%%%%%%%%%%%%%%%%%%%%%
On correct instances, the prover computes the certificate of the universal scheme in the form of an adjacency matrix, cuts it into $n^{\delta}$ pieces of size $n^{2-\delta}$, and gives $O(\log n)$ pieces to every node, in such a way that each vertex can see all the pieces in its neighborhood. 
Such an assignment of pieces is proved to exist via probabilistic method and coupon collector. 
Note that we need the verification radius to be at least two, only to ensure that every node can check that the graph it reconstructs is the same as its neighbors.

Before proving Theorem~\ref{thm:warm up}, we show the following Lemma~\ref{lem:coupon collector} (inspired from the coupon collector theorem), that we will use several times.

\begin{restatable}{lemma}{lemmacouponcollector}
    \label{lem:coupon collector}
    Let $G$ be a $n$-vertex graph and $0 < d < n$. Let $\mathcal{P}$ be a set of size~$d$, called the set of \emph{pieces}. Then, there exists a way to assign $3 \log n$ pieces to every vertex of $G$ such that, for every vertex $u$ of degree at least $d$, each piece of $\mathcal{P}$ has been given to $u$ or one of its neighbors.
\end{restatable}
\begin{proof}
    Let us consider a random assignation of the pieces to the vertices of $G$, where each vertex gets assigned a subset of $\mathcal{P}$ of size $3 \log n$ pieces (uniformly at random, independently of all the other vertices). Let us fix some piece $p \in \mathcal{P}$ and some vertex $u$ of degree at least $d$. Observe that the probability that $p$ has not been given to $u$ not to its  neighbors is at most:
    \[ \left(1 - \frac{3\log n}{d}\right)^d \leqslant e^{-3 \log n} \leqslant \frac{1}{n^3}.\]
    By union bound, the probability that such an event occurs for some vertex $u$ of degree at least $d$ and some piece $p \in \mathcal{P}$ is at most $\frac{1}{n}$. This is strictly smaller than $1$ for all $n \geqslant 2$. Thus, there exists some correct assignment of the pieces.
\end{proof}

We may now proceed with the proof of Theorem~\ref{thm:warm up}.

\begin{proof}[Proof of Theorem~\ref{thm:warm up}.]

Let us prove the cases (i) and (ii) one after the other.

\begin{enumerate}[(i)] 
\item Let us describe a certification scheme of size $O(n^{2-\delta} \log n)$ for any property, in graphs of minimum degree $n^\delta$, if the verification radius is at least~$2$. Let $G = (V,E)$ be a graph with $|V| = n$, and minimum degree $n^\delta$.

\medskip{}
\textbf{Certification.} The prover assigns to each vertex $u$ a certificate divided in two fields: one common to every node (denoted by $\spanningtree$), and another specific to $u$ (denoted by $\pieces(u)$).

    \begin{itemize}
        \item In $\spanningtree$, the prover writes a spanning tree $\mathcal{T}$ of $G$, that is: the identifiers of the vertices and of their parent in $\mathcal{T}$. This uses $O(n \log n)$ bits.

        \item For $\pieces$, the prover cuts the adjacency matrix of $G$ in $n^\delta$ parts, each of size $n^{2-\delta}$, and number these parts from $1$ to $n^{\delta}$.
        Then, using Lemma~\ref{lem:coupon collector}, since all the vertices have degree at least $n^\delta$, it may write $3 \log n$ numbered pieces in each $\pieces(u)$ such that every vertex sees each piece at least once in its closed neighborhood. This certificate has size $O(n^{2-\delta} \log n)$.
    \end{itemize}

\medskip{}
\textbf{Verification.}
    The verification algorithm of each vertex $u$ is done in two steps.
    \begin{itemize}
        \item The first step is to check the correctness of $\spanningtree$. To do so, $u$ checks if $\spanningtree$ is the same in its certificate and the certificate of all its neighbors. Then, it checks if it is indeed a tree, and if each of its neighbors in $\spanningtree$ is also a neighbor in $G$. If it is not the case, $u$ rejects.
        
        If no vertex rejects at this point, $\spanningtree$ is a correct spanning tree known by all the vertices. In particular, every vertex~$u$ knows~$n$, the number of vertices in~$G$.

        \item The second and main step of the verification is the following. Since $u$ knows $n$, $u$ knows the number pieces (which is $n^\delta$). The vertex $u$ checks if it sees each numbered piece appears in $\pieces(v)$ for some vertex $v$ in its closed neighborhood, and rejects if it is not the case. If $u$ did not reject, it may thus reconstruct the whole graph $G$, and can also recover the graph reconstructed by every of its neighbors since the verification radius is at least~$2$. If some of these reconstructed graphs are different, $u$ rejects. Then, $u$ checks if its neighborhood is correctly written in its reconstructed graph. Finally, if $u$ did not reject before, it accepts if and only if its reconstructed graph satisfies the property.
    \end{itemize}

\medskip{}
\textbf{Correctness.}
    If no vertex rejects in the verification procedure, all the vertices reconstructed the same graph. Moreover, this reconstructed graph is equal to $G$, because every vertex checked that its neighborhood is correct in it. Since every vertex accepts if and only if this graph satisfies the property, the scheme is correct.

    \medskip{}

    \item Let us now describe a certification scheme of size $O(n \log^2 n)$ for any property in regular graphs, if the verification radius is at least~$2$. Let $G$ be an $n$-vertex $d$-regular graph.

    \medskip{}
    \textbf{Certification.}
    The certification is similar to the proof of (i), except that since $G$ is a $d$-regular graph, it can be represented with $O(d n \log n)$ bits using adjacency lists (instead of the adjacency matrix). The prover cuts it in $d$ pieces, each of size $O(n \log n)$, numbers them from $1$ to $d$, and applies Lemma~\ref{lem:coupon collector} to assign $3 \log n$ pieces in the certificate of every vertex. This yields certificates of size $O(n \log^2 n)$.

    \medskip{}
    \textbf{Verification.}
    The verification algorithm is the following one. Every vertex $u$ knows its own degree, which is equal to the number of pieces. So $u$ can check if it sees all the pieces in its closed neighborhood. Then, as for (i), $u$ reconstructs the whole graph $G$ and checks if all its neighbors reconstructed the same graph. Finally, $u$ checks if its own neighborhood is correctly written, and if the property is satisfied.

    \medskip{}
    \textbf{Correctness.} As for (i), if no vertex rejects in the verification procedure, all the vertices computed the same graph, which is equal to $G$ and satisfies the property.
    \qedhere
\end{enumerate}
\end{proof}

\subsection{Definitions: extended connected components, ECC table, and more}
%%%%%%%%%%%%%%%%%%%%%%%%%%%%%%%%%%%%%%%%%%%%%%
\label{sub:ecci-dfns-app}

We now introduce formally the layered map, that we have informally described in the introduction, along with other related notions.

%\subsubsection{Definitions}

\paragraph*{Extended connected components.}
%Let $G=(V,E)$ be an $n$-vertex graph, in which the vertices are given unique identifiers on $O(\log n)$ bits. 
Let $\varepsilon > 0$ and $N = \lceil \frac{1}{\varepsilon}\rceil$. We consider the following partition of the vertices depending on their degrees: for every $i \in \{1, \ldots, N\}$, let $V_i:=\{u \in V \; | \; n^{(i-1)\varepsilon} \leqslant \deg(u) < n^{i\varepsilon}\}$ (where $\deg(u)$ is the degree of the vertex $u$). 
Note that the sets $V_1, \ldots, V_N$ form a partition of $V$.
Let $L_i := \bigcup_{j\leqslant i} V_i$, that is, $L_i$ is the set of vertices having degree less than than $n^{i\varepsilon}$. And let $H_i := \bigcup_{j \geqslant i} V_i$, that is, $H_i$ is the set of vertices having degree at least $n^{(i-1)\varepsilon}$ (it is also $V\setminus L_{i-1}$). Note that $H_1 = L_N = V$.

Let $k \geqslant 2$ be the verification radius.
For each $i \in \{1, \ldots, N\}$, we partition $H_i$ in subsets called \emph{extended connected components of $H_i$} (abbreviated by ECC$_i$).
For $u,v \in H_i$, we say $u,v$ are \emph{$i$-linked} if there is a $u,v$ path in $G$ which does not have $2k-2$ consecutive vertices in $L_{i-1}$ (in other words, such a path should regularly contain vertices in $H_i$). Note that being $i$-linked is an equivalence relation.
An ECC$_i$ is an equivalence class for this relation.
By definition, two vertices in two different ECC$_i$'s are at distance at least $2k-1$ from each other in $G$.

%Intuitively speaking, we can ensure that all vertices in an $ECC_i$, will agree after looking at their neighborhood.
%\textcolor{blue}{S. : pas exactement: ils sont d'accord sur l'information des sommets de $L_i$...} \linda{+1} 
%\linda{I agree, this wording is a bit ambiguous here. maybe something like we can ensure that all vertices in an $ECC_i$, compute similar global certificates after looking at their neighborhood.}
Let us denote by $\mathcal{E}_i$ the set of all ECC$_i$'s (which is a partition of $H_i$). For each $C \in \cE_i$, we define the \emph{identifier of C}, denoted by $id(C)$, as the smallest identifier of a vertex in $C$.
Finally, we denote by $C(d)$ the set of vertices at distance exactly $d$ from $C$ (by convention, we set $C(0):=C$). Note that $C(0) \subseteq H_i$, and for all $d \geqslant 1$, we have $C(d) \subseteq L_{i-1} = V(G) \setminus H_i$. 
%Note also that for all distinct $C,C' \in \cE_i$ and for all $d \in \{0, \ldots, k-1\}$, we have $C(d) \cap C'(d) = \emptyset$.
%La: I commented that since it is basically the same as the observation. 

%\marginpar{potential shortening: delete green text or leave for later}
Note that for every $u \not \in H_i$ there is at most one $ECC_i$ $C \in \cE_i$ such that $u$ is at distance at most $k-1$ from $C$.
%Since being $i$-linked is an equivalence relation we obtain the following:
%\begin{observation}\label{ecci-unique}
%    For every $u \not \in H_i$ there is at most one $ECC_i$ $C \in \cE_i$ such that $u$ is at distance at most $k-1$ from $C$.
%\end{observation}
%La: incorporated the observation in the main text. Modified the rest accordingly.

\paragraph*{ECC-table.}
Let $T_G$ be the following table, with $n$ rows and $N$ columns, called the \emph{ECC-table of $G$}. The rows are indexed by the identifiers of the vertices, and the columns by $\{1, \ldots, N\}$. Let $u \in V$ and $i \in \{1, \ldots, N\}$. Let us describe the entry $T_G[u,i]$.
\begin{itemize}
    \item %If $u \in \bigcup_{C \in \cE_i} \bigcup_{0 \leqslant d \leqslant k-1} C(d)$, denote by 
    %that is if $u$ is at distance at most $k-1$ from some ECC$_i$, then this ECC$_i$ is unique. 
    If $u$ is at distance at most $k-1$ from some ECC$_i$ $C_u$, 
    we set $T_G[u,i]$ as $(id(C_u), d_u)$ where $d_u$ is the distance from $u$ to $C_u$. 
    Note that $u \in H_i$ if and only if $d_u = 0$. 
    %and $C_u$ is uniquely defined by \cref{ecci-unique}
    %thanks for the explanation theo
    \item Otherwise, we set $T_G[u,i] = \bot$.
\end{itemize}

In other words, for each $i \in \{1, \ldots, N\}$, the partition of $H_i$ in ECC$_i$'s is written in the table $T_G$, and for every vertex $u \in L_{i-1}$ "close" to some $C_u \in \cE_i$, the distance from $u$ to $C_u$ is also stored in~$T_G$. 
%(Note that it is indeed well-defined since a vertex is close to at most one ECC).

\paragraph*{Witnessed graph.}\label{pre-curly-dfn-app}
We define a relation $\preccurlyeq$ on $V$ in the following way. Let $u,v \in V$. We say that $v \preccurlyeq u$ if there exists $i \in \{1, \ldots, N\}$ such that $v \in V_i$, $u \in H_i$, and $u,v$ are in the same ECC$_i$. Note that, by definition, $\preccurlyeq$ is transitive. For every $u \in V$, let $V_{\preccurlyeq u}$ be the set $\{v \in V \; | \; v \preccurlyeq u\}$, called the \emph{set of vertices witnessed by $u$}. Let $G_{\preccurlyeq u}$ be the subgraph of $G$ obtained by keeping only edges having at least one endpoint in $V_{\preccurlyeq u}$, which is called the \emph{graph witnessed  by $u$}. 
Intuitively, if $v \in V_{\preccurlyeq u}$, then $u$ can check all the adjacencies of $v$ thanks to the certificates (see Theorem~\ref{thm:computation scheme for g_<u}). 

%Finally, let $f_\varepsilon(G,u):=(G_{\preccurlyeq u}, T_G)$.

\paragraph*{Local computation scheme.}
Now, let us introduce the notion of \emph{local computation scheme}. Informally, it is a tool allowing the vertices to perform a pre-computation at the beginning of their verification in a certification scheme\footnote{This new definition allows us to compose different computations using certificates, which is not possible in general with the standard certification definition, because of its binary output.}. We will use it so that each vertex $u$ will pre-compute the ECC-table $T_G$ and its witnessed graph $G_{\preccurlyeq u}$.

\begin{restatable}{definition}{defcomputationscheme}
    \label{def:local-computation-scheme}
    {\sffamily\normalshape (Local computation scheme).} Let $f$ be a function taking as input a graph $G$ and a vertex $u$ of $G$. A \emph{local computation scheme for $f$ with verification radius~$k$ and size~$s$} is a scheme where the prover gives certificates of size $s$ to the vertices, and each vertex either rejects, or outputs something, depending only on its view at distance~$k$. We also require the two following soundness conditions to be satisfied for every graph $G=(V,E)$:
\begin{enumerate}[(i)]
    \item if no vertex rejects, then every vertex $u \in V$ outputs $f(G,u)$;
    \item there exists a certificate assignment such that no vertex rejects.
\end{enumerate}
\end{restatable}

\subsection{Technical results}

We now design local computation schemes to compute the ECC-table and the witnessed graphs, as stated in the following Theorems~\ref{thm:computation scheme for T_G} and~\ref{thm:computation scheme for g_<u}.

\begin{restatable}{theorem}{ThmComputationTG}
    \label{thm:computation scheme for T_G}
    There exists $c>0$ such that, for all $0 < \varepsilon < 1$, there exists a local computation scheme for $f(G,u):=T_G$ with verification radius~$k$ and certificates of size $\frac{c}{\varepsilon} \cdot n \log n$.
\end{restatable}

%\begin{remark}
    Note that the output of the function $f$ in Theorem~\ref{thm:computation scheme for T_G} does not depend on $u$, but only on~$G$ and~$\varepsilon$.
%\end{remark}

\begin{proof}
    Let $0 < \varepsilon < 1$ and $N = \lceil \frac{1}{\varepsilon} \rceil$. Let us describe a local computation scheme for $f(G,u)=T_G$ of size $O(\frac{n}{\varepsilon} \log n)$.  
    
    \medskip{}
    \textbf{Certification.}
    The certificates of the vertices consists of information stored in three fields, denoted by $\spanningtree$, $\tab$ and $\comp$. The certificate is the same for all the vertices.
    %Let us first explain which information is stored in each of these fields.
    
    \begin{itemize}
        \item The prover chooses a spanning tree $\mathcal{T}$ of $G$ and writes it in $\spanningtree$. More precisely: it writes the identifiers of all the vertices and of their parents in $\mathcal{T}$ using $O(n \log n)$ bits.
        
        \item In $\tab$, the prover writes the table $T_G$. Since it has $n$ rows, $N$ columns, and $O(\log n)$ bits per cell, it has size $O(N n \log n)$.
        
        \item In $\comp$, the prover gives information to the nodes to check the correctness of the partition in $ECC_i$'s written in $\tab$.
        For each $i \in \{1, \ldots, N\}$, and for each $C \in \cE_i$, the prover constructs the graph $G_C$, where the vertices of $G_C$ are the vertices of $C$, and there is an edge between two vertices in $G_C$ if and only if they are at distance at most $2k-2$ in $G$. By definition of an ECC$_i$, $G_C$ is connected. Then, the prover chooses a spanning tree $\mathcal{T}_C$ of $G_C$ and writes its structure in $\comp$, with the identifier of the corresponding vertices. For each edge $(u,v)$ in $\T_C$, there exists a vertex $w \in V$ at distance at most $k-1$ from both $u$ and $v$ in $G$. The prover labels the edge $(u,v)$ in $\T_C$ by the identifier of $w$.
        For a given $C \in \cE_i$, $O(|C| \log n)$ bits are required. In total, since the prover does this for each $C \in \cE_i$ and every $i \in \{1, \ldots, N\}$, it uses $O(N n \log n)$ bits.
    \end{itemize}
    Since $N = \lceil\frac{1}{\varepsilon}\rceil$, the overall size of the certificate is thus $O(\frac{n}{\varepsilon} \log n)$.

    The verification consists in a simple consistency check, with the spanning trees enforcing the global correctness. 
    
    \medskip
\textbf{Verification.}
    The vertices perform the following verification procedure.

    \begin{enumerate}[(i)]
        \item First, every vertex checks that its certificate is the same as the certificates of its neighbors, and rejects if it is not the case.

        \item To check the correctness of $\spanningtree$, every vertex checks if it is indeed a tree, and if each of its neighbors written in $\spanningtree$ is indeed a neighbor in $G$. It rejects if it is not the case.
        
        If no vertex rejects at this point, the spanning tree of~$G$ written in $\spanningtree$ is correct. In particular, all the vertices know~$n$, and the whole set of the identifiers of vertices in~$G$.
        
        \item The next two steps of the verification consist in checking if the partition in ECC$_i$'s written in $T_G$ is correct. Every vertex~$u$ checks that for all $i \in \{1, \ldots, N\}$, and for all vertices $v, w \in H_i$ at distance at most $k-1$ from $u$, we have $\tab[v,i]=\tab[w,i]$, and if this common value is of the form $(id(C),0)$. If it is not the case, $u$ rejects.
        
        If no vertex rejects at this point, then for all
        $i \in \{1, \ldots, N\}$ and for every pair of vertices $v,w \in H_i$ which are at distance at most $2k-2$, $v$ and $w$ are written to be in the same ECC$_i$ in $\tab$. By transitivity: if two vertices are in the same ECC$_i$ in $G$, then they appear in the same ECC$_i$ in $\tab$.
        
        \item Every vertex $u \in V$ determines the index $j \in \{1, \ldots, N\}$ such that $u \in V_j$. For every $i \leqslant j$, $u$ does the following. We have $u \in H_i$. Let $C_u \in \cE_i$ be the ECC$_i$ of $u$. The vertex $u$ checks that $\T_{C_u}$ written in $\comp$ is indeed a tree. Moreover, for each edge $(u,v)$ in $\T_{C_u}$ labeled with the identifier of a vertex $w$, $u$ checks if it sees indeed $w$ at distance at most $k-1$ (and rejects if it is not the case).
        
        If no vertex rejects at this point, for all vertices $u,v \in V$, if $\tab[u,i] = \tab[v,i]$ and if this value is of the form $(id(C),0)$, then $u$ and $v$ are indeed two vertices of $H_i$ which are in the same ECC$_i$ (indeed, $\T_C$ is connected, so there exists a path from $u$ to $v$ in $\T_C$, which corresponds to a path in $G$ which does not have $2k-1$ consecutive vertices in $L_{i-1}$).
        
        Thus, together with step~(iii) of the verification, the partition in ECC$_i$'s written in $\tab$ is correct.

        \item The next step of the verification consists in checking if the distances written in $\tab$ are also correct. To do so, every vertex $u$ does the following. For each $i \in \{1, \ldots, n\}$, if $u \notin H_i$, $u$ checks that $\tab[u,i] = \bot$ if and only if $u$ does not see any vertex of $H_i$ at distance at most $k-1$. And if $u$ sees vertices in some $C_u \in \cE_i$ at distance $d_u$, $u$ checks that $\tab[u,i]=(id(C_u),d_u)$. If it is not the case, $u$ rejects.
    \end{enumerate}

    \medskip{}
    \textbf{Computation.} If a vertex $u$ did not reject during the verification phase, it outputs $\tab$. If no vertex rejects, we have $\tab=T_G$ (since both the partition in ECC$_i$'s and the distances to the ECC$_i$'s written in $\tab$ are correct). Moreover, if the prover gives the certificates as described above, no vertex will reject. Thus, the computation scheme is correct.
\end{proof}

%\textcolor{red}{Note that, in the proof of Theorem~\ref{thm:computation scheme for T_G}, all the vertices receive the same certificate.}
% La: mentioned before the itemize.

%\linda{Note, that $G_{\preccurlyeq u}$ is defined in \cref{pre-curly-dfn-app}.}
%Or maybe we should move this definition to be closer to the theorem statement
% La: right, it's not so nice, but I guess we can trust the reader to go back to the definition section... :/
\begin{restatable}{theorem}{ThmComputationGu}
    \label{thm:computation scheme for g_<u}
    There exists $c>0$ such that, for all $0 < \varepsilon < 1$, there exists a local computation scheme for $f(G,u):=G_{\preccurlyeq u}$ of verification radius~$k$ and size $\frac{c}{\varepsilon} \cdot n^{1+\varepsilon} \cdot \log^2 n$.
\end{restatable}

\begin{proof}[Proof of Theorem~\ref{thm:computation scheme for g_<u}.]
    Let $0 < \varepsilon < 1$ and $N = \lceil \frac{1}{\varepsilon} \rceil$. We describe the local computation scheme.
    % \footnote{Note that in theorem the verification radius is 2, but this result is just a component of a scheme with radius $k\geq 2$, hence the ECC are still defined with respect to this general $k$.} 
    %for $f(G,u)=G_{\preccurlyeq u}$ which has size $O(\frac{1}{\varepsilon}\cdot n^{1+\varepsilon}\cdot \log^2 n)$.

    \medskip{}
    \textbf{Certification.}
    Let us describe the certificates given by the prover to the vertices on correct instances. First, it gives to every vertex $u$ its certificate in the local computation scheme for $T_G$ given by Theorem~\ref{thm:computation scheme for T_G}, which has size $O(\frac{n}{\varepsilon} \log n)$.
    Then, the prover gives other additional information to $u$, denoted by $\pieces(u)$, which is the following.  For every $i \in \{1, \ldots, N\}$, the prover constructs the graph $G_i$ obtained from~$G$ by keeping only edges having at least one endpoint in $L_i$. Since each vertex in $L_i$ has degree at most $n^{i \varepsilon}$, $G_i$ has at most $n^{1+i\varepsilon}$ edges. Then, the prover cuts the adjacency list of~$G_i$ (which has size $n^{1+i\varepsilon}\log n$) in $n^{(i-1)\varepsilon}$ pieces of size $n^{1+\varepsilon} \log n$, and number them from~$1$ to~$n^{(i-1)\varepsilon}$. Finally, for every vertex $u$, it writes $3 \log n$ numbered pieces of $G_i$ in $\pieces(u)$ in such a way: for each $u \in H_i$, $u$ sees every piece $P$ of $G_i$ in the certificate of at least one of its neighbors. Since each vertex $u \in H_i$ has degree at least $n^{(i-1)\varepsilon}$, and since there are $n^{(i-1)\varepsilon}$ different pieces of $G_i$, this is possible (see Lemma~\ref{lem:coupon collector} 
    % %%%%%%%%%%%% SHORT vs LONG VERSIONS %%%%%%%%%%%%%
    %in the appendix).
    ).
    %%%%%%%%%%%%%%%%%%%%%%%%%%%%%%%%%%%%%%%%%%%%%%%%%%%%
    In total, since the prover does this for every $i \in \{1, \ldots, N\}$, $\pieces(u)$ has size $O(N n^{1+\varepsilon} \log^2 n)$, so the size of the overall certificate is $O(\frac{n^{1+\varepsilon}}{\varepsilon} \log^2 n)$.

    \medskip{}
    \textbf{Verification.} The vertices perform the following verification procedure.
    \begin{enumerate}[(i)]
        \item First, each vertex applies the verification of the local computation scheme for $T_G$ given by Theorem~\ref{thm:computation scheme for T_G}.
        
        \item To verify $\pieces(u)$, each vertex $u$ does the following: if it sees two numbered pieces $P \in \pieces(v)$ and $P' \in \pieces(w)$ for some $v,w$ at distance at most $k-1$, such that $P$ and $P'$ are two pieces of $G_j$ for some $j \in \{1,\ldots,N\}$ which are numbered the same, $u$ checks that these two pieces are indeed the same. If it is not the case, $u$ rejects.
        
        \item Finally, for every $i \in \{1, \ldots, N\}$, each $u \in V_i$ does the following verification. For every $j \leqslant i$, and for every numbered piece $P$ of $G_j$, $u$ checks if it sees $P$ in a certificate $\pieces(v)$ for some $v$ in its closed neighborhood. If it is not the case, $u$ rejects.
        And if $u$ sees that its own edges and non-edges are not correctly written in the piece $P$ of $G_i$ where it should be, it rejects.
    \end{enumerate}

    \medskip{}
    \textbf{Computation.}
    Let us describe the computation of every vertex $u$. If no vertex rejected at step~(i), $u$ computed the ECC-table $T_G$. In particular, $u$ knows the partition of $H_i$ in ECC$_i$'s for every $i \in \{1, \ldots, N\}$. So $u$ can compute the set $V_{\preccurlyeq u}$.
    
    \begin{claim}
        \label{claim:v<u}
        Let $u,v \in V$ such that $v \preccurlyeq u$. If no vertex rejected during the verification phase, then $u$ can compute all the edges and non-edges having $v$ as an endpoint.
    \end{claim}
    \begin{proof}
        Let $i \in \{1, \ldots, N\}$ such that $v \in V_i$ and $u \in H_i$.
        Because of step (iii) of the verification, the vertex $u$ sees all the parts of $G_i$ in its certificate or in the certificate of some of its neighbors. However, nothing ensures that these parts are correct. But thanks to step (ii), any two vertices in $H_i$ at distance at most $2k-2$ get the same pieces of $G_i$ in their neighborhoods. So by transitivity, two vertices which are in the same ECC$_i$ have the same pieces of $G_i$ in their neighborhood. In particular, since $u$ and $v$ are in the same ECC$_i$, and since the edges and non-edges of $v$ are correctly written in the pieces which are in the closed neighborhood of $v$ (otherwise it would have rejected at step (iii)), then $u$ knows the edges and non-edges of~$v$.
    \end{proof}

    Thus, if no vertex rejects in the verification phase, thanks to Claim~\ref{claim:v<u}, every vertex $u$ is able to reconstruct the graph $G_{\preccurlyeq u}$ thanks to the certificates $\pieces$ it sees its neighborhood. Moreover, if the prover gives the certificates as described above, no vertex will reject. So the computation scheme is correct.
\end{proof}

\begin{restatable}{remark}{RemECCs}
    \label{rem:computation witnessed graph other vertices}
    Note that, with the certificates given in the certification scheme of Theorem~\ref{thm:computation scheme for g_<u}, each vertex $v$ can also output the witnessed graph $G_{\preccurlyeq u}$ of any vertex $u$ at distance at most $k-1$ from itself. Indeed, $v$ can see all the pieces of $G_{\preccurlyeq u}$ which are spread in certificates of vertices in $N[u]$.
\end{restatable}
\section{Formal proofs of the forbidden subgraph certifications}
\label{app:advanced-upper-bounds}

In this section, we prove all our upper bounds: Theorem~\ref{thm:p_4k_free_subquadratic}, Theorem~\ref{thm:H_4k_free_subquadratic}, Theorem~\ref{thm:3k-quasilinear} and Theorem~\ref{thm:P_143k_free_subquadratic}. These are presented by increasing level of technicality. 
In this section, we prove all our upper bounds: Theorem~\ref{thm:p_4k_free_subquadratic}, Theorem~\ref{thm:H_4k_free_subquadratic}, Theorem~\ref{thm:3k-quasilinear} and Theorem~\ref{thm:P_143k_free_subquadratic}. These are presented by increasing level of technicality. 

\subsection{Upper bound for $P_{4k-1}$ in $\Tilde{O}(n^{3/2})$}

\thmPFourkFreeSubquadratic*

\begin{proof}[Proof of Theorem~\ref{thm:p_4k_free_subquadratic}.]
    Let $G = (V,E)$ be an $n$-vertex graph.

    \medskip{}
    \textbf{Certification.}
    The certification of the prover is the following. First, it gives to each vertex its certificates in the local computation schemes given by Theorems~\ref{thm:computation scheme for T_G} and~\ref{thm:computation scheme for g_<u}, with $\varepsilon=\frac{1}{2}$, so that each vertex~$u$ can compute $T_G$ and $G_{\preccurlyeq u
    }$. These certificates have size $O(n^{3/2}\log^2 n)$.
    %Note that, using the notations of the proof of Theorems~\ref{thm:computation scheme for T_G} and~\ref{thm:computation scheme for g_<u}, we have $N=2$, so $V$ is partitioned into $V_1$ (the set of vertices having degree less than $\sqrt{n}$) and $V_2$ (the set of vertices having degree at least $\sqrt{n}$). We also have $H_1=V$, and there is a unique ECC$_1$ which is equal to $V$. Thus, for every $v \in V_1$ and $u \in V$, we have $v \in V_{\preccurlyeq u}$.

    Then, the prover adds some additional information in the certificates, in another field, which will be identical in the certificates of all the vertices. We will denote this field by $\lp$, and it is defined as follows (note that it has size $O(n \log n)$). 

    \begin{definition}[$P_v, \lp$]\label{def:lp}
    For every vertex $v \in V_1$ such that $T_G[v,2]=(id(C_v),d_v)$, let $P_v$ be a longest {induced} path starting from $v$ and having all its other vertices in $\bigcup_{0 \leq d < d_v} C_v(d)$,
    and let $\ell(P_v)$ be its length. 
    We define $\lp$ as being the array indexed by the identifiers of the vertices, such that $\lp[v]=\ell(P_v)$  for every $v \in \bigcup_{C \in \cE_2} \bigcup_{1 \leqslant d \leqslant k-1} C(d)$.
    \end{definition}

    \medskip{}
    \textbf{Verification.}
   % The verification of each vertex $u \in V$ consists in the following steps with $m = 4k$.
   We define for each $m \geq 1$ a protocol we call \emph{$m$-pathcheck} below.
      (We define the $m$-pathcheck for arbitrary $m$, because we will use it again when we prove \cref{thm:P_143k_free_subquadratic}.)
   In our verification phase each vertex runs the $m$-pathcheck for $m = 4k-1$.
    
\begin{definition}[$m$-pathcheck] \label{def:mpathcheck}
    Each vertex $u \in V$ takes as input its view at distance $k$ (with the certificates described above) and peforms the following steps:
    \begin{enumerate}[(i)]
    \item First, $u$ applies the verification of the local computation schemes of Theorems~\ref{thm:computation scheme for T_G} and~\ref{thm:computation scheme for g_<u}.
    If no vertex rejects during this verification phase, then every vertex $u$ computed $T_G$ and~$G_{\preccurlyeq u}$.
    \label{step:precompute}
    \item Then, $u$ checks that $\lp$ is the same in its certificate and the certificates of all its neighbors. If it is not the case, $u$ rejects.
    If no vertex rejects at this point, then $\lp$ is the same in the certificates of all vertices in $G$.
    %
    %

    % \item If $u \in V_2$, it does the following.
    
    % Let $C_u$ be its ECC$_2$. We have $\bigcup_{1 \leqslant d \leqslant k-1} C_u(d) \subseteq V_1 \subseteq V_{\preccurlyeq u}$. \linda{don't think this last one is quite correct or is it? Does all of $V_1$ really have to be in the same ECC$_1$ as $u$. It doesn't really seem to matter for your proof anyway}
    % Thus, for each $v \in \bigcup_{1 \leqslant d \leqslant k-1} C_u(d)$, $u$ can compute the length of the path $P_v$, and check if it is correctly written in $\lp$. If this verification fails, $u$ rejects.

    % If no vertex rejects at this point, then $\lp$ is correct.
    \item Suppose $u \in V_2$. Let $C_u$ denote the ECC$_2$ containing $u$. 
    %This should go in some outside observation
    By definition, of ECC and $V_{\preccurlyeq u}$, $\bigcup_{1 \leqslant d \leqslant k-1} C_u(d) \subseteq V_1 \subseteq V_{\preccurlyeq u}$
    %Since $H_1 = V$, by definition all vertices of $\bigcup_{1 \leqslant d \leqslant k-1} C_u(d) \subseteq V_1$ are in the same ECC$_1$ as $u$. So in particular, $\bigcup_{1 \leqslant d \leqslant k-1} C_u(d) \subseteq V_{\preccurlyeq u}$.
    Since $u$ computed $G_{\preccurlyeq u}$ in Step \ref{step:precompute}, 
    for each $v \in \bigcup_{1 \leqslant d \leqslant k-1} C_u(d)$, $u$ can compute the length of the path $P_v$, and check if it is correctly written in $\lp$. If this verification fails, $u$ rejects. 
    If no vertex rejects at this point, then $\lp$ is correct in the certificate of every vertex. \label{step:check-lp}

    % If $\lcp(u) \neq \bot$, $u$ performs the following additional checks.
    % First 
    % %Recall, $TG[u,2] = \bot$ if and only if $u$ is at distance at most $k-1$ from some ECC$_2$.
    
    % To do so, $u$ first checks that $\lcp(u)=\bot$ if and only if $u \notin \bigcup_{C \in \cE_2} \bigcup_{1 \leqslant d \leqslant k-1} C(d)$ (note that $u$ knows it thanks to $T_G$).
    % Then, if $u \in \bigcup_{C \in \cE_2} \bigcup_{1 \leqslant d \leqslant k-1} C(d)$, let $C_u \in \cE_2$ and $d_u \in \{1, \ldots, k-1\}$ such that $u \in C_u(d_u)$. Let $u' \in C_u$ at distance $d_u$ from $u$. By Remark~\ref{rem:computation witnessed graph other vertices}, $u$ can compute $G_{\preccurlyeq u'}$. In particular, $u$ knows the subgraph of $G$ induced by $\bigcup_{0 \leqslant d \leqslant k-1} C_u(d)$, and this is sufficient for $u$ to check the correctness of $\lcp(u)$ (and $u$ rejects if it is not correct).\label{step:check-constrained}

   % If no vertex rejects at this point, then~$\lcp(u)$ is correctly written in the certificate of every vertex~$u$.
    \end{enumerate}

    We may assume no vertex has rejected during steps \ref{step:precompute} to \ref{step:check-lp}, and thus every vertex $u$ has access to a (correct) $T_G$, ~$G_{\preccurlyeq u}$ and $\lp$.
    
    \begin{enumerate}[(i)]
    \setcounter{enumi}{3}
    \item If $G_{\preccurlyeq u}$ has an induced path $P$ on $m$ vertices, such that all the vertices in $P \setminus V_{\preccurlyeq u}$ are in distinct ECC$_2$'s, then $u$ rejects.
   \label{step:visiblempath} 
    %(in particular, this is the case if $P \subseteq V_{\preccurlyeq u}$).

   % \item If $u \in V_2$, $u$ performs the following checks. First let $C_u$ denote the ECC$_2$ containing $u$.

    \item %\linda{in previous draft this was step vii} 
    \label{step:uselp-simple}
    If $u \in V_2$, let us denote by $C_u$ its ECC$_2$. By definition, $V_{\preccurlyeq u} = V_1 \cup C_u$ so by step \ref{step:precompute}
$u$ knows the graph induced by $V_1 \cup C_u$.    %It rejects if it sees an induced path $P_{\text{start}}$ in $G[C_u \cup V_1]$ (which is itself an induced subgraph of $G_{\preccurlyeq u}$, \linda{
    It rejects if it sees an induced path $P_{\text{start}}$ in $G[C_u \cup V_1]$ satisfying all of the following:
        \begin{itemize}
            \item $P_{\text{start}}$ contains at least one vertex from $C_u$
            \item $P_{\text{start}}$ ends in a vertex $v$, such that $T_G[v,2] = (id(C_v), d_v)$ with $id(C_v) \neq id(C_u)$.  
            This implies that $v$ is at distance $d_v \leq k-1$ from $C_v \in \cE_2 \setminus C_u$.
            (Recall we assumed that no vertex rejected at step \ref{step:precompute} and so $u$ has access to $T_G$.)
            %\item $v$ is the only vertex in $P_{\text{start}} \cap (\bigcup_{d \leqslant d_v} C_v(d))$
            \item $v$ is the unique vertex in $P_{\text{start}}$ that is at distance at most $d_v$ from $C_v$. (again $u$ can check this because it has access to $T_G$.)
            \item $|V(P_{\text{start}})| + (\lp[v]) \geqslant m$
        \end{itemize}

    This case corresponds to the case represented on Figure~\ref{fig:P_start P_end}.
    
\end{enumerate}
\end{definition}

    \medskip{}
    \textbf{Properties of the $m$-pathcheck.}
  %  {\linda{I put some of the parts of the correctness proof that we reuse for $14/3$ in outside lemmas so we can refer to them more easily}
    We begin with some observations about the $m$-pathcheck. 
     \begin{lemma}
        \label{claim:4k-1 2}
        Suppose $G$ contains an induced path $P$ on $m$ vertices.
        %and every vertex in $G$ runs the $m$-pathcheck.
        If there is at most one $C \in \cE_2$ such that $|C \cap P| \geqslant 2$, some vertex rejects during the $m$-pathcheck.
    \end{lemma}
    \begin{proof}
       We may assume that no vertex rejects during the first three steps of the $m$-pathcheck.
        If $P \subseteq V_1$, then every vertex $u \in P$ rejects at step (iv). 
        Else, we have $P \cap V_2 \neq \emptyset$, and let $C \in \cE_2$ such that $|P \cap C|$ is maximized. By assumption, for every $C' \in \cE_2$ such that $C' \neq C$, we have $|P \cap C'| \leqslant 1$, so every $u \in C$ rejects at step (iv).
    \end{proof}

%     Applying Claims~\ref{claim:14/3k 1}, \ref{claim:14/3k 2} and~\ref{claim:14/3k 3}, we can assume that there exactly two ECC$_2$'s which intersect~$P$, both on at least two vertices.
% Let us denote these two ECC$_2$'s by $C_u$ and $C_v$.
% If there exists $d \in \{1, \ldots, k-1\}$ such that $|P \cap C_v(d)| = 1$, then any vertex $u \in C_u$ will reject at step~(vi) of the verification (with the exact same proof as in Theorem~\ref{thm:p_4k_free_subquadratic}).
\begin{lemma}\label{lem:4k-2ECC-1-uniquedistance}
    Suppose $G$ contains an induced path $P$ on $m$ vertices with the following properties:
    \begin{itemize}
        \item there are exactly two distinct ECC$_2$'s $C_u, C_v \in \cE_2$ which contain vertices of $P$
        \item for some $d \in \{1,2, \dots, k-1\}$, there is a unique vertex $v$ at distance exactly $d$ from $C_v$
    \end{itemize}
    Then some vertex rejects during the $m$-pathcheck.
\end{lemma}
\begin{proof}
      As depicted in Figure~\ref{fig:P_start P_end}, we can decompose $P$ in two consecutive parts $P_{\text{start}},P_{\text{end}}$ such that:
    \begin{itemize}
        \item $P_{\text{start}} \cap P_{\text{end}} = \{v\}$;
        \item $P_{\text{start}} \subseteq C_u \cup V_1$;
        \item $P_{\text{end}}\setminus \{v\} \subseteq \bigcup_{0 \leqslant d < d_v} C_v(d)$.
    \end{itemize}
    We may assume that no vertex rejects during the first four steps of the $m$-pathcheck.
    We claim that every $u \in P \cap C_u$ will reject at step (v).
    
    %Let us denote by $\ell(P_{\text{end}})$ the length of $P_{\text{end}}$.
    %Then length of $P_{\text{start}}$ is equal to $m-\ell(P_{\text{end}})$.
    By definition, $|V(P_{\text{start}})| = m- |V(P_{\text{end}})| -1$.
    Let $P_v$ denote the the longest induced path starting at $v$ and having all its other vertices strictly closer from $C_v$ than $v$.
    Hence, $|V(P_v)| \geqslant |V(P_{\text{end}})|$, 
    so $\ell(P_{\text{start}}) + \lp[v] \geqslant m$.
  %  Let $u \in P \cap C_u$ rejects at step (v).
    $P_{\text{start}}$ is an induced path of $G_{\preccurlyeq u}$, included in $V_{\preccurlyeq u}$, which satisfies the conditions making $u$ reject at step~\ref{step:uselp-simple}.
\end{proof}

    \medskip{}
    \textbf{Correctness.}
    Let us show that this certification scheme is correct. 
    Assume first by contradiction that $G$ contains an induced path $P$ of length $4k-1$. Let us show that for every assignment of certificates, at least one vertex rejects.
    If no vertex rejects at step (i) of the verification,  then every vertex $u$ knows its witnessed graph $G_{\preccurlyeq u}$ and the ECC-table $T_G$ (note that $u$ also knows $V_{\preccurlyeq u}$, because it can be computed directly from $T_G$). If no vertex rejects in step (iii), then $\lp$ is also correct.
    Then, let us prove the two following Claims~\ref{claim:4k-1 1} and~\ref{claim:4k-1 2}.
    
    \begin{claim}
        \label{claim:4k-1 1}
        There are at most two $C \in \cE_2$ such that $|C \cap P| \geqslant 2$. And if there are exactly two, then $P$ does not contain any vertex from any other ECC$_2$.
    \end{claim}
    \begin{proof}
        This simply follows from the fact that $P$ has length $4k-1$, and two different ECC$_2$'s are at distance at least $2k-1$.
    \end{proof}
    
    % \begin{claim}
    %     \label{claim:4k-1 2}
    %     If there is at most one $C \in \cE_2$ such that $|C \cap P| \geqslant 2$, then there exists a vertex which rejects at step (iv) of the verification.
    % \end{claim}
    % \begin{proof}
    %     If $P \subseteq V_1$, then every vertex $u \in P$ rejects at step (iv).
    %     Else, we have $P \cap V_2 \neq \emptyset$, and let $C \in \cE_2$ such that $|P \cap C|$ is maximized. By assumption, for every $C' \in \cE_2$ such that $C' \neq C$, we have $|P \cap C'| \leqslant 1$, so every $u \in C$ rejects at step (iv).
    % \end{proof}

    Applying Claim~\ref{claim:4k-1 1} and~\cref{claim:4k-1 2}, we can assume that there are exactly two different ECC$_2$'s which intersect $P$, and both on at least two vertices. Let us denote them by $C_u$ and $C_v$. Since $P$ is a path going through $C_u$ and $C_v$ which are at distance at least $2k-1$ from each other, for all $d \in \{1, \ldots, k-1\}$, we have \mbox{$P \cap C_u(d) \neq \emptyset$}, and \mbox{$P \cap C_v(d) \neq \emptyset$}. Moreover, since $P$ has length $4k-1$, and since $P \cap C_u(0)$, $P \cap C_v(0)$ both contain at least two vertices, there exists a set $P \cap C_u(d)$ or $P \cap C_v(d)$ which contains only one vertex, for some $d \in \{1, \ldots, k-1\}$. 
    Hence, by applying \cref{lem:4k-2ECC-1-uniquedistance} for $m = 4k-1$ we obtain that a some vertex rejects during our verification protocol as desired. 
    % By symmetry, assume that there exists $v \in V$ and $d_v \in \{1, \ldots, k-1\}$ such that $P \cap C_v(d_v) = \{v\}$. As depicted on Figure~\ref{fig:P_start P_end}, we can decompose $P$ in two consecutive parts $P_{\text{start}},P_{\text{end}}$ such that:
    % \begin{itemize}
    %     \item $P_{\text{start}} \cap P_{\text{end}} = \{v\}$;
    %     \item $P_{\text{start}} \subseteq C_u \cup V_1$;
    %     \item $P_{\text{end}}\setminus \{v\} \subseteq \bigcup_{0 \leqslant d < d_v} C_v(d)$.
    % \end{itemize}

    \begin{figure}
        \centering
        \scalebox{0.5}{
        \begin{tikzpicture}[x=0.75pt,y=0.75pt,yscale=-1,xscale=1]
%uncomment if require: \path (0,484); %set diagram left start at 0, and has height of 484

%Shape: Polygon Curved [id:ds9771891593097473] 
\draw  [color={rgb, 255:red, 255; green, 255; blue, 255 }  ,draw opacity=1 ][fill={rgb, 255:red, 74; green, 144; blue, 226 }  ,fill opacity=0.29 ] (223,47.5) .. controls (234,90.5) and (218,179.5) .. (109,185.5) .. controls (0,191.5) and (26,108.5) .. (76,65.5) .. controls (126,22.5) and (212,4.5) .. (223,47.5) -- cycle ;
%Shape: Polygon Curved [id:ds05760046522444928] 
\draw  [color={rgb, 255:red, 255; green, 255; blue, 255 }  ,draw opacity=1 ][fill={rgb, 255:red, 74; green, 144; blue, 226 }  ,fill opacity=0.29 ] (598,415.5) .. controls (574,469.5) and (454,524.5) .. (440,448.5) .. controls (426,372.5) and (534,310.5) .. (577,313.5) .. controls (620,316.5) and (622,361.5) .. (598,415.5) -- cycle ;
%Curve Lines [id:da4918581723265819] 
\draw  [dash pattern={on 4.5pt off 4.5pt}]  (129,197.5) .. controls (237,190.5) and (251,100.5) .. (243,59.5) ;
%Curve Lines [id:da8272831274650023] 
\draw  [dash pattern={on 4.5pt off 4.5pt}]  (151,211.5) .. controls (259,204.5) and (273,114.5) .. (265,73.5) ;
%Curve Lines [id:da42909851438111657] 
\draw  [dash pattern={on 4.5pt off 4.5pt}]  (420,432.5) .. controls (408,356.5) and (516,292.5) .. (557,297.5) ;
%Curve Lines [id:da6590182211570761] 
\draw  [dash pattern={on 4.5pt off 4.5pt}]  (393,415.5) .. controls (381,339.5) and (489,275.5) .. (530,280.5) ;
%Curve Lines [id:da3759652936393266] 
\draw  [dash pattern={on 4.5pt off 4.5pt}]  (182,270.5) .. controls (295,282.5) and (347,202.5) .. (341,133.5) ;
%Curve Lines [id:da40909892149880933] 
\draw [color={rgb, 255:red, 25; green, 25; blue, 116 }  ,draw opacity=1 ][line width=1.5]    (138,104) .. controls (143,100.25) and (246,189) .. (235,202) .. controls (224,215) and (149,151.5) .. (141,161.5) .. controls (133,171.5) and (173,204) .. (189,216) .. controls (205,228) and (249,249.5) .. (269,234.5) .. controls (289,219.5) and (179,125.5) .. (204,100.5) .. controls (229,75.5) and (330.44,206.65) .. (356,241.5) .. controls (381.56,276.35) and (381,300.5) .. (365,307.5) .. controls (349,314.5) and (310,293.5) .. (311,327.5) .. controls (312,361.5) and (501,429.5) .. (526,404.5) .. controls (551,379.5) and (414,310.5) .. (425,294.5) .. controls (436,278.5) and (517,354.5) .. (533,343.5) .. controls (549,332.5) and (441,270.5) .. (466,245.5) ;
%Shape: Circle [id:dp1702270742307299] 
\draw  [fill={rgb, 255:red, 208; green, 2; blue, 27 }  ,fill opacity=1 ] (331,203.5) .. controls (331,201.01) and (328.99,199) .. (326.5,199) .. controls (324.01,199) and (322,201.01) .. (322,203.5) .. controls (322,205.99) and (324.01,208) .. (326.5,208) .. controls (328.99,208) and (331,205.99) .. (331,203.5) -- cycle ;
%Shape: Circle [id:dp8075386116355057] 
\draw  [fill={rgb, 255:red, 208; green, 2; blue, 27 }  ,fill opacity=1 ] (582,372.5) .. controls (582,370.01) and (579.99,368) .. (577.5,368) .. controls (575.01,368) and (573,370.01) .. (573,372.5) .. controls (573,374.99) and (575.01,377) .. (577.5,377) .. controls (579.99,377) and (582,374.99) .. (582,372.5) -- cycle ;
%Shape: Polygon Curved [id:ds9823904733154182] 
\draw  [draw opacity=0][fill={rgb, 255:red, 208; green, 2; blue, 27 }  ,fill opacity=0.1 ] (141,334.5) .. controls (47,219.5) and (99,204.5) .. (157,191.5) .. controls (215,178.5) and (221.81,148.94) .. (228.63,139.88) .. controls (235.44,130.81) and (252.5,95.5) .. (241,53.5) .. controls (229.5,11.5) and (323.56,1.5) .. (425,40.5) .. controls (526.44,79.5) and (598,244.5) .. (628,325.5) .. controls (658,406.5) and (613,499.5) .. (508,501.5) .. controls (403,503.5) and (235,449.5) .. (141,334.5) -- cycle ;

% Text Node
\draw (100,62.4) node [anchor=north west][inner sep=0.75pt]    {\huge $C_v\ =\ C_v( 0)$};
% Text Node
\draw (456,422) node [anchor=north west][inner sep=0.75pt]    {\huge $C_u\ =\ C_u( 0)$};
% Text Node
\draw (72,185) node [anchor=north west][inner sep=0.75pt]    {\huge $C_v(1)$};
% Text Node
\draw (97,204.4) node [anchor=north west][inner sep=0.75pt]    {\huge $C_v(2)$};
% Text Node
\draw (560,291.4) node [anchor=north west][inner sep=0.75pt]    {\huge $C_u(1)$};
% Text Node
\draw (533,267.4) node [anchor=north west][inner sep=0.75pt]    {\huge $C_u(2)$};
% Text Node
\draw (422,171) node [anchor=north west][inner sep=0.75pt]    {\Huge $V_{1}$};
% Text Node
\draw (118,261.4) node [anchor=north west][inner sep=0.75pt]    {\huge $C_v(d_{v})$};
% Text Node
\draw (283,104.4) node [anchor=north west][inner sep=0.75pt]    {\Huge $\textcolor[rgb]{0.1,0.1,0.45}{ P_{\text{end}}}$};
% Text Node
\draw (320,377.7) node [anchor=north west][inner sep=0.75pt]    {\Huge $\textcolor[rgb]{0.1,0.1,0.45}{P_{\text{start}}}$};
% Text Node
\draw (342,186.4) node [anchor=north west][inner sep=0.75pt]    {\huge $\textcolor[rgb]{0.82,0.01,0.11}{v}$};
% Text Node
\draw (585,353.4) node [anchor=north west][inner sep=0.75pt]    {\huge $\textcolor[rgb]{0.82,0.01,0.11}{u}$};
% Text Node
\draw (130,380.380) node [anchor=north west][inner sep=0.75pt]    {\Huge $\textcolor[rgb]{0.82,0.01,0.11}{V_{\preccurlyeq u}}$};

\end{tikzpicture}
        }
        \caption{The decomposition of an induced path in two consecutive paths $P_{\text{start}}$ and $P_{\text{end}}$ with the properties mentioned in the proof of Theorem~\ref{thm:p_4k_free_subquadratic}.}
        \label{fig:P_start P_end}
        
    \end{figure}
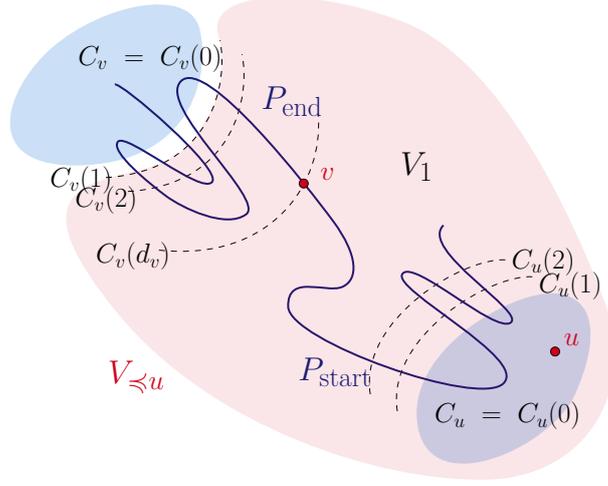
       
    \smallskip
    Conversely, assume that $G$ does not have an induced path of length $4k-1$. Let us show that there exists an assignment of the certificates such that every vertex accepts. This assignment is the following one: the prover attributes the certificates such that no vertex reject in the local computation schemes of  Theorems~\ref{thm:computation scheme for T_G} and~\ref{thm:computation scheme for g_<u}, and gives the correct value of $\lp$ to every vertex. With such a certificate, it is straightforward that a vertex can not reject at steps (i), (ii) and~(iii).
    
    Assume by contradiction that a vertex $u$ rejects at step (iv). Then, there exists an induced path $P$ of length $4k-1$ in $G_{\preccurlyeq u}$, such that the vertices in $P \setminus V_{\preccurlyeq u}$ are in distinct ECC$_2$'s. Since $G$ does not have an induced path of length $4k-1$, $P$ is not induced in $G$. It implies that there are two vertices of $P$ linked by an edge in $G$, an a non-edge in $G_{\preccurlyeq u}$. Thus, these two vertices are in $V \setminus V_{\preccurlyeq u}$ (because edges having an endpoint in $V_{\preccurlyeq u}$ are in $G_{\preccurlyeq u}$). So by definition of step~(iv), these two vertices are in distinct ECC$_2$'s. In particular, they can not be neighbors in $G$, which is a contradiction.

    Finally, assume by contradiction that a vertex $u$ rejects at step (v). It means that there exists an induced path $P_{\text{start}}$ in $G_{\preccurlyeq u}$, included in $V_{\preccurlyeq u}$, ending in some vertex $v$ satisfying the conditions of step (v). Since $P_{\text{start}} \subseteq V_{\preccurlyeq u}$, it is also induced in $G$. Moreover, there exists an induced path $P_v$ in $G$ of length $\ell(P_v)$ starting in $v$ and having all its other vertices in $\bigcup_{d < d_v}C_v(d)$. Since $v$ is the only vertex in $P_{\text{start}} \cap (\bigcup_{d \leqslant d_v} C_v(d))$, then the concatenation of $P_{\text{start}}$ and $P_v$ is still an induced path in $G$. Moreover, it has length $\ell(P_{\text{start}}) + \ell(P_v) - 1 \geqslant 4k-1$. It is a contradiction.
    Thus, no vertex rejects at step~(v), so all the vertices accept at step (vi), which concludes the proof.
\end{proof}

\subsection{Upper bound for general graphs of size $4k-1$ in $\Tilde{O}(n^{3/2})$}
\label{sec:H_4k}

\thmHFourkFreeSubquadratic*
%\begin{theorem}
%    \label{thm:certification of H-free, |H|=4k-1}
%   Let $H$ be a graph on at most $4k-1$ vertices. We can certify $H$-free graphs by looking at distance $k$ with certificates of size $O(k \cdot 2^{4k} \cdot n + n^{3/2} \log^2 n)$.
%\end{theorem}

Before proving Theorem~\ref{thm:H_4k_free_subquadratic}, let us give some informal intuition on the proof.
The general idea is the same as in the proof of Theorem~\ref{thm:p_4k_free_subquadratic}: first, the prover uses the local computation schemes of Theorems~\ref{thm:computation scheme for T_G} and~\ref{thm:computation scheme for g_<u} with $\varepsilon = \frac{1}{2}$ (such that each vertex $u$ computes $G_{\preccurlyeq u}$ and $T_G$), and then it adds some additional information in the certificates.
The main difference is the following one. Theorem~\ref{thm:p_4k_free_subquadratic} is the particular case where the graph $H$ is a path. In this case, the information the prover adds in the certificates is, for each vertex $v \in V$ close to some $C_v \in \cE_2$, the length of the longest induced path $P_v$ starting from $v$ and having all its other vertices closer from $C_v$ than $v$. This enables a vertex $u$ in an ECC$_2$ $C_u \neq C_v$ to check, for every induced path $P$ in $G_{\preccurlyeq u}$, if it can be extended into a path of length $4k-1$ using $P_v$.
Here, for an arbitrary graph $H$, we generalize it: the information the prover will add is, intuitively, for each vertex $v$ close from an ECC$_2$ $C_v$, and for every induced subgraph $H'$ of $H$, if $H'$ is induced in the part of $G$ consisting in $v$ and all the vertices closer from $C_v$ than $v$. With this information, every vertex $u$ in an ECC$_2$ $C_u \neq C_v$ will be able to detect if some induced graph in $G_{\preccurlyeq u}$ can be extended into $H$ using $H'$.

\paragraph*{Pointed graphs.}
Let us introduce some definitions about pointed graphs that will be useful in the proof of Theorem~\ref{thm:H_4k_free_subquadratic}.
A \emph{pointed graph} is a tuple $(H, S)$ where $H$ is a graph and $S \subseteq V(H)$ is a set of \emph{pointed} vertices.
If $(H_1,S_1)$, $(H_2,S_2)$ are two pointed graphs, their \emph{disjoint union} is the pointed graph $(H_1 \cup H_2, S_1 \cup S_2)$, where $H_1 \cup H_2$ is the graph obtained by taking a copy of $H_1$ and a copy of $H_2$, and no edge between them.
Finally, let $G$ be a graph, $H$ an induced subgraph of $G$ and $S \subseteq V(H)$. The \emph{complement} of $(H,S)$ in $G$ is the pointed subgraph $(\overline{H},S)$ where $\overline{H}$ is the subgraph of $G$ induced by $(V(G) \setminus V(H)) \cup S$.
%If $G$ is a graph, $u$ a vertex of $G$, and $(H,\{h\})$ a pointed graph (with only one pointed vertex $h$), we say that $H$ is an \emph{induced subgraph of $G$ pointing at $u$} if 

\begin{proof}[Proof of Theorem~\ref{thm:H_4k_free_subquadratic}.]

Let $H$ be a graph with $4k-1$ vertices. Let us describe a certification scheme for $H$-free graphs with verification radius~$k$ and certificates of size \mbox{$O(k \cdot 2^{4k} \cdot n + n^{3/2} \log^2 n)$}.

\medskip{}
\textbf{Certification.}
The certification of the prover is the following. First,  it gives to each vertex its certificates in the local computation schemes given by Theorems~\ref{thm:computation scheme for T_G} and~\ref{thm:computation scheme for g_<u}, with $\varepsilon = \frac{1}{2}$. This part of the certificate has size $O(n^{3/2} \log^2 n)$.

Then, the prover adds some additional information in the certificates, in a field which will be identical in the certificates of all the vertices. We will denote this field by $\htab$, and it is defined as follows.
Let $\mathcal{P}(H)$ be the set containing all the pointed graphs $(H',\{h\})$, where $h \in V(H)$, and $H'$ is a subgraph of $H$ induced by $\{h\}$ and an union of connected components of $H \setminus \{h\}$. An example is shown on Figure~\ref{fig:P(H)}.
The field $\htab$ contains a table which has at most $n$ rows and $(4k-1) \cdot 2^{4k-1}$ columns. The rows are indexed by the identifier of the vertices which are at distance at most $k-1$ from $V_2$, and the columns are indexed by the pointed graphs $(H',\{h\}) \in \mathcal{P}(H)$.
Let $v \in V$, such that $v$ is at distance $d_v \leqslant k-1$ from some $C_v \in \cE_2$. $\htab[v, (H',\{h\})]$ is equal to $1$ if $H'$ is an induced subgraph of $\bigcup_{d \leqslant d_v} C_v(d)$, where all the vertices in $H' \setminus \{h\}$ are mapped in $\bigcup_{d < d_v} C_v(d)$, and $h$ is mapped to $v$. Otherwise, $\htab[v, (H,\{h\})]$ is equal to $0$.
An example with the pointed graph $H'$ of Figure~\ref{fig:P(H)} is shown on Figure~\ref{fig:1 in Htab}.

\begin{figure}
    \centering
    \scalebox{0.8}{
    \begin{tikzpicture}[x=0.75pt,y=0.75pt,yscale=-1,xscale=1]
%uncomment if require: \path (0,455); %set diagram left start at 0, and has height of 455

%Straight Lines [id:da7265888673829374] 
\draw    (116.25,171) -- (154.5,171) ;
%Straight Lines [id:da35674036749504434] 
\draw    (185.5,153) -- (223.75,153) ;
%Straight Lines [id:da15152858734427255] 
\draw    (154.5,171) -- (185.5,153) ;
%Straight Lines [id:da3201288954072503] 
\draw    (154.5,171) -- (185.5,189) ;
%Straight Lines [id:da267813608326131] 
\draw    (116.25,136.75) -- (116.25,171) ;
%Straight Lines [id:da21683142101676856] 
\draw    (96,107) -- (116.25,136.75) ;
%Straight Lines [id:da658377041087818] 
\draw    (116.25,77.25) -- (136.5,107) ;
%Straight Lines [id:da20031527446228414] 
\draw    (116.25,77.25) -- (96,107) ;
%Straight Lines [id:da8542870799738029] 
\draw    (136.5,107) -- (116.25,136.75) ;
%Straight Lines [id:da15077453530585982] 
\draw    (78,171) -- (116.25,171) ;
%Straight Lines [id:da7133903418207117] 
\draw    (116.25,171) -- (116.25,205.25) ;
%Straight Lines [id:da56261859217099] 
\draw    (116.25,205.25) -- (99,233) ;
%Straight Lines [id:da1258437903138877] 
\draw    (116.25,205.25) -- (133.5,233) ;
%Straight Lines [id:da8147255036098099] 
\draw    (133.5,233) -- (99,233) ;
%Shape: Circle [id:dp5716905702494233] 
\draw  [fill={rgb, 255:red, 255; green, 255; blue, 255 }  ,fill opacity=1 ] (128.25,233) .. controls (128.25,230.1) and (130.6,227.75) .. (133.5,227.75) .. controls (136.4,227.75) and (138.75,230.1) .. (138.75,233) .. controls (138.75,235.9) and (136.4,238.25) .. (133.5,238.25) .. controls (130.6,238.25) and (128.25,235.9) .. (128.25,233) -- cycle ;
%Shape: Circle [id:dp962405070416441] 
\draw  [fill={rgb, 255:red, 255; green, 255; blue, 255 }  ,fill opacity=1 ] (93.75,233) .. controls (93.75,230.1) and (96.1,227.75) .. (99,227.75) .. controls (101.9,227.75) and (104.25,230.1) .. (104.25,233) .. controls (104.25,235.9) and (101.9,238.25) .. (99,238.25) .. controls (96.1,238.25) and (93.75,235.9) .. (93.75,233) -- cycle ;
%Shape: Circle [id:dp4203180331032972] 
\draw  [fill={rgb, 255:red, 255; green, 255; blue, 255 }  ,fill opacity=1 ] (111,205.25) .. controls (111,202.35) and (113.35,200) .. (116.25,200) .. controls (119.15,200) and (121.5,202.35) .. (121.5,205.25) .. controls (121.5,208.15) and (119.15,210.5) .. (116.25,210.5) .. controls (113.35,210.5) and (111,208.15) .. (111,205.25) -- cycle ;
%Shape: Circle [id:dp7799573465341435] 
\draw  [fill={rgb, 255:red, 255; green, 255; blue, 255 }  ,fill opacity=1 ] (149.25,171) .. controls (149.25,168.1) and (151.6,165.75) .. (154.5,165.75) .. controls (157.4,165.75) and (159.75,168.1) .. (159.75,171) .. controls (159.75,173.9) and (157.4,176.25) .. (154.5,176.25) .. controls (151.6,176.25) and (149.25,173.9) .. (149.25,171) -- cycle ;
%Shape: Circle [id:dp8205401352023108] 
\draw  [fill={rgb, 255:red, 255; green, 255; blue, 255 }  ,fill opacity=1 ] (218.5,153) .. controls (218.5,150.1) and (220.85,147.75) .. (223.75,147.75) .. controls (226.65,147.75) and (229,150.1) .. (229,153) .. controls (229,155.9) and (226.65,158.25) .. (223.75,158.25) .. controls (220.85,158.25) and (218.5,155.9) .. (218.5,153) -- cycle ;
%Shape: Circle [id:dp804837240360494] 
\draw  [fill={rgb, 255:red, 255; green, 255; blue, 255 }  ,fill opacity=1 ] (180.25,153) .. controls (180.25,150.1) and (182.6,147.75) .. (185.5,147.75) .. controls (188.4,147.75) and (190.75,150.1) .. (190.75,153) .. controls (190.75,155.9) and (188.4,158.25) .. (185.5,158.25) .. controls (182.6,158.25) and (180.25,155.9) .. (180.25,153) -- cycle ;
%Shape: Circle [id:dp7684776824058503] 
\draw  [fill={rgb, 255:red, 255; green, 255; blue, 255 }  ,fill opacity=1 ] (180.25,189) .. controls (180.25,186.1) and (182.6,183.75) .. (185.5,183.75) .. controls (188.4,183.75) and (190.75,186.1) .. (190.75,189) .. controls (190.75,191.9) and (188.4,194.25) .. (185.5,194.25) .. controls (182.6,194.25) and (180.25,191.9) .. (180.25,189) -- cycle ;
%Shape: Circle [id:dp489008982456553] 
\draw  [fill={rgb, 255:red, 255; green, 255; blue, 255 }  ,fill opacity=1 ] (111,77.25) .. controls (111,74.35) and (113.35,72) .. (116.25,72) .. controls (119.15,72) and (121.5,74.35) .. (121.5,77.25) .. controls (121.5,80.15) and (119.15,82.5) .. (116.25,82.5) .. controls (113.35,82.5) and (111,80.15) .. (111,77.25) -- cycle ;
%Shape: Circle [id:dp9182807098220264] 
\draw  [fill={rgb, 255:red, 255; green, 255; blue, 255 }  ,fill opacity=1 ] (131.25,107) .. controls (131.25,104.1) and (133.6,101.75) .. (136.5,101.75) .. controls (139.4,101.75) and (141.75,104.1) .. (141.75,107) .. controls (141.75,109.9) and (139.4,112.25) .. (136.5,112.25) .. controls (133.6,112.25) and (131.25,109.9) .. (131.25,107) -- cycle ;
%Shape: Circle [id:dp597036743876556] 
\draw  [fill={rgb, 255:red, 255; green, 255; blue, 255 }  ,fill opacity=1 ] (90.75,107) .. controls (90.75,104.1) and (93.1,101.75) .. (96,101.75) .. controls (98.9,101.75) and (101.25,104.1) .. (101.25,107) .. controls (101.25,109.9) and (98.9,112.25) .. (96,112.25) .. controls (93.1,112.25) and (90.75,109.9) .. (90.75,107) -- cycle ;
%Shape: Circle [id:dp2937199132958578] 
\draw  [fill={rgb, 255:red, 255; green, 255; blue, 255 }  ,fill opacity=1 ] (111,136.75) .. controls (111,133.85) and (113.35,131.5) .. (116.25,131.5) .. controls (119.15,131.5) and (121.5,133.85) .. (121.5,136.75) .. controls (121.5,139.65) and (119.15,142) .. (116.25,142) .. controls (113.35,142) and (111,139.65) .. (111,136.75) -- cycle ;
%Shape: Circle [id:dp4899837026792334] 
\draw  [fill={rgb, 255:red, 255; green, 255; blue, 255 }  ,fill opacity=1 ] (111,171) .. controls (111,168.1) and (113.35,165.75) .. (116.25,165.75) .. controls (119.15,165.75) and (121.5,168.1) .. (121.5,171) .. controls (121.5,173.9) and (119.15,176.25) .. (116.25,176.25) .. controls (113.35,176.25) and (111,173.9) .. (111,171) -- cycle ;
%Shape: Circle [id:dp6760608628665208] 
\draw  [fill={rgb, 255:red, 255; green, 255; blue, 255 }  ,fill opacity=1 ] (72.75,171) .. controls (72.75,168.1) and (75.1,165.75) .. (78,165.75) .. controls (80.9,165.75) and (83.25,168.1) .. (83.25,171) .. controls (83.25,173.9) and (80.9,176.25) .. (78,176.25) .. controls (75.1,176.25) and (72.75,173.9) .. (72.75,171) -- cycle ;
%Straight Lines [id:da8509514124259059] 
\draw    (362.25,137.75) -- (362.25,172) ;
%Straight Lines [id:da7583265276355078] 
\draw    (342,108) -- (362.25,137.75) ;
%Straight Lines [id:da0792622522442441] 
\draw    (362.25,78.25) -- (382.5,108) ;
%Straight Lines [id:da3567997208083018] 
\draw    (362.25,78.25) -- (342,108) ;
%Straight Lines [id:da25595889274568007] 
\draw    (382.5,108) -- (362.25,137.75) ;
%Straight Lines [id:da8996134596293575] 
\draw    (324,172) -- (362.25,172) ;
%Shape: Circle [id:dp9527859349700933] 
\draw  [fill={rgb, 255:red, 255; green, 255; blue, 255 }  ,fill opacity=1 ] (357,78.25) .. controls (357,75.35) and (359.35,73) .. (362.25,73) .. controls (365.15,73) and (367.5,75.35) .. (367.5,78.25) .. controls (367.5,81.15) and (365.15,83.5) .. (362.25,83.5) .. controls (359.35,83.5) and (357,81.15) .. (357,78.25) -- cycle ;
%Shape: Circle [id:dp340856731110425] 
\draw  [fill={rgb, 255:red, 255; green, 255; blue, 255 }  ,fill opacity=1 ] (377.25,108) .. controls (377.25,105.1) and (379.6,102.75) .. (382.5,102.75) .. controls (385.4,102.75) and (387.75,105.1) .. (387.75,108) .. controls (387.75,110.9) and (385.4,113.25) .. (382.5,113.25) .. controls (379.6,113.25) and (377.25,110.9) .. (377.25,108) -- cycle ;
%Shape: Circle [id:dp11288399298555751] 
\draw  [fill={rgb, 255:red, 255; green, 255; blue, 255 }  ,fill opacity=1 ] (336.75,108) .. controls (336.75,105.1) and (339.1,102.75) .. (342,102.75) .. controls (344.9,102.75) and (347.25,105.1) .. (347.25,108) .. controls (347.25,110.9) and (344.9,113.25) .. (342,113.25) .. controls (339.1,113.25) and (336.75,110.9) .. (336.75,108) -- cycle ;
%Shape: Circle [id:dp15534670562282749] 
\draw  [fill={rgb, 255:red, 255; green, 255; blue, 255 }  ,fill opacity=1 ] (357,137.75) .. controls (357,134.85) and (359.35,132.5) .. (362.25,132.5) .. controls (365.15,132.5) and (367.5,134.85) .. (367.5,137.75) .. controls (367.5,140.65) and (365.15,143) .. (362.25,143) .. controls (359.35,143) and (357,140.65) .. (357,137.75) -- cycle ;
%Shape: Circle [id:dp8973927387943738] 
\draw  [fill={rgb, 255:red, 245; green, 166; blue, 35 }  ,fill opacity=1 ][line width=1.5]  (357,172) .. controls (357,169.1) and (359.35,166.75) .. (362.25,166.75) .. controls (365.15,166.75) and (367.5,169.1) .. (367.5,172) .. controls (367.5,174.9) and (365.15,177.25) .. (362.25,177.25) .. controls (359.35,177.25) and (357,174.9) .. (357,172) -- cycle ;
%Shape: Circle [id:dp1108968900956131] 
\draw  [fill={rgb, 255:red, 255; green, 255; blue, 255 }  ,fill opacity=1 ] (318.75,172) .. controls (318.75,169.1) and (321.1,166.75) .. (324,166.75) .. controls (326.9,166.75) and (329.25,169.1) .. (329.25,172) .. controls (329.25,174.9) and (326.9,177.25) .. (324,177.25) .. controls (321.1,177.25) and (318.75,174.9) .. (318.75,172) -- cycle ;

% Text Node
\draw (57,77.4) node [anchor=north west][inner sep=0.75pt]    {\LARGE $H$};
% Text Node
\draw (305,75.4) node [anchor=north west][inner sep=0.75pt]    {\LARGE $H'$};
% Text Node
\draw (369.5,175.4) node [anchor=north west][inner sep=0.75pt]    {\LARGE $h$};
% Text Node
\draw (123.5,174.4) node [anchor=north west][inner sep=0.75pt]    {\LARGE $h$};

\end{tikzpicture}

    }
    \caption{An example of a graph $H$ and a pointed graph $(H',\{h\}) \in \mathcal{P}(H)$.}
    \label{fig:P(H)}
\end{figure}
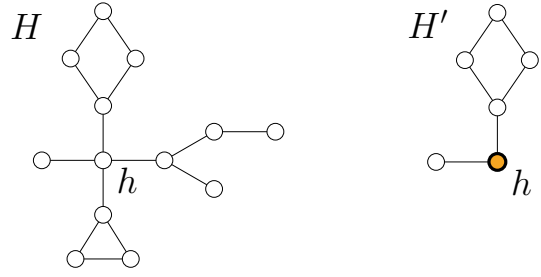

\begin{figure}
    \centering
    \scalebox{0.8}{
    \begin{tikzpicture}[x=0.75pt,y=0.75pt,yscale=-1,xscale=1]
%uncomment if require: \path (0,491); %set diagram left start at 0, and has height of 491

%Shape: Polygon Curved [id:ds6124722764356915] 
\draw  [color={rgb, 255:red, 255; green, 255; blue, 255 }  ,draw opacity=1 ][fill={rgb, 255:red, 74; green, 144; blue, 226 }  ,fill opacity=0.29 ] (255.5,145) .. controls (289.5,187) and (281.5,241) .. (211.5,243) .. controls (141.5,245) and (126.5,196) .. (143.5,159) .. controls (160.5,122) and (221.5,103) .. (255.5,145) -- cycle ;
%Curve Lines [id:da094690386662583] 
\draw  [dash pattern={on 4.5pt off 4.5pt}]  (324.5,144) .. controls (384.5,222) and (373.5,296) .. (256.5,306) ;
%Straight Lines [id:da9204514962896702] 
\draw    (333.5,223) -- (359.25,246) ;
%Straight Lines [id:da36810950207754534] 
\draw    (310.5,230) -- (333.5,223) ;
%Straight Lines [id:da3619254623332827] 
\draw    (290.5,213) -- (315.5,204) ;
%Straight Lines [id:da4437801593486048] 
\draw    (310.5,230) -- (290.5,213) ;
%Straight Lines [id:da8753802084247936] 
\draw    (315.5,204) -- (333.5,223) ;
%Straight Lines [id:da38075437221010056] 
\draw    (322.5,253) -- (359.25,246) ;
%Shape: Circle [id:dp06575168944141674] 
\draw  [fill={rgb, 255:red, 255; green, 255; blue, 255 }  ,fill opacity=1 ] (285.25,213) .. controls (285.25,210.1) and (287.6,207.75) .. (290.5,207.75) .. controls (293.4,207.75) and (295.75,210.1) .. (295.75,213) .. controls (295.75,215.9) and (293.4,218.25) .. (290.5,218.25) .. controls (287.6,218.25) and (285.25,215.9) .. (285.25,213) -- cycle ;
%Shape: Circle [id:dp36205544355770125] 
\draw  [fill={rgb, 255:red, 255; green, 255; blue, 255 }  ,fill opacity=1 ] (310.25,204) .. controls (310.25,201.1) and (312.6,198.75) .. (315.5,198.75) .. controls (318.4,198.75) and (320.75,201.1) .. (320.75,204) .. controls (320.75,206.9) and (318.4,209.25) .. (315.5,209.25) .. controls (312.6,209.25) and (310.25,206.9) .. (310.25,204) -- cycle ;
%Shape: Circle [id:dp4411389403460684] 
\draw  [fill={rgb, 255:red, 255; green, 255; blue, 255 }  ,fill opacity=1 ] (305.25,230) .. controls (305.25,227.1) and (307.6,224.75) .. (310.5,224.75) .. controls (313.4,224.75) and (315.75,227.1) .. (315.75,230) .. controls (315.75,232.9) and (313.4,235.25) .. (310.5,235.25) .. controls (307.6,235.25) and (305.25,232.9) .. (305.25,230) -- cycle ;
%Shape: Circle [id:dp02089630948170218] 
\draw  [fill={rgb, 255:red, 255; green, 255; blue, 255 }  ,fill opacity=1 ] (328.25,223) .. controls (328.25,220.1) and (330.6,217.75) .. (333.5,217.75) .. controls (336.4,217.75) and (338.75,220.1) .. (338.75,223) .. controls (338.75,225.9) and (336.4,228.25) .. (333.5,228.25) .. controls (330.6,228.25) and (328.25,225.9) .. (328.25,223) -- cycle ;
%Shape: Circle [id:dp8741168690410183] 
\draw  [fill={rgb, 255:red, 245; green, 166; blue, 35 }  ,fill opacity=1 ][line width=1.5]  (354,246) .. controls (354,243.1) and (356.35,240.75) .. (359.25,240.75) .. controls (362.15,240.75) and (364.5,243.1) .. (364.5,246) .. controls (364.5,248.9) and (362.15,251.25) .. (359.25,251.25) .. controls (356.35,251.25) and (354,248.9) .. (354,246) -- cycle ;
%Shape: Circle [id:dp8999850031511751] 
\draw  [fill={rgb, 255:red, 255; green, 255; blue, 255 }  ,fill opacity=1 ] (317.25,253) .. controls (317.25,250.1) and (319.6,247.75) .. (322.5,247.75) .. controls (325.4,247.75) and (327.75,250.1) .. (327.75,253) .. controls (327.75,255.9) and (325.4,258.25) .. (322.5,258.25) .. controls (319.6,258.25) and (317.25,255.9) .. (317.25,253) -- cycle ;

% Text Node
\draw (366.5,249.4) node [anchor=north west][inner sep=0.75pt]    {\LARGE $v$};
% Text Node
\draw (132,229.4) node [anchor=north west][inner sep=0.75pt]    {\LARGE $C_{v}$};
% Text Node
\draw (203,298.4) node [anchor=north west][inner sep=0.75pt]    {\LARGE $C_{v}(d_{v})$};

\end{tikzpicture}

    }
    \caption{The illustration of $\htab[v,(H',\{h\})]=1$, where $(H',\{h\})$ is the pointed graph of Figure~\ref{fig:P(H)}.}
    \label{fig:1 in Htab}
\end{figure}
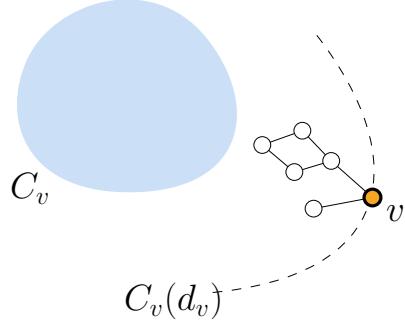

Since the table~$\htab$ has at most $n$ rows and $(4k-1) \cdot 2^{4k-1}$ columns, with each entry of constant size, its has size $O(k \cdot 2^{4k} \cdot n)$. The overall size of the certificate is thus $O(k \cdot 2^{4k} \cdot n + n^{3/2} \log^2 n)$.

\medskip{}
\textbf{Verification.}
The verification of each vertex $u \in V$ consists in the following.

\begin{enumerate}[(i)]
    \item First, $u$ applies the verification of the local computation schemes of Theorems~\ref{thm:computation scheme for T_G} and~\ref{thm:computation scheme for g_<u}. If no vertex rejects during this first verification phase, then every vertex $u$ computed~$T_G$ and~$G_{\preccurlyeq u}$.

    \item Then, $u$ checks that $\htab$ is the same in its certificate and in the certificate of all its neighbors. If it is not the case, $u$ rejects.

    If no vertex rejects, $\htab$ is the same in the certificates of all the vertices in $G$. The next step will consist in checking its correctness.

    \item If $u \in V_2$, it does the following. Let $C_u$ be its ECC$_2$. We have $\bigcup_{0 \leqslant d \leqslant k-1} C_u(d) \subseteq V_{\preccurlyeq u}$.
    Thus, for every $v \in \bigcup_{0 \leqslant d \leqslant k-1} C_u(d)$ and every $(H,\{h\}) \in \mathcal{P}(H)$, $u$ can check if $\htab[v,(H',\{h\})]$ is correct. If it is not correct, $u$ rejects.

    If no vertex rejects at this point, then $\htab$ is correct.

    \item If $G_{\preccurlyeq u}$ has an induced $H$, such that all the vertices in $H \setminus V_{\preccurlyeq u}$ are in distinct ECC$_2$'s, then $u$ rejects (in particular, this is the case if $H \subseteq V_{\preccurlyeq u}$).

    The last steps are similar to step~(v) in the verification of Theorem~\ref{thm:p_4k_free_subquadratic}: $u$ will check that it can not obtain a graph $H$ by extending a subgraph of $G_{\preccurlyeq u}$ into $H$, using a subgraph $H'$ such that $\htab[v,(H',\{h\})]=1$ for some vertices $v$ and $h$.
    
    \item If $u \in V_2$, let $C_u \in \cE_2$ such that $u \in C_u$. The vertex $u$ rejects if there exists two pointed graphs $(H_\text{start},\{h\}), (H_\text{end},\{h\}) \in \mathcal{P}(H)$ which are complementary of each other in $H$, and a vertex $v \in C_v(d_v)$ for some $C_v \neq C_u$, $d_v \leqslant k-1$, such that the following conditions are satisfied:
    \begin{itemize}
        \item $H_\text{start}$ is an induced subgraph of $G_{\preccurlyeq u}$, included in $V_{\preccurlyeq u}$, which contains at least one vertex from $C_u$
        \item $h$ is mapped to $v$, and it is the only vertex in $H_\text{start} \cap (\bigcup_{d \leqslant d_v} C_v(d))$
        \item $\htab[v, (H_\text{end}, \{h\})]=1$, where $(H_\text{end}, \{h\})$ is the complementary of $(H_\text{start},\{h\})$ in~$H$
    \end{itemize}
    An illustration is shown on Figure~\ref{fig:H-free verification (v)}.

    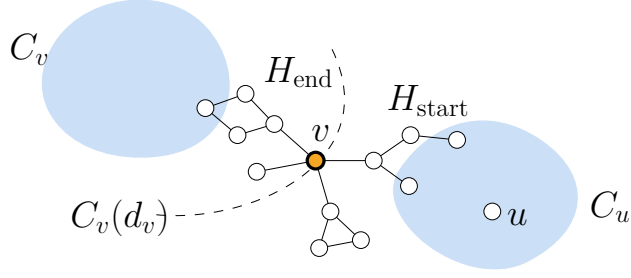
\begin{figure}
        \centering
        \scalebox{0.8}{
        \begin{tikzpicture}[x=0.75pt,y=0.75pt,yscale=-1,xscale=1]
%uncomment if require: \path (0,491); %set diagram left start at 0, and has height of 491

%Shape: Polygon Curved [id:ds23794127928215691] 
\draw  [color={rgb, 255:red, 255; green, 255; blue, 255 }  ,draw opacity=1 ][fill={rgb, 255:red, 74; green, 144; blue, 226 }  ,fill opacity=0.29 ] (367.5,171) .. controls (400.5,194) and (397.5,228) .. (358.5,248) .. controls (319.5,268) and (257.5,220) .. (278.5,189) .. controls (299.5,158) and (334.5,148) .. (367.5,171) -- cycle ;
%Straight Lines [id:da38873335310317736] 
\draw    (226.25,184) -- (235.5,216) ;
%Straight Lines [id:da168272941219324] 
\draw    (254.5,234) -- (235.5,216) ;
%Straight Lines [id:da6895637096866191] 
\draw    (228.5,239) -- (235.5,216) ;
%Straight Lines [id:da8101582506841856] 
\draw    (254.5,234) -- (228.5,239) ;
%Straight Lines [id:da4612349535539202] 
\draw    (226.25,184) -- (262.5,184) ;
%Straight Lines [id:da1010690959315298] 
\draw    (285.5,168) -- (262.5,184) ;
%Straight Lines [id:da036773451674090074] 
\draw    (262.5,184) -- (284.5,200) ;
%Straight Lines [id:da8612505275489301] 
\draw    (285.5,168) -- (314.5,171) ;
%Shape: Polygon Curved [id:ds4230901609599499] 
\draw  [color={rgb, 255:red, 255; green, 255; blue, 255 }  ,draw opacity=1 ][fill={rgb, 255:red, 74; green, 144; blue, 226 }  ,fill opacity=0.29 ] (158.5,100) .. controls (184.5,129) and (176.5,177) .. (126.5,182) .. controls (76.5,187) and (42.5,153) .. (59.5,116) .. controls (76.5,79) and (132.5,71) .. (158.5,100) -- cycle ;
%Curve Lines [id:da5080213665768468] 
\draw  [dash pattern={on 4.5pt off 4.5pt}]  (236.5,114) .. controls (270.5,185) and (185.5,224) .. (126.5,218) ;
%Straight Lines [id:da5595208995486274] 
\draw    (200.5,161) -- (226.25,184) ;
%Straight Lines [id:da830012142804227] 
\draw    (177.5,168) -- (200.5,161) ;
%Straight Lines [id:da7836104184931317] 
\draw    (157.5,151) -- (182.5,142) ;
%Straight Lines [id:da449268091989256] 
\draw    (177.5,168) -- (157.5,151) ;
%Straight Lines [id:da028055949006332348] 
\draw    (182.5,142) -- (200.5,161) ;
%Straight Lines [id:da6025892866058733] 
\draw    (189.5,191) -- (226.25,184) ;
%Shape: Circle [id:dp0015705961774659105] 
\draw  [fill={rgb, 255:red, 255; green, 255; blue, 255 }  ,fill opacity=1 ] (152.25,151) .. controls (152.25,148.1) and (154.6,145.75) .. (157.5,145.75) .. controls (160.4,145.75) and (162.75,148.1) .. (162.75,151) .. controls (162.75,153.9) and (160.4,156.25) .. (157.5,156.25) .. controls (154.6,156.25) and (152.25,153.9) .. (152.25,151) -- cycle ;
%Shape: Circle [id:dp07798369648368042] 
\draw  [fill={rgb, 255:red, 255; green, 255; blue, 255 }  ,fill opacity=1 ] (177.25,142) .. controls (177.25,139.1) and (179.6,136.75) .. (182.5,136.75) .. controls (185.4,136.75) and (187.75,139.1) .. (187.75,142) .. controls (187.75,144.9) and (185.4,147.25) .. (182.5,147.25) .. controls (179.6,147.25) and (177.25,144.9) .. (177.25,142) -- cycle ;
%Shape: Circle [id:dp9302301968537122] 
\draw  [fill={rgb, 255:red, 255; green, 255; blue, 255 }  ,fill opacity=1 ] (172.25,168) .. controls (172.25,165.1) and (174.6,162.75) .. (177.5,162.75) .. controls (180.4,162.75) and (182.75,165.1) .. (182.75,168) .. controls (182.75,170.9) and (180.4,173.25) .. (177.5,173.25) .. controls (174.6,173.25) and (172.25,170.9) .. (172.25,168) -- cycle ;
%Shape: Circle [id:dp5759192903349604] 
\draw  [fill={rgb, 255:red, 255; green, 255; blue, 255 }  ,fill opacity=1 ] (195.25,161) .. controls (195.25,158.1) and (197.6,155.75) .. (200.5,155.75) .. controls (203.4,155.75) and (205.75,158.1) .. (205.75,161) .. controls (205.75,163.9) and (203.4,166.25) .. (200.5,166.25) .. controls (197.6,166.25) and (195.25,163.9) .. (195.25,161) -- cycle ;
%Shape: Circle [id:dp1193474101667451] 
\draw  [fill={rgb, 255:red, 245; green, 166; blue, 35 }  ,fill opacity=1 ][line width=1.5]  (221,184) .. controls (221,181.1) and (223.35,178.75) .. (226.25,178.75) .. controls (229.15,178.75) and (231.5,181.1) .. (231.5,184) .. controls (231.5,186.9) and (229.15,189.25) .. (226.25,189.25) .. controls (223.35,189.25) and (221,186.9) .. (221,184) -- cycle ;
%Shape: Circle [id:dp8791768145310446] 
\draw  [fill={rgb, 255:red, 255; green, 255; blue, 255 }  ,fill opacity=1 ] (184.25,191) .. controls (184.25,188.1) and (186.6,185.75) .. (189.5,185.75) .. controls (192.4,185.75) and (194.75,188.1) .. (194.75,191) .. controls (194.75,193.9) and (192.4,196.25) .. (189.5,196.25) .. controls (186.6,196.25) and (184.25,193.9) .. (184.25,191) -- cycle ;
%Shape: Circle [id:dp734404729366549] 
\draw  [fill={rgb, 255:red, 255; green, 255; blue, 255 }  ,fill opacity=1 ] (249.25,234) .. controls (249.25,231.1) and (251.6,228.75) .. (254.5,228.75) .. controls (257.4,228.75) and (259.75,231.1) .. (259.75,234) .. controls (259.75,236.9) and (257.4,239.25) .. (254.5,239.25) .. controls (251.6,239.25) and (249.25,236.9) .. (249.25,234) -- cycle ;
%Shape: Circle [id:dp8432678954320222] 
\draw  [fill={rgb, 255:red, 255; green, 255; blue, 255 }  ,fill opacity=1 ] (223.25,239) .. controls (223.25,236.1) and (225.6,233.75) .. (228.5,233.75) .. controls (231.4,233.75) and (233.75,236.1) .. (233.75,239) .. controls (233.75,241.9) and (231.4,244.25) .. (228.5,244.25) .. controls (225.6,244.25) and (223.25,241.9) .. (223.25,239) -- cycle ;
%Shape: Circle [id:dp5205888409329529] 
\draw  [fill={rgb, 255:red, 255; green, 255; blue, 255 }  ,fill opacity=1 ] (230.25,216) .. controls (230.25,213.1) and (232.6,210.75) .. (235.5,210.75) .. controls (238.4,210.75) and (240.75,213.1) .. (240.75,216) .. controls (240.75,218.9) and (238.4,221.25) .. (235.5,221.25) .. controls (232.6,221.25) and (230.25,218.9) .. (230.25,216) -- cycle ;
%Shape: Circle [id:dp944280700607715] 
\draw  [fill={rgb, 255:red, 255; green, 255; blue, 255 }  ,fill opacity=1 ] (257.25,184) .. controls (257.25,181.1) and (259.6,178.75) .. (262.5,178.75) .. controls (265.4,178.75) and (267.75,181.1) .. (267.75,184) .. controls (267.75,186.9) and (265.4,189.25) .. (262.5,189.25) .. controls (259.6,189.25) and (257.25,186.9) .. (257.25,184) -- cycle ;
%Shape: Circle [id:dp027586853797893274] 
\draw  [fill={rgb, 255:red, 255; green, 255; blue, 255 }  ,fill opacity=1 ] (279.25,200) .. controls (279.25,197.1) and (281.6,194.75) .. (284.5,194.75) .. controls (287.4,194.75) and (289.75,197.1) .. (289.75,200) .. controls (289.75,202.9) and (287.4,205.25) .. (284.5,205.25) .. controls (281.6,205.25) and (279.25,202.9) .. (279.25,200) -- cycle ;
%Shape: Circle [id:dp6921937470164073] 
\draw  [fill={rgb, 255:red, 255; green, 255; blue, 255 }  ,fill opacity=1 ] (280.25,168) .. controls (280.25,165.1) and (282.6,162.75) .. (285.5,162.75) .. controls (288.4,162.75) and (290.75,165.1) .. (290.75,168) .. controls (290.75,170.9) and (288.4,173.25) .. (285.5,173.25) .. controls (282.6,173.25) and (280.25,170.9) .. (280.25,168) -- cycle ;
%Shape: Circle [id:dp4796123302201757] 
\draw  [fill={rgb, 255:red, 255; green, 255; blue, 255 }  ,fill opacity=1 ] (309.25,171) .. controls (309.25,168.1) and (311.6,165.75) .. (314.5,165.75) .. controls (317.4,165.75) and (319.75,168.1) .. (319.75,171) .. controls (319.75,173.9) and (317.4,176.25) .. (314.5,176.25) .. controls (311.6,176.25) and (309.25,173.9) .. (309.25,171) -- cycle ;
%Shape: Circle [id:dp6006109773785061] 
\draw  [fill={rgb, 255:red, 255; green, 255; blue, 255 }  ,fill opacity=1 ] (331.25,216) .. controls (331.25,213.1) and (333.6,210.75) .. (336.5,210.75) .. controls (339.4,210.75) and (341.75,213.1) .. (341.75,216) .. controls (341.75,218.9) and (339.4,221.25) .. (336.5,221.25) .. controls (333.6,221.25) and (331.25,218.9) .. (331.25,216) -- cycle ;

% Text Node
\draw (222.5,160.4) node [anchor=north west][inner sep=0.75pt]    {\LARGE $v$};
% Text Node
\draw (34,103.4) node [anchor=north west][inner sep=0.75pt]    {\LARGE $C_{v}$};
% Text Node
\draw (72,207.4) node [anchor=north west][inner sep=0.75pt]    {\LARGE $C_{v}( d_{v})$};
% Text Node
\draw (396,202.4) node [anchor=north west][inner sep=0.75pt]    {\LARGE $C_{u}$};
% Text Node
\draw (344.5,213.4) node [anchor=north west][inner sep=0.75pt]    {\LARGE $u$};
% Text Node
\draw (269.5,135.4) node [anchor=north west][inner sep=0.75pt]    {\LARGE $H_\text{start}$};
% Text Node
\draw (192.5,117.4) node [anchor=north west][inner sep=0.75pt]    {\LARGE $H_\text{end}$};

\end{tikzpicture}

        }
        \caption{Illustration of the condition making $u$ reject at step~(v) of the verification. Here, the graph $H$ is the one of Figure~\ref{fig:P(H)}. The pointed graph $(H_{\text{end}},\{h\})$ is the pointed graph $(H',\{h\})$ of Figure~\ref{fig:P(H)}.
        \label{fig:H-free verification (v)}.}
    \end{figure}

    \item If $u \in V_2$, let $C_u \in \cE_2$ such that $c \in C_u$. The vertex $u$ rejects if there exists three pointed graphs $(H_\text{start},\{h_1,h_2\})$, $(H_\text{end}^{(1)},\{h_1\})$, $(H_\text{end}^{(2)},\{h_2\})$, and two vertices $v_1 \in C_{v_1}(d_{v_1})$, $v_2 \in C_{v_2}(d_{v_2})$ for some $d_{v_1},d_{v_2} \leqslant k-1$ where $C_{v_1},C_{v_2}$ are two different ECC$_2$'s different from $C_u$, such that the following conditions are satisfied:
    \begin{itemize}
        \item $H_\text{end}^{(1)}$, $H_\text{end}^{(2)}$ are two disjoint induced subgraphs of~$H$
        \item $(H_{\text{start}},\{h_1,h_2\})$ is the complementary in~$H$ of the disjoint union of $(H_\text{end}^{(1)},\{h_1\})$ and $(H_\text{end}^{(2)},\{h_2\})$
        \item $H_\text{start}$ is an induced subgraph of $G_{\preccurlyeq u}$, included in $V_{\preccurlyeq u}$, which contains at least vertex from $C_u$
        \item $h_1$ (resp. $h_2$) is mapped to $v_1$ (resp. $v_2$), and it is the only vertex in $H_\text{start} \cap (\bigcup_{d \leqslant d_{v_1}} C_{v_1}(d))$ (resp. in $H_\text{start} \cap (\bigcup_{d \leqslant d_{v_2}} C_{v_2}(d))$)
        \item $\htab[v_1, (H_\text{end}^{(1)})] = 1$ and $\htab[v_2, (H_\text{end}^{(2)})] = 1$
    \end{itemize}
    An illustration is shown on Figure~\ref{fig:H-free verification (vi)}.

    \begin{figure}
        \centering
        \scalebox{0.8}{
        \begin{tikzpicture}[x=0.75pt,y=0.75pt,yscale=-1,xscale=1]
%uncomment if require: \path (0,491); %set diagram left start at 0, and has height of 491

%Shape: Polygon Curved [id:ds9504680124337159] 
\draw  [color={rgb, 255:red, 255; green, 255; blue, 255 }  ,draw opacity=1 ][fill={rgb, 255:red, 74; green, 144; blue, 226 }  ,fill opacity=0.29 ] (428.5,232) .. controls (455.5,266) and (441.5,297) .. (414.5,311) .. controls (387.5,325) and (354.5,314) .. (335.5,290) .. controls (316.5,266) and (328.5,241) .. (346.5,229) .. controls (364.5,217) and (401.5,198) .. (428.5,232) -- cycle ;
%Straight Lines [id:da9610874150223382] 
\draw    (288.25,136) -- (255.5,110) ;
%Straight Lines [id:da36683945986761235] 
\draw    (285.5,182) -- (331.5,163) ;
%Straight Lines [id:da07549917059948741] 
\draw    (288.25,136) -- (285.5,182) ;
%Straight Lines [id:da10890148954602463] 
\draw    (233.5,204) -- (216.5,181) ;
%Straight Lines [id:da9452888767913329] 
\draw    (201.5,210) -- (233.5,204) ;
%Shape: Polygon Curved [id:ds05759257370937487] 
\draw  [color={rgb, 255:red, 255; green, 255; blue, 255 }  ,draw opacity=1 ][fill={rgb, 255:red, 74; green, 144; blue, 226 }  ,fill opacity=0.29 ] (364.5,54) .. controls (390.5,84) and (358.5,111) .. (308.5,116) .. controls (258.5,121) and (237.5,70) .. (254.5,45) .. controls (271.5,20) and (338.5,24) .. (364.5,54) -- cycle ;
%Straight Lines [id:da07445068953966494] 
\draw    (285.5,182) -- (311.25,205) ;
%Straight Lines [id:da32981493890040425] 
\draw    (311.25,205) -- (344.5,215) ;
%Straight Lines [id:da8937051919856945] 
\draw    (316.5,238) -- (311.25,205) ;
%Straight Lines [id:da7642932728978665] 
\draw    (233.5,204) -- (285.5,182) ;
%Shape: Circle [id:dp8830559717931729] 
\draw  [fill={rgb, 255:red, 255; green, 255; blue, 255 }  ,fill opacity=1 ] (306,205) .. controls (306,202.1) and (308.35,199.75) .. (311.25,199.75) .. controls (314.15,199.75) and (316.5,202.1) .. (316.5,205) .. controls (316.5,207.9) and (314.15,210.25) .. (311.25,210.25) .. controls (308.35,210.25) and (306,207.9) .. (306,205) -- cycle ;
%Shape: Circle [id:dp7648029300836112] 
\draw  [fill={rgb, 255:red, 255; green, 255; blue, 255 }  ,fill opacity=1 ] (280.25,182) .. controls (280.25,179.1) and (282.6,176.75) .. (285.5,176.75) .. controls (288.4,176.75) and (290.75,179.1) .. (290.75,182) .. controls (290.75,184.9) and (288.4,187.25) .. (285.5,187.25) .. controls (282.6,187.25) and (280.25,184.9) .. (280.25,182) -- cycle ;
%Shape: Circle [id:dp2068062586395083] 
\draw  [fill={rgb, 255:red, 255; green, 255; blue, 255 }  ,fill opacity=1 ] (326.25,163) .. controls (326.25,160.1) and (328.6,157.75) .. (331.5,157.75) .. controls (334.4,157.75) and (336.75,160.1) .. (336.75,163) .. controls (336.75,165.9) and (334.4,168.25) .. (331.5,168.25) .. controls (328.6,168.25) and (326.25,165.9) .. (326.25,163) -- cycle ;
%Shape: Circle [id:dp5740320722828233] 
\draw  [fill={rgb, 255:red, 255; green, 255; blue, 255 }  ,fill opacity=1 ] (228.25,204) .. controls (228.25,201.1) and (230.6,198.75) .. (233.5,198.75) .. controls (236.4,198.75) and (238.75,201.1) .. (238.75,204) .. controls (238.75,206.9) and (236.4,209.25) .. (233.5,209.25) .. controls (230.6,209.25) and (228.25,206.9) .. (228.25,204) -- cycle ;
%Shape: Circle [id:dp7804246944274184] 
\draw  [fill={rgb, 255:red, 255; green, 255; blue, 255 }  ,fill opacity=1 ] (211.25,181) .. controls (211.25,178.1) and (213.6,175.75) .. (216.5,175.75) .. controls (219.4,175.75) and (221.75,178.1) .. (221.75,181) .. controls (221.75,183.9) and (219.4,186.25) .. (216.5,186.25) .. controls (213.6,186.25) and (211.25,183.9) .. (211.25,181) -- cycle ;
%Shape: Polygon Curved [id:ds28609840186731816] 
\draw  [color={rgb, 255:red, 255; green, 255; blue, 255 }  ,draw opacity=1 ][fill={rgb, 255:red, 74; green, 144; blue, 226 }  ,fill opacity=0.29 ] (156.5,167) .. controls (189.5,190) and (187.5,243) .. (126.5,238) .. controls (65.5,233) and (59.5,210) .. (67.5,185) .. controls (75.5,160) and (123.5,144) .. (156.5,167) -- cycle ;
%Curve Lines [id:da14532988988260653] 
\draw  [dash pattern={on 4.5pt off 4.5pt}]  (143.5,134) .. controls (196.5,138) and (220.5,206) .. (185.5,240) ;
%Curve Lines [id:da01610472714028688] 
\draw  [dash pattern={on 4.5pt off 4.5pt}]  (378.5,115) .. controls (319.5,142) and (263.5,150) .. (236.5,109) ;
%Straight Lines [id:da6755944298353103] 
\draw    (150.5,218) -- (201.5,210) ;
%Shape: Circle [id:dp12839085871112188] 
\draw  [fill={rgb, 255:red, 245; green, 166; blue, 35 }  ,fill opacity=1 ][line width=1.5]  (196.25,210) .. controls (196.25,207.1) and (198.6,204.75) .. (201.5,204.75) .. controls (204.4,204.75) and (206.75,207.1) .. (206.75,210) .. controls (206.75,212.9) and (204.4,215.25) .. (201.5,215.25) .. controls (198.6,215.25) and (196.25,212.9) .. (196.25,210) -- cycle ;
%Shape: Circle [id:dp2516658981535669] 
\draw  [fill={rgb, 255:red, 255; green, 255; blue, 255 }  ,fill opacity=1 ] (145.25,218) .. controls (145.25,215.1) and (147.6,212.75) .. (150.5,212.75) .. controls (153.4,212.75) and (155.75,215.1) .. (155.75,218) .. controls (155.75,220.9) and (153.4,223.25) .. (150.5,223.25) .. controls (147.6,223.25) and (145.25,220.9) .. (145.25,218) -- cycle ;
%Straight Lines [id:da12082866778933943] 
\draw    (255.5,110) -- (299.5,94) ;
%Shape: Circle [id:dp1671695555965177] 
\draw  [fill={rgb, 255:red, 255; green, 255; blue, 255 }  ,fill opacity=1 ] (250.25,110) .. controls (250.25,107.1) and (252.6,104.75) .. (255.5,104.75) .. controls (258.4,104.75) and (260.75,107.1) .. (260.75,110) .. controls (260.75,112.9) and (258.4,115.25) .. (255.5,115.25) .. controls (252.6,115.25) and (250.25,112.9) .. (250.25,110) -- cycle ;
%Straight Lines [id:da3070371000469424] 
\draw    (316.5,238) -- (348.5,248) ;
%Straight Lines [id:da36309947909111995] 
\draw    (344.5,215) -- (348.5,248) ;
%Shape: Circle [id:dp037515762483332105] 
\draw  [fill={rgb, 255:red, 255; green, 255; blue, 255 }  ,fill opacity=1 ] (311.25,238) .. controls (311.25,235.1) and (313.6,232.75) .. (316.5,232.75) .. controls (319.4,232.75) and (321.75,235.1) .. (321.75,238) .. controls (321.75,240.9) and (319.4,243.25) .. (316.5,243.25) .. controls (313.6,243.25) and (311.25,240.9) .. (311.25,238) -- cycle ;
%Shape: Circle [id:dp7901602671587918] 
\draw  [fill={rgb, 255:red, 255; green, 255; blue, 255 }  ,fill opacity=1 ] (339.25,215) .. controls (339.25,212.1) and (341.6,209.75) .. (344.5,209.75) .. controls (347.4,209.75) and (349.75,212.1) .. (349.75,215) .. controls (349.75,217.9) and (347.4,220.25) .. (344.5,220.25) .. controls (341.6,220.25) and (339.25,217.9) .. (339.25,215) -- cycle ;
%Shape: Circle [id:dp6898604504179402] 
\draw  [fill={rgb, 255:red, 255; green, 255; blue, 255 }  ,fill opacity=1 ] (343.25,248) .. controls (343.25,245.1) and (345.6,242.75) .. (348.5,242.75) .. controls (351.4,242.75) and (353.75,245.1) .. (353.75,248) .. controls (353.75,250.9) and (351.4,253.25) .. (348.5,253.25) .. controls (345.6,253.25) and (343.25,250.9) .. (343.25,248) -- cycle ;
%Shape: Circle [id:dp7517353375821846] 
\draw  [fill={rgb, 255:red, 255; green, 255; blue, 255 }  ,fill opacity=1 ] (389.25,287) .. controls (389.25,284.1) and (391.6,281.75) .. (394.5,281.75) .. controls (397.4,281.75) and (399.75,284.1) .. (399.75,287) .. controls (399.75,289.9) and (397.4,292.25) .. (394.5,292.25) .. controls (391.6,292.25) and (389.25,289.9) .. (389.25,287) -- cycle ;
%Straight Lines [id:da8665564980224484] 
\draw    (299.5,94) -- (288.25,136) ;
%Shape: Circle [id:dp7605653485298026] 
\draw  [fill={rgb, 255:red, 245; green, 166; blue, 35 }  ,fill opacity=1 ][line width=1.5]  (283,136) .. controls (283,133.1) and (285.35,130.75) .. (288.25,130.75) .. controls (291.15,130.75) and (293.5,133.1) .. (293.5,136) .. controls (293.5,138.9) and (291.15,141.25) .. (288.25,141.25) .. controls (285.35,141.25) and (283,138.9) .. (283,136) -- cycle ;
%Shape: Circle [id:dp9417531962975413] 
\draw  [fill={rgb, 255:red, 255; green, 255; blue, 255 }  ,fill opacity=1 ] (294.25,94) .. controls (294.25,91.1) and (296.6,88.75) .. (299.5,88.75) .. controls (302.4,88.75) and (304.75,91.1) .. (304.75,94) .. controls (304.75,96.9) and (302.4,99.25) .. (299.5,99.25) .. controls (296.6,99.25) and (294.25,96.9) .. (294.25,94) -- cycle ;

% Text Node
\draw (208.75,213.4) node [anchor=north west][inner sep=0.75pt]    {\LARGE $v_{1}$};
% Text Node
\draw (33,192.4) node [anchor=north west][inner sep=0.75pt]    {\LARGE $C_{v_{1}}$};
% Text Node
\draw (81,124.4) node [anchor=north west][inner sep=0.75pt]    {\LARGE $C_{v_{1}}( d_{v_{1}})$};
% Text Node
\draw (435,299.4) node [anchor=north west][inner sep=0.75pt]    {\LARGE $C_{u}$};
% Text Node
\draw (403.5,277.4) node [anchor=north west][inner sep=0.75pt]    {\LARGE $u$};
% Text Node
\draw (245.5,220.4) node [anchor=north west][inner sep=0.75pt]    {\LARGE $H_\text{start}$};
% Text Node
\draw (128.5,241.4) node [anchor=north west][inner sep=0.75pt]    {\LARGE $H_\text{end}^{(1)}$};
% Text Node
\draw (295.5,139.4) node [anchor=north west][inner sep=0.75pt]    {\LARGE $v_{2}$};
% Text Node
\draw (383,102.4) node [anchor=north west][inner sep=0.75pt]    {\LARGE $C_{v_{2}}( d_{v_{2}})$};
% Text Node
\draw (324,9.4) node [anchor=north west][inner sep=0.75pt]    {\LARGE $C_{v_{2}}$};
% Text Node
\draw (206.5,81.4) node [anchor=north west][inner sep=0.75pt]    {\LARGE $H_\text{end}^{(2)}$};

\end{tikzpicture}

        }
        \caption{Illustration of the condition making $u$ reject at step~(vi) of the verification. The graph $H$ is the one of Figure~\ref{fig:P(H)}. The graph $H_\text{end}^{(1)}$ is a path on two vertices and the graph $H_\text{end}^{(2)}$ is a triangle.}
        \label{fig:H-free verification (vi)}
    \end{figure}
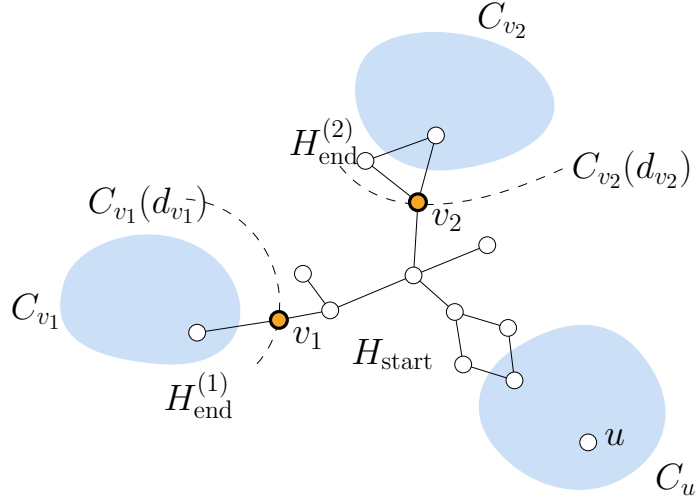

    \item If $u$ did not reject at this point, it accepts.
\end{enumerate}

\medskip{}
\textbf{Correctness.} Let us show that this certification scheme is correct. First, assume that $G$ contains an induced $H$, ant let us show that for every assignment of certificates, at least one vertex rejects. If no vertex rejects at step~(i) of the verification, then every vertex $u$ computed the ECC-table $T_G$ and its witnessed graph $G_{\preccurlyeq u}$ ($u$ knows also $V_{\preccurlyeq u}$ because it can be computed directly from $T_G$). If no vertex rejects at step~(iii), then~$\htab$ is also correct. Then, we have the following Claim~\ref{claim:number of ECCs intersecting H}, and its proof is analogous as proofs of Claims~\ref{claim:4k-1 1} and~\ref{claim:4k-1 2}.

\begin{claim}
    \label{claim:number of ECCs intersecting H}
    \begin{enumerate}
        \item There are at most three $C \in \cE_2$ such that $C \cap H \neq \emptyset$.
        \item If there is at most one $C \in \cE_2$ such that $|C \cap H| \geqslant 2$, then there exists a vertex which rejects at step~(iv) of the verification.
    \end{enumerate}
\end{claim}

By Claim~\ref{claim:number of ECCs intersecting H}, we can assume that there are at least two $C \in \cE_2$ such that $|C \cap H| \geqslant 2$. Now, there are two cases.

\begin{itemize}
    \item If there are exactly two ECC$_2$'s intersecting $H$, denoted by $C_1$ and $C_2$, since $H$ has $4k-1$ vertices and that, by assumption, we have $|C_i \cap H| \geqslant 2$ for $i \in \{1,2\}$, there exists $1 \leqslant d \leqslant k-1$ and $i \in \{1,2\}$ such that $|C_i(d)|=1$. Without loss of generality, assume that $i=1$, and let us denote by $v$ the unique vertex in $C_1(d)$. Let $u \in C_2$. Then, we can decompose $H$ in two parts $H_\text{start}$ and $H_\text{end}$ satisfying the conditions making $u$ reject at step~(v).

    \item If there are exactly three ECC$_2$'s intersecting $H$, denoted by $C_1$, $C_2$ and $C_3$, there is at most one $i \in \{1,2,3\}$ such that, for all $d \in \{1, \ldots, k-1\}$, $|C_i(d)| \geqslant 2$. Without loss of generality, assume that there exists $d_1, d_2 \in \{1, \ldots, k-1\}$ such that $|C_i(d_i)|=1$ for all $i \in \{1,2\}$. Let us denote by $v_1$ (resp. by $v_2$) the only vertex in $C_1(d_1)$ (resp. in $C_2(d_2)$). Let $u \in C_3$. Then, we can decompose $H$ in three parts $H_\text{start}$, $H_\text{end}^{(1)}$, $H_\text{end}^{(2)}$ satisfying the conditions making $u$ reject at step~(vi).
\end{itemize}

Conversely, assume that $G$ does not have an induced $H$, and let us show that there exists an assignment of the certificates such that every vertex accepts. This assignment is the following one: the prover attributes the certificates such that no vertex rejects in the local computation schemes of Theorems~\ref{thm:computation scheme for T_G} and~\ref{thm:computation scheme for g_<u}, and gives the correct value of $\htab$ to every vertex. With such a certificate, a vertex can not reject at steps~(i), (ii) and (iii). With exactly the same argument as in the proof of Theorem~\ref{thm:H_4k_free_subquadratic}, a vertex can not reject at step~(iv).

By contradiction, assume that a vertex $u$ rejects at step~(v). Then, there exists two pointed graphs $(H_\text{start}, \{h\})$, $(H_\text{end},\{h\})$ satisfying the conditions making $u$ reject at step~(v). By the properties they satisfy, we can glue them in $h$ to obtain a graph $H$, which is still induced in $G$, which is a contradiction. Similarly, if a vertex $u$ rejects at step~(vi), there exists three pointed graphs $(H_\text{start}, \{h_1,h_2\})$, $(H_\text{end}^{(1)},h_1)$, $(H_\text{end}^{(2)},h_2)$ satisfying the conditions making $u$ reject at step~(vi). We can again glue them on $h_1$, $h_2$ to obtain an induced $H$ in $G$, which is a contradiction. Thus, all the vertices accept at step~(vii), which concludes the proof.
\end{proof}

 We stated Theorem~\ref{thm:H_4k_free_subquadratic} with an induced subgraph perspective but one can easily remark that the proof technique is general enough to be directly extended to subgraphs rather than induced subgraphs. Indeed, we simply have to change the definition of $\htab$ in such a way $\htab[v, (H',\{h\})]$ is equal to $1$ if $H'$ is a subgraph of $\bigcup_{d \leqslant d_v} C_v(d)$ instead of an induced subgraph in order to get the same result. We actually think that our proof technique can be used more generally to get subquadratic bounds in the polynomial regime.

%%%%%%%%%%%%%%%%%%%%%%%%%%%%%%%%%%%%%%%%%%%%%%%%%%%%%%%%%%%
\subsection{Upper bound for paths of length at most $3k-1$ in $\tilde{O}(n)$}
\label{sec:P3k}
%%%%%%%%%%%%%%%%%%%%%%%%%%%%%%%%%%%%%%%%%%%%%%%%%%%%%%%%%%%

\thmThreekQuasilinear*

To prove this result, we will in fact prove the following Proposition~\ref{prop:1+epsilon}, and Theorem~\ref{thm:3k-quasilinear} will be a simple corollary. While, until now, all the results where based on $\varepsilon=\frac 12$ and the partition of the vertex set into two parts, we will use the whole power of our machinery in this proof by choosing an arbitrarily small value of $\varepsilon$; this leads to new technicalities we have to deal with in the proof. 

\begin{proposition}\label{prop:1+epsilon}
    There exists $c > 0$ such that, for all $0 < \varepsilon < 1$, there exists a certification scheme for $P_{3k-1}$-free graphs using at most $\frac{c}{\varepsilon} \cdot \log^2 n \cdot n^{1+\varepsilon}$ bits.
\end{proposition}

\begin{proof}[Proof of Theorem~\ref{thm:3k-quasilinear} assuming Proposition~\ref{prop:1+epsilon}]
    The certification in size $O(n \log^3 n)$ is the following one. Let $n$ be the number of vertices in the graph. The prover writes $n$ in the certificate of every vertex of the graph, and certifies the correctness of this information with $O(n \log n)$ bits (by coding a spanning tree). Then, it uses the certification scheme given by Proposition~\ref{prop:1+epsilon}, with $\varepsilon=1/\log n $. This uses $O(n \log^3 n)$ bits in total, since $n^{1/\log n}$ is a constant.
\end{proof}

\begin{proof}[Proof of Proposition~\ref{prop:1+epsilon}]
    Let $0 < \varepsilon < 1$ and $N = \left\lceil \frac{1}{\varepsilon} \right\rceil$. Let us describe a certification scheme for $P_{3k-1}$-free graphs. We will show that it is correct and uses at most $\frac{c}{\varepsilon} \cdot \log^2 n \cdot n^{1+\varepsilon}$ bits, where $c$ is a constant independent of $\varepsilon$. Let $G = (V,E)$ be a graph and $n = |V|$.

\medskip{}
\textbf{Certification.}
    Let us describe the certificates given by the prover to the vertices. First, the prover gives to every vertex its certificates in the local computation schemes of Theorems~\ref{thm:computation scheme for T_G} and~\ref{thm:computation scheme for g_<u}, so that each vertex $u$ can compute $T_G$ and $G_{\preccurlyeq u}$. These certificates have size $O(\frac{n^{1+\varepsilon}}{\varepsilon} \log^2 n)$.
    Then, as in the certification scheme of Theorem~\ref{thm:p_4k_free_subquadratic}, the prover adds some information in the certificates, in another field, denoted by $\lp$, which is identical in the certificates of all the vertices.
    In~$\lp$, the prover writes a table which has $n$ rows and $N$ columns.
    The rows are indexed with the identifiers of the vertices, and the columns by $\{1, \ldots, N\}$. Let $v \in V$ and $i \in \{1, \ldots, N\}$.
    If $v \notin \bigcup_{C \in \cE_i} \bigcup_{0 \leqslant d \leqslant k-1} C(d)$, we set $\lp[v,i]:=\bot$.
    Else, there exists a unique $C_v \in \cE_i$ and $d_v \leqslant k-1$ such that $v \in C_v(d_v)$.
    In this case, let $P_v$ denote the longest induced path starting from $v$ and having all its other vertices in $\bigcup_{0 \leqslant d < d_v} C_v(d)$, and let us denote its length by $\ell(P_v)$. We set $\lp[v,i]:=\ell(P_v)$.

    Since $\lp$ has $n$ rows, $N$ columns, and each entry is written on at most $\log n$ bits, its size it $O(N n \log n)$. Thus, the overall size of each certificate remains $O(\frac{n^{1+\varepsilon}}{\varepsilon} \log^2 n)$.

\medskip{}
\textbf{Verification.}
    The verification of each vertex $u \in V$ consists in the following steps:

    \begin{enumerate}[(i)]
        \item First, $u$ applies the verification of the local computation schemes given by Theorems~\ref{thm:computation scheme for T_G} and~\ref{thm:computation scheme for g_<u}.

        If no vertex rejects, then every vertex $u$ computed its witnessed graph $G_{\preccurlyeq u}$ and the ECC-table $T_G$ of $G$.

        \item Then, $u$ checks that the field $\lp$ is written the same in its certificate and the certificates of all its neighbors. If it is not the case, $u$ rejects.
        
        If no vertex rejects, $\lp$ is the same in the certificates of all the vertices in $G$. The next step will consist in checking its correctness.

        \item Let $i_u \in \{1, \ldots, N\}$ be such that $u \in V_{i_u}$. For each $i \in \{1, \ldots, i_u\}$, $u$ does the following.
        We have $u \in H_i$. Let us denote by $C_u$ its ECC$_i$. Let $d \in \{1, \ldots, k-1\}$, and let $v \in C_u(d)$.
        Since $v \notin C_u$, we have $v \in L_{i-1}$ (indeed, if $v \in H_i$, it would imply $v \in C_u$ because $v$ is at distance at most $k-1$ from $C_u$).
        Let $i_v \in \{1, \ldots, i-1\}$ be such that $v \in V_{i_v}$.
        We have $C_u \subseteq H_i \subseteq H_{i_v}$, so $u,v$ are in the same ECC$_{i_v}$.
        Thus, we have $v \preccurlyeq u$.
        Since this true for all $d \in \{1, \ldots, k-1\}$ and $v \in C_u(d)$, we have $\bigcup_{1 \leqslant d \leqslant k-1} C_u(d) \subseteq V_{\preccurlyeq u}$.
        Thus, for all $v \in \bigcup_{1 \leqslant d \leqslant k-1} C_u(d)$, $u$ can check if $\lp[v,i]$ is correct, and rejects if it is not the case.

        If no vertex rejects at this point, then $\lp$ is correct.

        \item If $G_{\preccurlyeq u}$ has an induced path $P$ of length $3k-1$, which is included in $V_{\preccurlyeq u}$, then $u$ rejects. %\linda{is this really $4k-1$ not $3k-1$?} 
%        \linda{is there possibly an error here, I assume we would want also paths of $G_{\preccurlyeq u}$ such that the vertices of $P$ not in $V_{\preccurlyeq u}$ are all in different ECC$_i$ for some $i > 1$ (and thus they are not adjacent). }
        \item Finally, $u$ rejects if it sees, in its view at distance~$k$ %\linda{don't you want $V_{\preccurlyeq u}$ not the view of distance $k$ because we actually have more by step (i)}
        , an induced path $P$ of some length $\ell(P)$, which goes from a vertex $v_1$ to a vertex $v_2$, such that the following conditions are satisfied:
        \begin{itemize}
            \item at most one vertex in $P$ is at distance exactly~$k$ from $u$
            \item there exists $i \in \{1, \ldots, N\}$ and 
            two distinct $C_{v_1}, C_{v_2} \in \cE_i$ and $d_{v_1},d_{v_2} \in \{0, \ldots, k-1\}$ such that $v_1 \in C_{v_1}(d_{v_1})$ and $v_2 \in C_{v_2}(d_{v_2})$
            \item $v_1$ (resp. $v_2$) is the only vertex in $P \cap (\bigcup_{0 \leqslant d \leqslant d_{v_1}} C_{v_1}(d))$ (resp. in $P \cap (\bigcup_{0 \leqslant d \leqslant d_{v_2}} C_{v_2}(d))$)
            \item $\ell(P) + (\lp[v_1,i]-1) + (\lp[v_2,i]-1) \geqslant 3k-1$ 
        \end{itemize}

        \item If $u$ did not reject previously, it accepts.
    \end{enumerate}

    \medskip{}
    \textbf{Correctness.}
    Let us show that this certification scheme is correct. First, assume that $G$ has an induced path $P$ of length $3k-1$, and let us show that for every assignment of certificates, at least one vertex rejects.
    If no vertex rejects at step~(i), then every vertex~$u$ knows the ECC-table $T_G$, and its witnessed graph $G_{\preccurlyeq u}$.
    If no vertex rejects in steps~(ii) and~(iii), then $\lp$ is correctly written in the certificates.
    
    Let us decompose $P$ in three consecutive parts $P_1, P_2, P_3$, having respectively $k,k-1,k$ vertices. Let $x_1$ (resp. $x_2,x_3$) be the vertex of maximum degree in $P_1$ (resp. in $P_2,P_3$). Let $i_1, i_2, i_3 \in \{1, \ldots, N\}$ be such that $x_1 \in V_{i_1}$, $x_2 \in V_{i_2}$ and $x_3 \in V_{i_3}$ (see Figure~\ref{fig:P1P2P3} for an illustration). Note that the distance between $x_1$ and $x_2$ is at most $2k-2$. Similarly, the distance between $x_2$ and $x_3$ is at most $2k-2$.
    The two following Claims~\ref{claim:i2>=min(i1,i3)} and~\ref{claim:i2<=min(i3,i1)} show that, in all cases, at least one vertex rejects.

    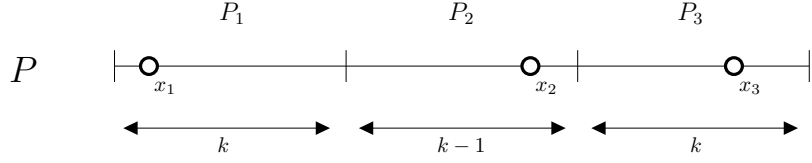
\begin{figure}
        \centering
        \scalebox{0.8}{
        \begin{tikzpicture}[x=0.75pt,y=0.75pt,yscale=-1,xscale=1]
%uncomment if require: \path (0,453); %set diagram left start at 0, and has height of 453

%Straight Lines [id:da34457073213729916] 
\draw    (103,106) -- (247.5,106) ;
%Straight Lines [id:da44347665266045655] 
\draw    (103,96) -- (103,106) ;
%Straight Lines [id:da6019994214812885] 
\draw    (247.5,106) -- (392,106) ;
%Straight Lines [id:da17166937398689874] 
\draw    (392,106) -- (536.5,106) ;
%Straight Lines [id:da30721494472723554] 
\draw    (103,106) -- (103,116) ;
%Straight Lines [id:da8784685713582623] 
\draw    (247.5,106) -- (247.5,116) ;
%Straight Lines [id:da025690641237997136] 
\draw    (247.5,96) -- (247.5,106) ;
%Straight Lines [id:da26880262337083016] 
\draw    (392,96) -- (392,106) ;
%Straight Lines [id:da6942159903691403] 
\draw    (392,106) -- (392,116) ;
%Straight Lines [id:da8107167458263409] 
\draw    (536.5,96) -- (536.5,106) ;
%Straight Lines [id:da010138467744151947] 
\draw    (536.5,106) -- (536.5,116) ;
%Shape: Circle [id:dp7657243920312456] 
\draw  [fill={rgb, 255:red, 255; green, 255; blue, 255 }  ,fill opacity=1 ][line width=1.5]  (119.25,106) .. controls (119.25,103.1) and (121.6,100.75) .. (124.5,100.75) .. controls (127.4,100.75) and (129.75,103.1) .. (129.75,106) .. controls (129.75,108.9) and (127.4,111.25) .. (124.5,111.25) .. controls (121.6,111.25) and (119.25,108.9) .. (119.25,106) -- cycle ;
%Shape: Circle [id:dp17607594654612235] 
\draw  [fill={rgb, 255:red, 255; green, 255; blue, 255 }  ,fill opacity=1 ][line width=1.5]  (357.25,106) .. controls (357.25,103.1) and (359.6,100.75) .. (362.5,100.75) .. controls (365.4,100.75) and (367.75,103.1) .. (367.75,106) .. controls (367.75,108.9) and (365.4,111.25) .. (362.5,111.25) .. controls (359.6,111.25) and (357.25,108.9) .. (357.25,106) -- cycle ;
%Shape: Circle [id:dp8173683897324276] 
\draw  [fill={rgb, 255:red, 255; green, 255; blue, 255 }  ,fill opacity=1 ][line width=1.5]  (484.25,106) .. controls (484.25,103.1) and (486.6,100.75) .. (489.5,100.75) .. controls (492.4,100.75) and (494.75,103.1) .. (494.75,106) .. controls (494.75,108.9) and (492.4,111.25) .. (489.5,111.25) .. controls (486.6,111.25) and (484.25,108.9) .. (484.25,106) -- cycle ;
%Straight Lines [id:da9478017838066836] 
\draw    (111.5,145) -- (234.5,145) ;
\draw [shift={(237.5,145)}, rotate = 180] [fill={rgb, 255:red, 0; green, 0; blue, 0 }  ][line width=0.08]  [draw opacity=0] (8.93,-4.29) -- (0,0) -- (8.93,4.29) -- cycle    ;
\draw [shift={(108.5,145)}, rotate = 0] [fill={rgb, 255:red, 0; green, 0; blue, 0 }  ][line width=0.08]  [draw opacity=0] (8.93,-4.29) -- (0,0) -- (8.93,4.29) -- cycle    ;
%Straight Lines [id:da20782288694675288] 
\draw    (258.5,145) -- (381.5,145) ;
\draw [shift={(384.5,145)}, rotate = 180] [fill={rgb, 255:red, 0; green, 0; blue, 0 }  ][line width=0.08]  [draw opacity=0] (8.93,-4.29) -- (0,0) -- (8.93,4.29) -- cycle    ;
\draw [shift={(255.5,145)}, rotate = 0] [fill={rgb, 255:red, 0; green, 0; blue, 0 }  ][line width=0.08]  [draw opacity=0] (8.93,-4.29) -- (0,0) -- (8.93,4.29) -- cycle    ;
%Straight Lines [id:da1510630887314074] 
\draw    (403.5,145) -- (526.5,145) ;
\draw [shift={(529.5,145)}, rotate = 180] [fill={rgb, 255:red, 0; green, 0; blue, 0 }  ][line width=0.08]  [draw opacity=0] (8.93,-4.29) -- (0,0) -- (8.93,4.29) -- cycle    ;
\draw [shift={(400.5,145)}, rotate = 0] [fill={rgb, 255:red, 0; green, 0; blue, 0 }  ][line width=0.08]  [draw opacity=0] (8.93,-4.29) -- (0,0) -- (8.93,4.29) -- cycle    ;

% Text Node
\draw (167,65) node [anchor=north west][inner sep=0.75pt]    {\Large $P_{1}$};
% Text Node
\draw (310,65) node [anchor=north west][inner sep=0.75pt]    {\Large $P_{2}$};
% Text Node
\draw (453,65) node [anchor=north west][inner sep=0.75pt]    {\Large $P_{3}$};
% Text Node
\draw (126.5,114.65) node [anchor=north west][inner sep=0.75pt]    {$x_{1}$};
% Text Node
\draw (364.5,114.65) node [anchor=north west][inner sep=0.75pt]    {$x_{2}$};
% Text Node
\draw (491.5,114.65) node [anchor=north west][inner sep=0.75pt]    {$x_{3}$};
% Text Node
\draw (166,149.4) node [anchor=north west][inner sep=0.75pt]    {$k$};
% Text Node
\draw (303,149.4) node [anchor=north west][inner sep=0.75pt]    {$k-1$};
% Text Node
\draw (461,149.4) node [anchor=north west][inner sep=0.75pt]    {$k$};
\draw (36,99) node [anchor=north west][inner sep=0.75pt]    {\LARGE $P$};

\end{tikzpicture}
        }

        \caption{The decomposition of $P$ in three consecutive parts $P_1$, $P_2$, $P_3$. For each $i \in \{1,2,3\}$, $x_i$ is the vertex of maximum degree in $P_i$.}
        \label{fig:P1P2P3}
    \end{figure}

    \begin{claim}
        \label{claim:i2>=min(i1,i3)}
        If $i_2 \geqslant \min(i_1,i_3)$, then at least one vertex reject at step (iv) of the verification.
    \end{claim}
    \begin{proof}
        Assume that $i_2 \geqslant i_1$. Then, we have $x_2 \in H_{i_1}$. Since $x_1$ and $x_2$ are at distance at most $2k-2$, there are in the same ECC$_{i_1}$. Thus, we have $x_1 \preccurlyeq x_2$. Note also that, since $x_1$ is the vertex of maximum degree in $P_1$, we have $P_1 \subseteq V_{\preccurlyeq x_1}$. Thus, by transitivity, we have $P_1 \subseteq V_{\preccurlyeq x_2}$. Moreover, we also have $P_2 \subseteq V_{\preccurlyeq x_2}$.
        Finally, there are two cases:
        \begin{itemize}
            \item If $i_2 \geqslant i_3$, by the same argument, we have $P_3 \subseteq V_{\preccurlyeq x_2}$. So $P$ is included in $V_{\preccurlyeq x_2}$, and $x_2$ rejects at step~(iv).
            \item If $i_3 > i_2$, we have $x_2 \preccurlyeq x_3$. By transitivity we have $P_1 \subseteq V_{\preccurlyeq x_3}$ and $P_2 \subseteq V_{\preccurlyeq x_3}$. Since we also have $P_3 \subseteq V_{\preccurlyeq x_3}$, $P$ is included in $V_{\preccurlyeq x_3}$ and $x_3$ rejects at step~(iv). \qedhere
        \end{itemize}
    \end{proof}

    \begin{claim}
        \label{claim:i2<=min(i3,i1)}
        If $i_2 \leqslant \min(i_1,i_3)$, then at least one vertex rejects.
    \end{claim}
    \begin{proof}
        Without loss of generality, assume that $i_2 \leqslant i_3 \leqslant i_1$.
        We have $x_1 \in H_{i_3}$. There are two cases, depending on whether $x_1$ and $x_3$ are in the same ECC$_{i_3}$ or not.
        \begin{itemize}
            \item If $x_1$ and $x_3$ are in the same ECC$_{i_3}$, we have $x_3 \preccurlyeq x_1$. Thus, by transitivity, $P_3 \subseteq V_{\preccurlyeq x_1}$. Since we also have $P_1 \subseteq P_{\preccurlyeq x_1}$ and $P_2 \in V_{\preccurlyeq x_1}$ (because $i_2 \leqslant i_1$), then $P$ is included in $V_{\preccurlyeq x_1}$ so $x_1$ rejects at step~(iv).

            \item If $x_1$ and $x_3$ are not in the same ECC$_{i_3}$, let us prove that some vertex rejects at step~(v). Let $C_{x_1}, C_{x_3} \in \cE_{i_3}$ be such that $x_1 \in C_{x_1}$ and $x_3 \in C_{x_3}$.
            Since $P$ passes through $C_{x_1}$ and $C_{x_3}$, in the part between $x_1$ and $x_3$, it passes also through $C_{x_1}(d)$ and $C_{x_3}(d)$ for all $d \in \{0, \ldots, k-1\}$.
            Moreover, since $P$ has length $3k-1$, there exists at least $k+1$ sets among $\{C_{x_1}(0), \ldots, C_{x_1}(k-1), C_{x_3}(0), \ldots, C_{x_3}(k-1)\}$ through which $P$ passes only once.
            Thus, in the part of $P$ between $x_1$ and $x_3$, there exists $d_1, d_3 \in \{0, \ldots, k-1\}$ such that $P$ passes only once through $C_{x_1}(d_1)$ and $C_{x_3}(d_3)$, in some vertices denoted by $u_1, u_3$, and $u_1,u_3$ are at distance at most $2k-1$ from each other.
            Let $P'$ denote the part of $P$ between $u_1$ and $u_3$.
            Let $u \in P'$ such that $u$ is at distance at most $k-1$ from every vertex in $P'$, except $u_1$, and $u,u_1$ are at distance at most $k$. This is depicted on Figure~\ref{fig:u1u3u}.
            
            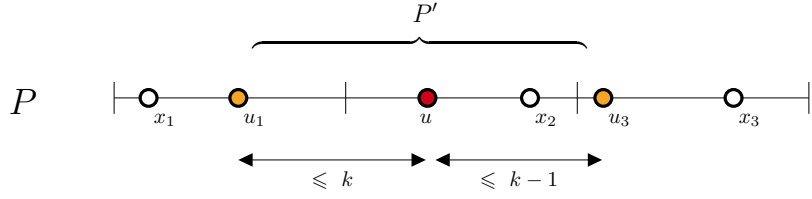
\begin{figure}
                \centering
                \scalebox{0.8}{
                \begin{tikzpicture}[x=0.75pt,y=0.75pt,yscale=-1,xscale=1]
%uncomment if require: \path (0,453); %set diagram left start at 0, and has height of 453

%Straight Lines [id:da34457073213729916] 
\draw    (103,106) -- (247.5,106) ;
%Straight Lines [id:da44347665266045655] 
\draw    (103,96) -- (103,106) ;
%Straight Lines [id:da6019994214812885] 
\draw    (247.5,106) -- (392,106) ;
%Straight Lines [id:da17166937398689874] 
\draw    (392,106) -- (536.5,106) ;
%Straight Lines [id:da30721494472723554] 
\draw    (103,106) -- (103,116) ;
%Straight Lines [id:da8784685713582623] 
\draw    (247.5,106) -- (247.5,116) ;
%Straight Lines [id:da025690641237997136] 
\draw    (247.5,96) -- (247.5,106) ;
%Straight Lines [id:da26880262337083016] 
\draw    (392,96) -- (392,106) ;
%Straight Lines [id:da6942159903691403] 
\draw    (392,106) -- (392,116) ;
%Straight Lines [id:da8107167458263409] 
\draw    (536.5,96) -- (536.5,106) ;
%Straight Lines [id:da010138467744151947] 
\draw    (536.5,106) -- (536.5,116) ;
%Shape: Circle [id:dp7657243920312456] 
\draw  [fill={rgb, 255:red, 255; green, 255; blue, 255 }  ,fill opacity=1 ][line width=1.5]  (119.25,106) .. controls (119.25,103.1) and (121.6,100.75) .. (124.5,100.75) .. controls (127.4,100.75) and (129.75,103.1) .. (129.75,106) .. controls (129.75,108.9) and (127.4,111.25) .. (124.5,111.25) .. controls (121.6,111.25) and (119.25,108.9) .. (119.25,106) -- cycle ;
%Shape: Circle [id:dp17607594654612235] 
\draw  [fill={rgb, 255:red, 255; green, 255; blue, 255 }  ,fill opacity=1 ][line width=1.5]  (357.25,106) .. controls (357.25,103.1) and (359.6,100.75) .. (362.5,100.75) .. controls (365.4,100.75) and (367.75,103.1) .. (367.75,106) .. controls (367.75,108.9) and (365.4,111.25) .. (362.5,111.25) .. controls (359.6,111.25) and (357.25,108.9) .. (357.25,106) -- cycle ;
%Shape: Circle [id:dp8173683897324276] 
\draw  [fill={rgb, 255:red, 255; green, 255; blue, 255 }  ,fill opacity=1 ][line width=1.5]  (484.25,106) .. controls (484.25,103.1) and (486.6,100.75) .. (489.5,100.75) .. controls (492.4,100.75) and (494.75,103.1) .. (494.75,106) .. controls (494.75,108.9) and (492.4,111.25) .. (489.5,111.25) .. controls (486.6,111.25) and (484.25,108.9) .. (484.25,106) -- cycle ;
%Shape: Circle [id:dp02942182783323921] 
\draw  [fill={rgb, 255:red, 245; green, 166; blue, 35 }  ,fill opacity=1 ][line width=1.5]  (175.25,106) .. controls (175.25,103.1) and (177.6,100.75) .. (180.5,100.75) .. controls (183.4,100.75) and (185.75,103.1) .. (185.75,106) .. controls (185.75,108.9) and (183.4,111.25) .. (180.5,111.25) .. controls (177.6,111.25) and (175.25,108.9) .. (175.25,106) -- cycle ;
%Shape: Circle [id:dp9707582405796078] 
\draw  [fill={rgb, 255:red, 208; green, 2; blue, 27 }  ,fill opacity=1 ][line width=1.5]  (293.25,106) .. controls (293.25,103.1) and (295.6,100.75) .. (298.5,100.75) .. controls (301.4,100.75) and (303.75,103.1) .. (303.75,106) .. controls (303.75,108.9) and (301.4,111.25) .. (298.5,111.25) .. controls (295.6,111.25) and (293.25,108.9) .. (293.25,106) -- cycle ;
%Shape: Circle [id:dp5706294133698875] 
\draw  [fill={rgb, 255:red, 245; green, 166; blue, 35 }  ,fill opacity=1 ][line width=1.5]  (403,106) .. controls (403,103.1) and (405.35,100.75) .. (408.25,100.75) .. controls (411.15,100.75) and (413.5,103.1) .. (413.5,106) .. controls (413.5,108.9) and (411.15,111.25) .. (408.25,111.25) .. controls (405.35,111.25) and (403,108.9) .. (403,106) -- cycle ;
%Straight Lines [id:da916095257870815] 
\draw    (183.5,145) -- (285.5,145) -- (294.5,145) ;
\draw [shift={(297.5,145)}, rotate = 180] [fill={rgb, 255:red, 0; green, 0; blue, 0 }  ][line width=0.08]  [draw opacity=0] (8.93,-4.29) -- (0,0) -- (8.93,4.29) -- cycle    ;
\draw [shift={(180.5,145)}, rotate = 0] [fill={rgb, 255:red, 0; green, 0; blue, 0 }  ][line width=0.08]  [draw opacity=0] (8.93,-4.29) -- (0,0) -- (8.93,4.29) -- cycle    ;
%Straight Lines [id:da7185466275874717] 
\draw    (306.5,145) -- (404.5,145) ;
\draw [shift={(407.5,145)}, rotate = 180] [fill={rgb, 255:red, 0; green, 0; blue, 0 }  ][line width=0.08]  [draw opacity=0] (8.93,-4.29) -- (0,0) -- (8.93,4.29) -- cycle    ;
\draw [shift={(303.5,145)}, rotate = 0] [fill={rgb, 255:red, 0; green, 0; blue, 0 }  ][line width=0.08]  [draw opacity=0] (8.93,-4.29) -- (0,0) -- (8.93,4.29) -- cycle    ;

% Text Node
\draw (126.5,114.65) node [anchor=north west][inner sep=0.75pt]    {$x_{1}$};
% Text Node
\draw (364.5,114.65) node [anchor=north west][inner sep=0.75pt]    {$x_{2}$};
% Text Node
\draw (491.5,114.65) node [anchor=north west][inner sep=0.75pt]    {$x_{3}$};
% Text Node
\draw (182.5,114.65) node [anchor=north west][inner sep=0.75pt]    {$u_{1}$};
% Text Node
\draw (410.25,114.65) node [anchor=north west][inner sep=0.75pt]    {$u_{3}$};
% Text Node
\draw (292.5,114.65) node [anchor=north west][inner sep=0.75pt]    {$u$};
% Text Node
\draw (330,151) node [anchor=north west][inner sep=0.75pt]    {$\leqslant \ k-1$};
% Text Node
\draw (225,151) node [anchor=north west][inner sep=0.75pt]    {$\leqslant \ k$};
% Text Node
\draw (288,45) node [anchor=north west][inner sep=0.75pt]    {\Large $P'$};
% Text Node
\draw (36,99) node [anchor=north west][inner sep=0.75pt]    {\LARGE $P$};
\draw (188,65) node [anchor=north west][inner sep=0.75pt]    {\Huge $\overbrace{\hspace{5.5cm}}$};

\end{tikzpicture}
                }
                \caption{The vertex $u_1$ (resp. $u_3$) is the only vertex in $P \cap C_{x_1}(d_1)$ (resp. in $P \cap C_{x_3}(d_3)$). The vertices~$u_1$ and~$u_3$ are at distance at most $2k-1$ from each other. The vertex~$u$ is such that its distance to~$u_1$ is at most~$k$, and its distance to~$u_3$ is at most~$k-1$.}
                \label{fig:u1u3u}
            \end{figure}
            
            So $P'$ is an induced path in the view of $u$ at distance $k$. Let us denote its length by $\ell(P')$.
            Since $P$ has length $3k-1$, by definition of $\lp$, we have: $$(\lp[u_1,i_3]-1) + \ell(P') + (\lp[u_3,i_3]-1) \geqslant 3k-1$$
            Thus, the vertex $u$ rejects at step~(vi). \qedhere
        \end{itemize}
    \end{proof}
    
    Conversely, assume that there is no induced $P_{3k-1}$ in $G$. Then, let us prove that with the certificates described above, no vertex rejects. It is straightforward that no vertex rejects at steps~(i), (ii) and~(iii).
    By contradiction, assume that a vertex $u$ rejects at step~(iv). Then, there exists an induced path $P$ of length $3k-1$ in $G_{\preccurlyeq u}$ which is included in $V_{\preccurlyeq u}$. So $P$ is also induced in $G$, which is a contradiction.
    Finally, assume that a vertex~$u$ rejects at step~(v).
    Then, there is a path $P$ induced in the view of~$u$ at distance $k$ and two vertices $v_1$, $v_2$ satisfying the conditions making $u$ reject at step~(v).
    Let $i \in \{1, \ldots, N\}$ be such that $v_1$ and $v_2$ are at distance at most $k-1$ from two distinct ECC$_i$'s, denoted by $C_{v_1}$ and $C_{v_2}$. Let $d_{v_1}$ (resp. $d_{v_2}$) be such that $v_1 \in C_{v_1}(d_{v_1})$ (resp. $v_2 \in C_{v_2}(d_{v_2})$).
    By definition of $\lp$, there exists a path $P_{v_1}$ (resp. $P_{v_2}$) which starts in $v_1$ (resp. $v_2$), which has all its other vertices in $\bigcup_{0 \leqslant d < d_{v_1}} C_{v_1}(d)$ (resp. in $\bigcup_{0 \leqslant d < d_{v_2}} C_{v_2}(d)$), and such that $\ell(P) + (\ell(P_{v_1})-1) + (\ell(P_{v_2})-1) \geqslant 3k-1$.
    So we can glue the paths $P_{v_1}$, $P$ and $P_{v_2}$ to obtain a path of length $3k-1$ in $G$, which is still induced. This is a contradiction. Thus, all the vertices accept at step~(vi).
\end{proof}

\subsection{Upper bound for paths of length $\frac{14}{3}k-2$ 
%{$\Bigl \lceil \frac{14}{3}k\Bigr \rceil-2$} 
%\textcolor{blue}{SZ: why not $-1$ ? %bc the proof is using paths on $m$ vertices which have length $m-1$.
in $\tilde{O}(n^{3/2})$}
\label{sec:P14/3}
%the prove is using paths on $m$ vertices which have length $m-1$.

%\thmThreekQuasilinear* 

In this subsection we prove the following theorem.
\thmPImprovedPathsFreeSubquadratic* 

%\begin{theorem}
%    \label{thm:P_14/3k}
%    There exists a certification scheme for $P_{\left \lceil \frac{14}{3} k \right \rceil - 1}$-free graphs with verification radius~$k$ and certificates of size $O(n^{3/2}\log^2 n)$.
%\end{theorem}

\begin{proof}
    Let $m = \left \lceil \frac{14}{3} k \right \rceil - 1$. Note that $m < \frac{14}{3} k$. Let us describe a certification scheme for $P_m$-free graphs, with certificates of size $O(n^{3/2}\log^2 n)$. Let $G = (V,E)$ be a graph and $n = |V|$.
    Let $\varepsilon = \frac{1}{2}$. We will use the definitions of \cref{sub:ecci-dfns-app} with respect to $\varepsilon$.
    We remind the reader of the following key pieces of notation.
   % For each $i \in \{1, 2\}$, 
    %
    %
    %LC: I added this, feel free to disagree and remove
    \begin{itemize}
  % $N = \lceil \frac{1}{\varepsilon}\rceil$.\\
  %$V_i:=\{u \in V \; | \; n^{(i-1)\varepsilon} \leqslant deg(u) <   n^{i\varepsilon}\}$ 
    \item $V_1$ is the set of vertices with degree less than $\sqrt{n}$. $V_2 = V \setminus V_1$.
   \item $H_1 = V$, $H_2 = V_2$. %\linda{so we need not have seperate notation for $H_i, V_i$}
   \item $\mathcal{E}_2$ is the partition of $H_2$ into ECC$_2$'s. By definition all of $V$ is in one ECC$_1$.
   \item For $C \in \cE_2$, we denote by $C(d)$ the set of vertices at distance exactly $d$ from $C$
   \item If $u$ is at distance at most $k-1$ from some ECC$_2$ $C_u$, 
    $T_G[u,2]$ is $(id(C_u), d_u)$ where $d_u$ is the distance from $u$ to $C_u$. (Recall, $C_u$ is uniquely defined.) Otherwise, $T_G[u,2] = \bot$.
    \item For each $u \in V_1$, by definition $V_{\preccurlyeq u} = V_1$. For each $u \in V_2$, by definition $V_{\preccurlyeq u} = V_1 \cup C_u$ where $C_u$ denotes the ECC$_2$ containing $u$. 
    \item $G_{\preccurlyeq u}$ is the graph obtained from $G$ by deleting all edges that do not have at least one endpoint in $V_{\preccurlyeq u}$. %See \cref{pre-curly-dfn} for further remarks on these definitions.
     \item For every vertex $v \in V_1$ such that $v$ is at distance at most $k-1$ from a (unique) ECC$_2$ $C_v$, we let $P_v$ be the longest induced path starting from $v$ with other vertices in $P_v$ strictly closer to $C_v$ than $v$. Let $\ell(P_v)$ be the length of $P$.
     \item We define $\lp$ to be an array indexed by identifiers of vertices, such that $\lp[v] = \ell(P_v)$ for every vertex $v \in \bigcup_{C \in \cE_2} \bigcup_{1 \leqslant d \leqslant k-1} C(d)$, and it is $\bot$ otherwise. 
   \end{itemize}

%(Note that it is indeed well-defined since a vertex is close to at most one ECC).

    \medskip{}
    \textbf{Certification.} 
    Let us describe the certificates given by the prover to the vertices.
    
    We give every vertex $v$ the certificates of \cref{thm:p_4k_free_subquadratic} and some additional information denoted by $\lcp$.
    We repeat the definition of the certificates of \cref{thm:p_4k_free_subquadratic} for the readers convenience:
    %LC: in future versions we should consider deleting this or making it more succinct
   First, it gives to every vertex its certificates in the local computation schemes of Theorems~\ref{thm:computation scheme for T_G} and~\ref{thm:computation scheme for g_<u}, with $\varepsilon$, so that each vertex $u$ can compute $T_G$ and $G_{\preccurlyeq u}$. This part of the certificate has size $O(n^{3/2} \log^2 n)$.
   Then the prover adds the field $\lp$ to each vertex. It has size $O(n \log n)$.

    For every $v \in V$, let us describe what information is stored in $\lcp(v)$.
    %If $v \notin \bigcup_{C \in \cE_2} \bigcup_{1 \leqslant d \leqslant k-1} C(d)$, we set $\lcp(v)=\bot$.
 %   \linda{is it possible you want $0 \leq d \leqslant k-1$ here?}
  %  \linda{just a guess}
    If $v \in V_2$ or $v \in V_1$ and $v$ has distance at least $k$ from every ECC$_2$, then we set $\lcp(v)=\bot$.
    %If $v \in V_2$ or if $v \in V_1$ and $v$ is not at distance at most $k-1$ from some ECC$_2$, we set $\lcp(v)=\bot$. \lindaPhere do we want 
    %Else, let $C_v \in \cE_2$ and $d_v \in \{1, \ldots, k-1\}$ such that $v \in C_v(d_v)$.
    Otherwise, $\lcp(v)$ consists in a table having at most $n$ rows, and $4$ columns. The rows are indexed by vertices of~$G$. The content of each cell is the following. Let $v' \in C_v(d_v) \setminus \{v\}$.
    Let $Q$ denote the set $\bigcup_{0 \leqslant d < d_v} C_v(d)$.
    In other words, $Q$ is the set of vertices in $G$ with that are strictly closer to $C_v$ than $v$.
    \begin{itemize}
        \item In $\lcp(v)[v',1]$, the prover writes the length of the longest induced path starting from $v$ and having all the other vertices in $Q \cup \{v'\}$.

        \item In $\lcp(v)[v',2]$, the prover writes the length of the longest induced path starting from $v'$ and having all the other vertices in $Q \setminus N[v]$.

        \item In $\lcp(v)[v',3]$, the prover writes the length of the longest induced path starting at $v$, ending in $v'$, and having all the other vertices in $Q$.

        \item In $\lcp(v)[v',3]$, the prover writes the maximum sum of the lengths of two disjoint induced paths, which are anti-complete one to each other, which start respectively at $v$ and $v'$, and which have all the other vertices in $Q$.
    \end{itemize}

    Note that $\lp$ and $\lcp$ both have size $O(n \log n)$. Thus, the overall size of each certificate is $O(n^{3/2}\log^2 n)$.

\medskip{}
\textbf{Verification.}
The verification of each vertex $u \in V$ consists in the following steps.
First $u$ performs the $m$-pathcheck as described in \cref{def:mpathcheck}.
(Recall $m = \lceil 14k/3 \rceil  -1$.)
We will refer to the steps of \cref{def:mpathcheck} as step (i) - step \ref{step:uselp-simple}.

Then $u$ performs the following additional steps:
    \begin{enumerate}[(i)]
    \setcounter{enumi}{5}
    
    \item $u$ checks the correctness of $\lcp(u)$. %\linda{here is $u \in V_2$ or a general vertex}
    \begin{itemize}
    \item Since $V_1, V_2$ are defined only by the degrees of vertices, $u$ knows whether it is in $V_1$ or $V_2$.
        If $u \in V_2$, and $T_G[u,2] \neq \bot$, $u$ rejects. %LC: I think we actually don't need this but it made it somewhat clearer for me
        If $u \in V_1$, $u$ checks whether $\lcp(u) = \bot$ if and only if $T_G[u,2] = \bot$. If not $u$, rejects.
        \label{step:check-constrained-parti one}
    \item  Suppose $\lcp(u) \neq \bot$. Then by the previous item and the definition of $T_u$, there is a (unique) $C_u \in \cE_2$ and $d_u \in \{1, \ldots, k-1\}$ such that $u$ is at distance exactly $d_u$ from $C_u$.
    %\linda{I though that this is actually for $d_u \in \{0, \dots, k-1\}$?}
    Let $u' \in C_u$ at distance $d_u$ from $u$.
    By Remark~\ref{rem:computation witnessed graph other vertices}, $u$ can compute $G_{\preccurlyeq u'}$. 
    In particular, $u$ knows the subgraph of $G$ induced by $\bigcup_{0 \leqslant d \leqslant k-1} C_u(d)$.
    %ence, for each $v \in C_u(d) \setminus \{u\}$ it can 
    and this is sufficient for $u$ to check the correctness of $\lcp(u)$ by definition of $\lcp$ (and $u$ rejects if it is not correct.) \label{step:check-constrained} 
    \end{itemize}
    % If $\lcp(u) \neq \bot$, $u$ performs the following additional checks.
    % First 
    % %Recall, $TG[u,2] = \bot$ if and only if $u$ is at distance at most $k-1$ from some ECC$_2$.
    
    % To do so, $u$ first checks that $\lcp(u)=\bot$ if and only if $u \notin \bigcup_{C \in \cE_2} \bigcup_{1 \leqslant d \leqslant k-1} C(d)$ (note that $u$ knows it thanks to $T_G$).
    % Then, if $u \in \bigcup_{C \in \cE_2} \bigcup_{1 \leqslant d \leqslant k-1} C(d)$, let $C_u \in \cE_2$ and $d_u \in \{1, \ldots, k-1\}$ such that $u \in C_u(d_u)$. Let $u' \in C_u$ at distance $d_u$ from $u$. By Remark~\ref{rem:computation witnessed graph other vertices}, $u$ can compute $G_{\preccurlyeq u'}$. In particular, $u$ knows the subgraph of $G$ induced by $\bigcup_{0 \leqslant d \leqslant k-1} C_u(d)$, and this is sufficient for $u$ to check the correctness of $\lcp(u)$ (and $u$ rejects if it is not correct).\label{step:check-constrained}

   % If no vertex rejects at this point, then~$\lcp(u)$ is correctly written in the certificate of every vertex~$u$.
    \end{enumerate}

    We may assume no vertex has rejected at this point, and thus $u$ has access to a (correct) $T_G$, ~$G_{\preccurlyeq u}$,  $\lp$ and $\lcp(u)$.
    
    \begin{enumerate}[(i)]
    \setcounter{enumi}{6}
    \item If $u \in V_2$, let us denote by $C_u$ its ECC$_2$. It rejects if it sees an induced path $P_{\text{start}}$ in $G[C_u \cup V_1]$, and such that the following conditions are satisfied:
        \begin{itemize}
            \item $P_{\text{start}}$ contains at least one vertex from $C_u$
            \item $P_{\text{start}}$ goes from a vertex $v$ to a vertex $w$, such that $T_G[v,2] = (id(C_v), d_v)$, $T_G[w,2] = (id(C_w), d_w)$, %and $C_u$, $C_v$, $C_w$ are three distinct ECC$_2$'s
           % \item 
            and $\id(C_u), \id(C_v), \id(C_w)$ are all distinct. 
            Since $T_G$ is correct by step \ref{step:precompute},  $C_u, C_v, C_w$ are all distinct, $v$ is at distance at most $d_v \leq k-1$ from $C_v$ and $w$ is at distance at most $d_w \leq k-1$ from $C_w$.
            \item $v$ is the only vertex in $P_{\text{start}} \cap (\bigcup_{d \leqslant d_v} C_v(d))$, and $w$ is the only vertex in $P_{\text{start}} \cap (\bigcup_{d \leqslant d_w} C_w(d))$.
            In other words all other vertices of $P_{\text{start}}$ are further from $C_v$ (resp. $C_w$) than $v$ (resp. $w$). 
            $u$ checks this by looking at the entries of $T_G$ for the vertices in $P$.
            \item $\ell(P_{\text{start}}) + (\lp[v]-1) + (\lp[w]-1) \geqslant m$
        \end{itemize}

    This case is represented on Figure~\ref{fig:P14/3k case three ECCs}.
\label{step:threeECCS}
    \begin{figure}
        \centering
        \scalebox{0.5}{
            \begin{tikzpicture}[x=0.75pt,y=0.75pt,yscale=-1,xscale=1]
%uncomment if require: \path (0,711); %set diagram left start at 0, and has height of 711

%Shape: Polygon Curved [id:ds011239172530058128] 
\draw  [color={rgb, 255:red, 255; green, 255; blue, 255 }  ,draw opacity=1 ][fill={rgb, 255:red, 74; green, 144; blue, 226 }  ,fill opacity=0.29 ] (913.5,173) .. controls (945.5,205) and (954.5,286) .. (873.5,299) .. controls (792.5,312) and (737.5,226) .. (779.5,177) .. controls (821.5,128) and (881.5,141) .. (913.5,173) -- cycle ;
%Shape: Polygon Curved [id:ds10356896315670339] 
\draw  [color={rgb, 255:red, 255; green, 255; blue, 255 }  ,draw opacity=1 ][fill={rgb, 255:red, 74; green, 144; blue, 226 }  ,fill opacity=0.29 ] (236,166.5) .. controls (247,209.5) and (231,298.5) .. (122,304.5) .. controls (13,310.5) and (39,227.5) .. (89,184.5) .. controls (139,141.5) and (225,123.5) .. (236,166.5) -- cycle ;
%Shape: Polygon Curved [id:ds15014238114363032] 
\draw  [color={rgb, 255:red, 255; green, 255; blue, 255 }  ,draw opacity=1 ][fill={rgb, 255:red, 74; green, 144; blue, 226 }  ,fill opacity=0.29 ] (611,534.5) .. controls (587,588.5) and (467,643.5) .. (453,567.5) .. controls (439,491.5) and (547,429.5) .. (590,432.5) .. controls (633,435.5) and (635,480.5) .. (611,534.5) -- cycle ;
%Curve Lines [id:da8540509608803334] 
\draw  [dash pattern={on 4.5pt off 4.5pt}]  (142,316.5) .. controls (250,309.5) and (264,219.5) .. (256,178.5) ;
%Curve Lines [id:da3211266236085093] 
\draw  [dash pattern={on 4.5pt off 4.5pt}]  (164,330.5) .. controls (272,323.5) and (286,233.5) .. (278,192.5) ;
%Curve Lines [id:da24244603834902823] 
\draw  [dash pattern={on 4.5pt off 4.5pt}]  (433,551.5) .. controls (421,475.5) and (529,411.5) .. (570,416.5) ;
%Curve Lines [id:da5481734334057874] 
\draw  [dash pattern={on 4.5pt off 4.5pt}]  (406,534.5) .. controls (394,458.5) and (502,394.5) .. (543,399.5) ;
%Curve Lines [id:da34194867788139816] 
\draw  [dash pattern={on 4.5pt off 4.5pt}]  (195,389.5) .. controls (308,401.5) and (360,321.5) .. (354,252.5) ;
%Curve Lines [id:da0010803752541739264] 
\draw [color={rgb, 255:red, 25; green, 25; blue, 116 }  ,draw opacity=1 ][line width=1.5]    (151,223) .. controls (156,219.25) and (259,308) .. (248,321) .. controls (237,334) and (162,270.5) .. (154,280.5) .. controls (146,290.5) and (186,323) .. (202,335) .. controls (218,347) and (262,368.5) .. (282,353.5) .. controls (302,338.5) and (192,244.5) .. (217,219.5) .. controls (242,194.5) and (343.44,325.65) .. (369,360.5) .. controls (394.56,395.35) and (394,419.5) .. (378,426.5) .. controls (362,433.5) and (323,412.5) .. (324,446.5) .. controls (325,480.5) and (514,548.5) .. (539,523.5) .. controls (564,498.5) and (427,429.5) .. (438,413.5) .. controls (449,397.5) and (530,473.5) .. (546,462.5) .. controls (562,451.5) and (454,389.5) .. (479,364.5) .. controls (504,339.5) and (536.5,388) .. (561.5,363) .. controls (586.5,338) and (607.5,300) .. (667.5,302) .. controls (727.5,304) and (786.5,211) .. (806.5,226) .. controls (826.5,241) and (756.5,327) .. (761.5,344) .. controls (766.5,361) and (781.5,349) .. (806.5,324) ;
%Shape: Circle [id:dp24862589653623524] 
\draw  [fill={rgb, 255:red, 208; green, 2; blue, 27 }  ,fill opacity=1 ] (344,322.5) .. controls (344,320.01) and (341.99,318) .. (339.5,318) .. controls (337.01,318) and (335,320.01) .. (335,322.5) .. controls (335,324.99) and (337.01,327) .. (339.5,327) .. controls (341.99,327) and (344,324.99) .. (344,322.5) -- cycle ;
%Shape: Circle [id:dp9075407860747158] 
\draw  [fill={rgb, 255:red, 208; green, 2; blue, 27 }  ,fill opacity=1 ] (595,491.5) .. controls (595,489.01) and (592.99,487) .. (590.5,487) .. controls (588.01,487) and (586,489.01) .. (586,491.5) .. controls (586,493.99) and (588.01,496) .. (590.5,496) .. controls (592.99,496) and (595,493.99) .. (595,491.5) -- cycle ;
%Curve Lines [id:da10435576009057745] 
\draw  [dash pattern={on 4.5pt off 4.5pt}]  (743.5,199) .. controls (714.5,234) and (765.5,329) .. (844.5,316) ;
%Curve Lines [id:da9521274178812967] 
\draw  [dash pattern={on 4.5pt off 4.5pt}]  (718.5,215) .. controls (689.5,250) and (740.5,345) .. (819.5,332) ;
%Curve Lines [id:da03522278432611459] 
\draw  [dash pattern={on 4.5pt off 4.5pt}]  (651.5,261) .. controls (622.5,296) and (673.5,391) .. (752.5,378) ;
%Shape: Circle [id:dp5100757781669953] 
\draw  [fill={rgb, 255:red, 208; green, 2; blue, 27 }  ,fill opacity=1 ] (651,303.5) .. controls (651,301.01) and (648.99,299) .. (646.5,299) .. controls (644.01,299) and (642,301.01) .. (642,303.5) .. controls (642,305.99) and (644.01,308) .. (646.5,308) .. controls (648.99,308) and (651,305.99) .. (651,303.5) -- cycle ;

% Text Node
\draw (137,181.4) node [anchor=north west][inner sep=0.75pt]    {\huge $C_{v} \ $};
% Text Node
\draw (508,553.4) node [anchor=north west][inner sep=0.75pt]    {\huge $C_{u}$};
% Text Node
\draw (86,306.9) node [anchor=north west][inner sep=0.75pt]    {\huge $C_{v}( 1)$};
% Text Node
\draw (109,324.4) node [anchor=north west][inner sep=0.75pt]    {\huge $C_{v}( 2)$};
% Text Node
\draw (573,410.4) node [anchor=north west][inner sep=0.75pt]    {\huge $C_{u}( 1)$};
% Text Node
\draw (546,389.4) node [anchor=north west][inner sep=0.75pt]    {\huge $C_{u}( 2)$};
% Text Node
\draw (133,379.4) node [anchor=north west][inner sep=0.75pt]    {\huge $C_{v}( d_{v})$};
% Text Node
\draw (296,223.4) node [anchor=north west][inner sep=0.75pt]    {\Huge $\textcolor[rgb]{0.1,0.1,0.45}{P_v}$};
% Text Node
\draw (666,255.4) node [anchor=north west][inner sep=0.75pt]    {\Huge $\textcolor[rgb]{0.1,0.1,0.45}{P_w}$};
% Text Node
\draw (330,484.4) node [anchor=north west][inner sep=0.75pt]    {\Huge $\textcolor[rgb]{0.1,0.1,0.45}{P_{\text{start}}}$};
% Text Node
\draw (355,304.4) node [anchor=north west][inner sep=0.75pt]    {\huge $\textcolor[rgb]{0.82,0.01,0.11}{v}$};
% Text Node
\draw (598,472.4) node [anchor=north west][inner sep=0.75pt]    {\huge $\textcolor[rgb]{0.82,0.01,0.11}{u}$};
% Text Node
\draw (617,282.4) node [anchor=north west][inner sep=0.75pt]    {\huge $\textcolor[rgb]{0.82,0.01,0.11}{w}$};
% Text Node
\draw (757,371.4) node [anchor=north west][inner sep=0.75pt]    {\huge $C_{w}( d_{w})$};
% Text Node
\draw (848,305.9) node [anchor=north west][inner sep=0.75pt]    {\huge $C_{w}( 1)$};
% Text Node
\draw (824,327.4) node [anchor=north west][inner sep=0.75pt]    {\huge $C_{w}( 2)$};
% Text Node
\draw (840,205.4) node [anchor=north west][inner sep=0.75pt]    {\huge $C_{w} \ $};

\end{tikzpicture}}

        \caption{The case making $u$ reject at step~(vii) in the certification scheme of Theorem~\ref{thm:P_143k_free_subquadratic}}
        \label{fig:P14/3k case three ECCs}
    \end{figure}
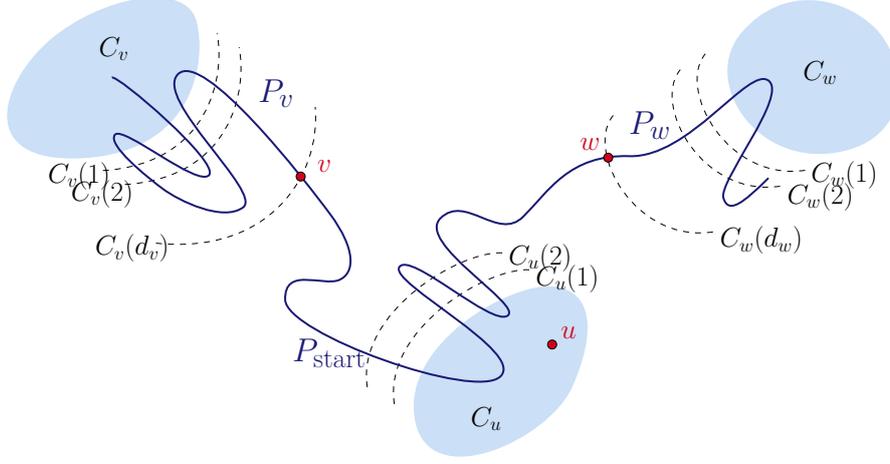
    
    \item \label{step:14/3:finalcase} Finally, if $u \in V_1$ and $u$ is at distance $d_u \in \{1, \dots, k-1\}$ from some vertex $u'$ in some ECC$_2$ $C_u$ we perform the following (recall $u$ can check this condition because $u$ has access to $T_G$).

    %if $u \in \bigcup_{C \in \cE_2} \bigcup_{1 \leqslant d \leqslant k-1} C(d)$, 
    
   % let $C_u \in \cE_2$ and $d_u \in \{1, \ldots, k-1\}$ such that $u \in C_u(d_u)$.
    %Let $u' \in C_u$ at distance $d_u$ from $u$.
    The vertex $u$ computes $G_{\preccurlyeq u'}$ (this is possible by Remark~\ref{rem:computation witnessed graph other vertices}), so $u$ knows the graph $G[C_u \cup V_1]$. 
    Then, $u$ rejects if there exists an induced path $P_{\text{start}}$ in $G[C_u \bigcup V_1]$, $C_v \in \cE_2 \setminus \{C_u\}$, $d_v \in \{1, \ldots, k-1\}$ and two vertices $v,v' \in C_v(d_v)$, with $v$ at distance at most $k$ from $u$, such that any of the following set of conditions is satisfied (see Figure~\ref{fig:P14/3k cases (viii)} for an illustration):
    \begin{enumerate}
        \item 
        \begin{itemize}
            \item $P_{\text{start}}$ ends in $v$ and is anti-complete to $v'$
            \item $v$ is the only vertex in $P_{\text{start}} \cap \left(\bigcup_{d \leqslant d_v} C_v(d)\right)$
            \item $\ell(P_{\text{start}}) + (\lcp(v)[v',1]-1) \geqslant m$
        \end{itemize}
        \vspace{0.2cm}
        \item
        \begin{itemize}
            \item $P_{\text{start}}$ ends in $v'$ and passes through $v$
            \item $v,v'$ are the only vertices in $P_{\text{start}} \cap \left(\bigcup_{d \leqslant d_v} C_v(d)\right)$
            \item $\ell(P_{\text{start}}) + (\lcp(v)[v',2]-1) \geqslant m$
        \end{itemize}
        \vspace{0.2cm}
        \item
        \begin{itemize}
            \item $P_{\text{start}}$ ends in $v$, and there exists an induced path in $G[C_u \cup V_1]$, denoted by $P_{\text{end}}$, which starts in $v'$ and is anti-complete to $P_{\text{start}}$
            \item $v,v'$ are the only vertices in $(P_{\text{start}} \cup P_{\text{end}}) \cap \left(\bigcup_{d \leqslant d_v} C_v(d)\right)$
            \item $\ell(P_{\text{start}}) + (\lcp(v)[v',3]-2) + \ell(P_{\text{end}}) \geqslant m$
        \end{itemize}
        \vspace{0.2cm}
        \item
        \begin{itemize}
            \item $P_{\text{start}}$ starts in $v$ and ends in $v'$
            \item $v,v'$ are the only vertices in $P_{\text{start}} \cap \left(\bigcup_{d \leqslant d_v} C_v(d)\right)$
            \item $\ell(P_{\text{start}}) + (\lcp(v)[v',4]-2) \geqslant m$
        \end{itemize}
    \end{enumerate}

    \begin{figure}[!ht]
                    \centering
                    \input{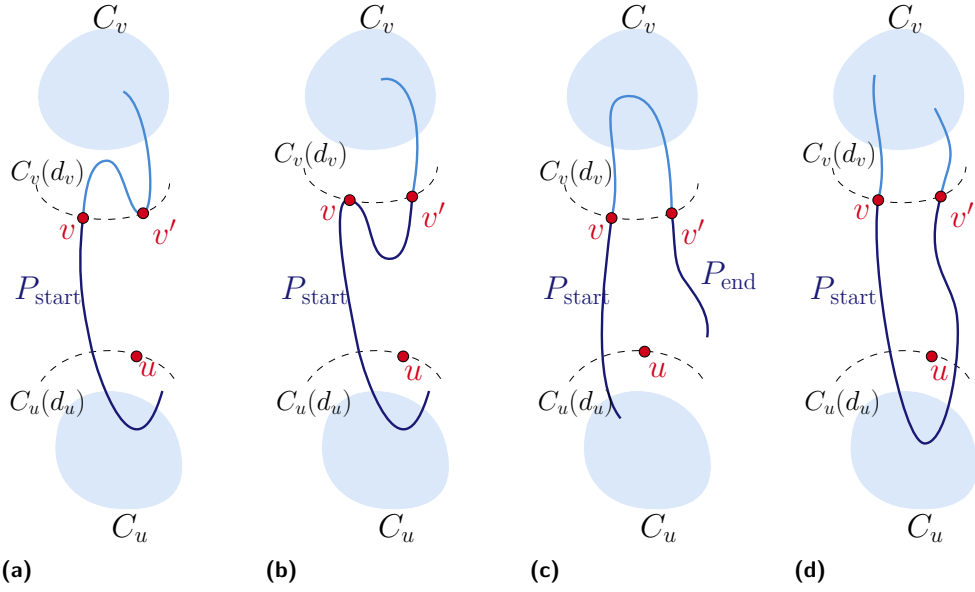}
                    \caption{The four cases making $u$ reject at step~(viii) of the verification algorithm in  Theorem~\ref{thm:P_143k_free_subquadratic}. In each of these cases, the vertices $u$ and $v$ are at distance at most $k$ (which enables $u$ to see the certificate of $v$).}
                    \label{fig:P14/3k cases (viii)}
                \end{figure}
    
    \item If $u$ did not reject in the previous steps, it accepts.
\end{enumerate}

\medskip{}
\textbf{Correctness.}
Let us prove that this certification scheme is correct. 
First let us assume for a contradiction, that $G$ contains an induced path $P$ of length $m$ and that every vertex accepts.
%First, assume that $G$ has an induced path $P$ of length $m$, and let us show that, for any assignment of certificates to the vertices, at least one vertex rejects. 
Using the $m$-pathcheck, every vertex $u$ can correctly compute $T_G$, $G_{\preccurlyeq u}$, and knows that the values of $\lp$ and $\lcp(u)$ written by the prover in the certificates are correct.
Then, let us prove the following Claims~\ref{claim:14/3k 1}, \ref{claim:14/3k 2} and~\ref{claim:14/3k 3}.

\begin{claim}
    \label{claim:14/3k 1}
   % $P$ can be decomposed into at most three subpaths $P_1,  such that 
   % There are at most three distinct $C \in \cE_2$ such that $P \cap C \neq \emptyset$.
    There are three subpaths $P_1, P_2, P_3$ of $P$ such that $P = P_1 \cup P_2 \cup P_3$ and for each $i \in \{1,2,3\}$ there is a $C_i \in \cE_2$ such that every vertex of $V(P_i) \subseteq C_i \cup V_1$.
\end{claim}

\begin{subproof}
    This simply follows from the fact that $P$ has length $m < \frac{14}{3}k$ and that two distinct ECC$_2$'s are at distance at least $2k$ from each other by definition of extended connected component. 
\end{subproof}

It follows from \cref{claim:14/3k 1} that there are at most three distinct $C \in \cE_2$ such that $P \cap C \neq \emptyset$.

\begin{claim}\label{claim:14/3k v2}
    $P\cap V_2 \neq \emptyset$
\end{claim}
\begin{subproof}
    Recall $V_1 \subseteq V_{\preccurlyeq u}$ for every vertex $u \in V$.
    In step \ref{step:visiblempath} of the $m$-pathcheck, every vertex checks whether there is an induced path of length $m$ with vertex set in  $V_{\preccurlyeq u}$.
\end{subproof}

By applying \cref{claim:4k-1 2} for $m$, we obtain that:
\begin{claim}
    \label{claim:14/3k 2}
    There are at least two distinct $C \in \cE_2$ such that $|P \cap C| \geqslant 2$.
\end{claim}
%no need to repeat this proof
% \begin{subproof}
%    % If $P \subseteq V_1$, then every vertex $u \in P$ rejects at step~(v). Else, we have $P \cap V_2 \neq \emptyset$. 
%     Let $C \in \cE_2$ such that $|P \cap C|$ is maximized.
%     By assumption, for every $C' \in \cE_2$ such that $C' \neq C$, we have $|P \cap C'| \leqslant 1$, so every $u \in C$ rejects at step~(v).
%\end{subproof}
%\linda{I believe up to here}
\begin{claim}
    \label{claim:14/3k 3}
    If there are exactly three $C \in \cE_2$ such that $P \cap C \neq \emptyset$, at least one vertex rejects.
\end{claim}

% Sebastian's original proof
%begin{proof}
%     Assume that there are exactly three $C \in \cE_2$ such that $P \cap C \neq \emptyset$. Let us denote these three ECC$_2$'s by $C_u$, $C_v$, $C_w$.
%     We can also assume that $P$ has at least $5$ vertices in $C_u\cup C_v \cup C_w$ (otherwise, the conditions of Claim~\ref{claim:14/3k 2} are satisfied so at least one vertex rejects).
%     Assume that $P$ passes through $C_v$, $C_u$, and $C_w$ in that order.
%     Since two distinct ECC$_2$'s are at distance at least $2k$ from each other, then $P$ contains at least $k-1$ vertices in $\bigcup_{1 \leqslant d \leqslant k-1} C_v(d)$, at least $2k-2$ vertices in $\bigcup_{1 \leqslant d \leqslant k-1} C_u(d)$, and at least $k-1$ vertices in $\bigcup_{1 \leqslant d \leqslant k-1} C_w(d)$.
%     Thus, $P$ contains at most $\frac{5}{3}k-2$ vertices in $\bigcup_{1 \leqslant d \leqslant k-1} C_v(d)$ (resp. in $\bigcup_{1 \leqslant d \leqslant k-1} C_w(d)$), so there exists $d_v \in \{1, \ldots, k-1\}$ (resp. $d_w \in \{1, \ldots, k-1\}$) such that $|P \cap C_v(d_v)| = 1$ (resp. $|P \cap C_w(d_w)| = 1$).
%     Let us denote by $v$ (resp. by $w$) the vertex such that $P \cap C_v(d_v) = \{v\}$ (resp. $P \cap C_w(d_w) = \{w\}$).
%     Then, any vertex $u \in C_u$ rejects, because $u$ sees the part of $P$ which goes from $v$ to $w$, and it satisfies the conditions making $u$ rejects at step~(vii).
% \end{proof}

\begin{proof}
    Assume that there are exactly three $C \in \cE_2$ such that $P \cap C \neq \emptyset$. Let us denote these three ECC$_2$'s by $C_1$, $C_2$, $C_3$.
    % By \cref{claim:14/3k 2}, we can also assume that $P$ has at least $5$ vertices in $C_1\cup C_2 \cup C_3$. %\linda{where do we use this assumption tho?}\textcolor{blue}{SZ: to handle the small values of $k$ (eg $k=2$), else you get}
    Without loss of generality, there is a subpath of $P$ passing through $C_1, C_3, C_2$, in that order. 
    For $x \in \{1,2,3\}$, let $B_x$ denote the set of vertices of $V(P)$ in $\bigcup_{1 \leqslant d \leqslant k-1} C_x(d)$.
    Since any two distinct ECC$_2$'s are at distance at least $2k{-1}$ from each other, $|B_1|, |B_2| \geqslant k-1$ and $|B_3| \geq 2k-2$.
    By definition of ECC, $B_1, B_2, B_3$ are pairwise disjoint.
    By \cref{claim:14/3k 2}, we can also assume that $P$ has at least $5$ vertices in $C_1\cup C_2 \cup C_3$.
    And by definition, $C_1 \cup C_2 \cup C_3$ is disjoint from $B_1 \cup B_2 \cup B_3$.
    By definition, $|V(P)| = m \leq  \frac{14}{3} k$.
    So by combining these observations we obtain:
    \begin{equation}\label{eq:b1b2}
        |B_1|,|B_2| \leq V(P) - 3k{+}3  - 5 < 2k-2
    \end{equation}   
  % %  \textcolor{blue}{SZ: here that's not true for small values of $k$ \linda{ah ok thanks!}, that's why you should count the 5 more vertices in $C_1 \cup C_2 \cup C_3$. If you count the 5 vertices you get: $|B_1| < \frac{14}{3}k - (3k-3) - 5 \leqslant \frac{5}{3}k - 2 < 2k-2$}. 
  Let $i \in \{1,2\}$.
    From the definition of ECC, there must be a vertex at distance $k-1$ from $C_i$ in $B_i$.
    Since $P$ contains a vertex of $C_i$ it follows that for each $d \in \{1,2, \dots, k-1\}$ there is some vertex $b \in B_i$ at distance exactly $d$ from $C_i$.
    Hence, it follows from \eqref{eq:b1b2} that there exists some 
     $d_i \in \{1, \ldots, k-1\}$ such that $|P \cap C_i(d_i)| = 1$. 
     %\textcolor{blue}{SZ: maybe use always the notation $x$ or always $i$ ?} thanks!
    Let $v_i$ denote the unique element of $P \cap C_i(d_i)$. 
    Let $P_{\text{start}}$ denote the $v_1v_2$-path of $P$.
    
     Let $u \in C_3$ be arbitrary. We will show that $u$ rejects at step (\ref{step:threeECCS}).
    By definition of ECC$_2$, the only vertices of $V_2$ in $P_{\text{start}}$ are contained in $C_3$.
    So in particular, $V(P_{\text{start}}) \subseteq P_{\preccurlyeq u}$
    Hence, $u$ will consider $P_{\text{start}}$ in step \ref{step:threeECCS} and it will reject if $\ell(P_{\text{start}}) + (\lp[v_1]-1) + (\lp[v_2]-1) \geqslant m$.

     By definition, there the graph obtained from $P$ by deleting all vertices of $P_{\text{start}}$ except for $v_1, v_2$ consists of two paths $P_1, P_2$ such that $v_1$ is an end of $P_1$ and $v_2$ is an end of $P_2$.
     Let $i \in \{1,2\}$.
     By \cref{claim:14/3k 1} we have that $P_i \subseteq V_i \cup C_i$.
     Since $P_i$ contains a vertex of $C_i$, by the definition of distance, there can be no vertex in $V(P_i)$ at distance strictly greater than $d_i$ from $v_i$.
     It follows that $V(P_i) \setminus \{v_i\} \subseteq \bigcup_{0 \leq d < d_i} C_i(d)$.
     Hence by definition, $\lp(v_i) \geq |V(P_i)| - 1$.

     Since $P = P_1 \cup P_2 \cup P_3$, it follows that $\ell(P_{\text{start}}) + (\lp[v_1]-1) + (\lp[v_2]-1) \geqslant m$ and therefore $u$ rejects in step~(\ref{step:threeECCS}).         
    % % By definition, there the graph obtained from $P$ by deleting all vertices of $P_{\text{start}}$ except for $v_1, v_2$ consists of two paths $P_1, P_2$ such that $v_1$ is an end of $P_1$ and $v_2$ is an end of $P_2$.
    % % We claim that for $i \in \{1,2\}$, $\lp(v_i) \geq |V(P_i)| - 1$.
    % By definition of $\lp$ it is enough to show 

    % %Hence, by our choice of $P_{\text{start}}$, $u$ will see $P_{\text{start}}$ and reject at step~(vii).
    % \linda{I think we need more explanation for why $u$ rejects... I.e. I think the proof is presuming that the value of $\lp(v_1)$ should be at least as long as the subpath of $P$ ending at $v_1$ and the same thing for $v_2$, but need some explanation for this}
    % \textcolor{blue}{SZ: you can decompose $P$ in three parts: the part before $v_1$, then $P_\text{start}$, and then the part after $v_2$. The sum of the lengths of these three parts is $m$, so if $\lp$ is correctly written in the certificates (otherwise a vertex rejects), the inequality in (vii) is satisfied}
    % \linda{sorry I should have made my question clearer. It is not crystal clear to me why $V(P_i) \setminus \{v_i\} \subseteq \bigcup_{0 \leq d < d_i} C_i(d)$}
    % \linda{actually I think its not correct under the current arguments, need to patch this}
  %  Then, any vertex $u \in C_u$ rejects, because $u$ sees the part of $P$ which goes from $v_1$ to $v_2$, and it satisfies the conditions making $u$ reject at step~(vii).
\end{proof}

Applying Claims~\ref{claim:14/3k 1}, \ref{claim:14/3k 2} and~\ref{claim:14/3k 3}, we can assume that there exactly two ECC$_2$'s which intersect~$P$, both on at least two vertices.
Let us denote these two ECC$_2$'s by $C_u$ and $C_v$.
\begin{claim}\label{claim:14/3 no unique distances}
    For all $d \in \{1, \ldots, k-1\}$, we have $|P \cap C_v(d)| \geqslant 2$, and similarly $|P \cap C_u(d)| \geqslant 2$.
\end{claim}
\begin{subproof} %LC: Arguably unecessary, remove in future versions
Since $C_1, C_2$ both meet $V(P)$, by definition of ECC for each $C \in \{C_v, C_u\}$ there is at least one vertex in $P$ at distance exactly $d$ from $C$.
If there exists $d \in \{1, \ldots, k-1\}$ and $x \in \{u,v\}$ such that $|P \cap C_x(d)| = 1$, then some vertex will reject by \cref{lem:4k-2ECC-1-uniquedistance}.
\end{subproof}

For brevity, let $\mathcal{B}_u := \bigcup_{0 \leqslant d \leqslant k-1} C_u(d)$ and let $\mathcal{B}_v := \bigcup_{0 \leqslant d \leqslant k-1} C_v(d)$.
Let us denote by $\alpha$ the number of vertices in $P \setminus (\mathcal{B}_u \cup \mathcal{B}_v)$.
In other words, $\alpha$ denotes the number of vertices in $P$ that are not \say{close} (at most $k-1$ away from) $C_u$ or $C_v$.
Thus by \cref{claim:14/3 no unique distances}, 
\begin{equation}\label{equation:alpha-2k/3}
\alpha \leq m - 4(k-1) \leq \frac{14}{3}k - 4k -4 < \frac{2}{3}k.
\end{equation}
% Since $m < \frac{14}{3}k$, and since for all $d \in \{0, \ldots, k-1\}$ we have $|P \cap C_u(d)| \geqslant 2$ and $|P \cap C_v(d)| \geqslant 2$, then $\alpha < \frac{2}{3}k$.
Recall, that by definition $\mathcal{B}_u$ and $\mathcal{B}_v$ are disjoint.
So, $|V(P)| = \lceil 14k / 3 \rceil -1 = \alpha + |P \cap \mathcal{B}_u| + |P \cap \mathcal{B}_v|$.
Hence, by the pigeonhole principle and the symmetry between $C_u$ and $C_v$
we may assume that 
%Moreover, we have either $|P \cap \left( \bigcup_{0 \leqslant d \leqslant k-1} C_u(d) \right)| < \frac{7}{3}k - \frac{\alpha}{2}$

%, or $|P \cap \left( \bigcup_{0 \leqslant d \leqslant k-1} C_v(d) \right)| < \frac{7}{3}k - \frac{\alpha}{2}$.

\begin{equation}
\label{eq:14/3k}
\left|P \cap \mathcal{B}_v \right| < \frac{7}{3}k - \frac{\alpha}{2}
\end{equation}

Let us prove the following:

\begin{claim}
    \label{claim:14/3k 4}
    There exists $u \in C_u(k-1)$, $d_v \in \{1, \ldots, k-1\}$ and $v \in P \cap C_v(d_v)$ such that $v$ is at distance at most $k$ from $u$, and $|P \cap C_v(d_v)|=2$.
\end{claim}
\begin{subproof}
    Since $P$ passes through $C_u$ and $C_v$, let us consider a connected subgraph of $P$, which is a path denoted by $P'$, such that $P'$ starts in $C_u$ and ends in $C_v$.
    Let $u$ be the last vertex in $P' \cap C_u(k-1)$ before $P'$ passes through $C_v$.
    Since $P'$ goes from $u \in C_u(k-1)$ to $C_v$, then it contains at least $k$ other vertices after $u$ (at least one in $C_v(d)$ for each $d \in \{0, \ldots, k-1\}$).
    Let us denote by $X$ the set of the $k$ vertices following $u$ in $P'$.
    All the vertices in $X$ are at distance at most $k$ from $u$.
    By definition of~$u$, we have $X \cap \left( \bigcup_{0 \leqslant d \leqslant k-1} C_u(d) \right) = \emptyset$.
    And since $\alpha < \frac{2}{3}k$ and $|X| = k$, we have that $X \cap \left( \bigcup_{1 \leqslant d \leqslant k-1} C_v(d) \right) \neq \emptyset$.
    We may assume that for each vertex $x_1 \in X \cap \mathcal{B}_v$ 
    there are at least two other distinct vertices $x_2, x_3 \in V(P) \cap \mathcal{B}_v$ such that $x_1, x_2, x_3$ are all the same distance from $C_v$.
    
    Let $D$ be the set of distances $d \in \{1,2, \dots, k\}$ for which some $x \in X$ is at distance $d$ from $C_v$ (i.e. $x \in C_v(d)$).
    % Let $I:= \{ d \in \{1, \ldots, k-1\} \; ; \; X \cap C_v(d) \neq \emptyset\}$. 
    For each $d \in D$, we have $|P \cap C_v(d)| \geqslant 3$.
    And, by definition of $X$, $D$ and $\alpha$, we also have $\left|\left(\bigcup_{d \in D} C_v(d) \right) \cap X\right| \geqslant |X| - |X \setminus (\mathcal{B}_u \cup \mathcal{B}_v)| \geqslant
    k-\alpha$.
    Hence, $\left|\left(\bigcup_{d \in D} C_v(d) \right) \cap V(P)\right| \geqslant 
    \max (3|D|, k - \alpha )$
    Moreover for all $d \notin D$, we have $|V(P) \cap C_v(d)| \geqslant 2$.
    Thus, we get:
    \begin{equation}
          |V(P) \cap \mathcal{B}_v| 
             = \left|\bigcup_{d \in D} C_v(d) \cap V(P) \right| + \left|\bigcup_{d \notin D} C_v(d) \cap V(P) \right| 
             \geqslant \max(3|D|,k-\alpha) + 2(k-|D|)
    \end{equation}
    % \begin{align*}
    %     \left |\mathcal{B}_v| & 
    %         =|\bigcup_{0 \leqslant d \leqslant k-1} C_v(d)\right|
    %         & = \left|\bigcup_{d \in I} C_v(d)\right| + \left|\bigcup_{d \notin I} C_v(d)\right| \\
    %         & \geqslant \max(3|I|,k-\alpha) + 2(k-|I|)
    % \end{align*}
    
    There are two cases:
    \begin{itemize}
        \item if $|D| \geqslant \frac{k}{3} - \frac{\alpha}{2}$, then we have:
        \begin{align*}
            \left|\mathcal{B}_v \cap V(P) \right|
            & \geqslant 3|D| + 2(k-|D|) \\
            & = 2k + |D| \\
            & \geqslant \frac{7}{3}k - \frac{\alpha}{2}
        \end{align*}
        which contradicts~(\ref{eq:14/3k}).
        \item if $|D| \leqslant \frac{k}{3} - \frac{\alpha}{2}$, then we have:
        \begin{align*}
           |V(P) \cap \mathcal{B}_v| 
            & \geqslant (k - \alpha) + 2(k-|D|) \\
            & \geqslant (k - \alpha) + 2\left(\frac{2k}{3} + \frac{\alpha}{2}\right) \\
            & \geqslant \frac{7}{3}k
        \end{align*}
        which also contradicts~(\ref{eq:14/3k}).
    \end{itemize}
    In all cases, we get a contradiction. Thus, there exists $d_v \in \{1, \ldots, k-1\}$ and a vertex $v \in X \cap C_v(d_v)$ such that $|P \cap C_v(d_v)| = 2$. 
\end{subproof}
    
Using Claim~\ref{claim:14/3k 4}, let us finish the proof by showing that $u$ rejects at step \ref{step:14/3:finalcase}. %step~(viii).

Let $u \in C_u(k-1)$, $d_v \in \{1, \ldots, k-1\}$ and $v \in P \cap C_v(d_v)$ at distance at most $k$ from $u$ such that $|P \cap C_v(d_v)| = 2$. Let $v'$ be the other vertex in $P \cap C_v(d_v)$. Then, $u$ rejects at step~(viii) of the verification. Indeed, there are several cases (see Figure~\ref{fig:P14/3k cases (viii)} for an illustration):
\begin{itemize}
    \item If $P$ has both endpoints in $V \setminus \left( \bigcup_{d \leqslant d_v} C_v(d) \right)$, the set of conditions~(c) is satisfied.
    \item If $P$ has both endpoints in $\bigcup_{d \leqslant d_v} C_v(d)$, the set of conditions~(d) is satisfied.
    \item If $P$ has exactly one endpoint in $\bigcup_{d \leqslant d_v} C_v(d)$, one of the sets of conditions~(a) or~(b) is satisfied. Indeed,  Claim~\ref{claim:14/3k 4} proves more precisely that $u \in P \cap C_u(k-1)$, and that $v$ is among the $k$ next vertices after $u$ in $P$. Let $w$ be the vertex after $v$ in $P$ (which exists, otherwise $v$ is an endpoint of $P$ which would have both endpoints in $\bigcup_{d \leqslant d_v} C_v(d)$). %\linda{<- @Sebastian, I am not sure why if $v$ is an endpoint of $P$ the other endpoint is in $\bigcup_{d \leqslant d_v} C_v(d)$}
  %  \textcolor{blue}{SZ: Because $v$ is on the part $P'$ of $P$ which goes from $C_u$ to $C_v$. Take $v$ as being the first vertex in $C_v(d_v)$ in the $k$ next vertices after $u$. If it is an endpoint, $P'$ does not reach $C_v$}
  %\linda{I confess I'm too tired to follow this reasoning, but I am willing to believe you are correct}
    If $w \in \bigcup_{d \leqslant d_v} C_v(d)$, then we are in case~(a), else we are in case~(b).
\end{itemize}

Conversely, let us show that if $G$ is $P_m$-free, then no vertex rejects with the certificates described above. By definition, no vertex will reject in steps~(i), (ii), (iii) and~\ref{step:check-constrained-parti one}.
By contradiction, assume that there exists $u \in V$ which %rejects at step~(v), (vi), (vii) or~(viii).
\ref{step:visiblempath}, \ref{step:uselp-simple}, \ref{step:threeECCS}   or \ref{step:14/3:finalcase}. 
\begin{itemize}
    \item If $u$ rejects at \ref{step:visiblempath}, there exists an induced path $P$ of length $m$ in $G_{\preccurlyeq u}$ such that all the vertices in $P \setminus V_{\preccurlyeq u}$ are in distinct ECC$_2$'s. So $P$ is also induced in $G$.

    \item If $u$ rejects at step~\ref{step:uselp-simple}, with the same proof as in Theorem~\ref{thm:p_4k_free_subquadratic}, we can reconstruct an induced path $P$ of length $m$ in $G$, which is a contradiction.
%LC: For future versions we should modularize this proof, so we can cite a lemma from the proof of \cref{thm:p_4k_free_subquadratic}. It also makes sense for us to write the prototocl for thm:p_4k_free_subquadratic in terms of a variable m for the length so that we can say in case of thm:p_4k_free_subquadratic we take m=4k, and in case of this proof we take m = 14/3k-1, so that we don't need to rewrite the steps of the certification where they are the same between the two protocol and can more easily invoke lemmas from the proof of hm:p_4k_free_subquadratic in the 14/3 context
    \item If $u$ rejects at step~\ref{step:threeECCS}, we can also reconstruct an induced path of length~$m$ in~$G$. Indeed, let $P_v$ (resp. $P_w$) be the path of length $\lp[v]$ (resp. $\lp[w]$) which starts at $v$ (resp. $w$) and has all its other vertices in $\bigcup_{d<d_v}C_v(d)$ (resp. in $\bigcup_{d<d_w}C_w(d)$). Then, we obtain an induced path of length $m$ by gluing $P_v$, $P_{\text{start}}$ and $P_w$.

    \item Finally, if $u$ rejects at step~\ref{step:14/3:finalcase}, we can also construct an induced path $P$ of length $m$ in $G$, as depicted on Figure~\ref{fig:P14/3k cases (viii)}. In each case, the definition of~$\lcp[v]$ ensures that the path obtained by gluing the parts is still induced in $G$.
\end{itemize}
In all cases, we obtain an induced path $P$ of length $m$ in~$G$, which is a contradiction. Thus, if $G$ is $P_m$-free, all vertices accept with the certificates described above. 
\end{proof}

% \subsection{Omitted verification part of the proof Theorem~\ref{thm:computation scheme for T_G}}
% \label{app:verif-table}
% \input{disc/disc-app-verif-table}

\newpage{}
\bibliography{biblio}
%\newpage{}

%\appendix

%\input{stacs/app-additional-related-work}

%\section{Completing the proof of our lower bounds}
%\label{app:lower-bounds}
%\input{stacs/stacs-app-lower-bounds}

%\input{stacs/stacs-app-eccs}

%\input{stacs/stacs-app-advanced-upper-bounds}

%\appendix

% \newpage{}

% \appendix

% \textbf{Table of Content of the Appendix}

% \medskip{}
% \begin{itemize}
%     \item Section~\ref{app:lower-bounds}: Omitted proofs from Section~\ref{sec:lower-bounds}
%     \item Section~\ref{app:warm-up-upper-bound}: Omitted proofs of Subsection~\ref{subsec:warm-up}
%     \item Section~\ref{app:verif-table}: Omitted verification part of the proof Theorem~\ref{thm:computation scheme for T_G}
%     %\item Section~\ref{app:three-halves-upper-bound}: Omitted proofs from Subsection~\ref{subsec:three-halves-upper-bound}
%     \item Section~\ref{app:advanced-upper-bounds}: Formal proofs of the forbidden subgraph certifications
% \end{itemize}

% \medskip{}

% \section{Omitted proofs from Section~\ref{sec:lower-bounds}}
% \label{app:lower-bounds}
% \input{app-lower-bounds}

% \section{Omitted proofs of Subsection~\ref{subsec:warm-up} }
% \label{app:warm-up-upper-bound}
% \input{app-warm-up-upper-bound}

% \section{Omitted verification part of the proof Theorem~\ref{thm:computation scheme for T_G}}
% \label{app:verif-table}
% \input{app-verif-table}

% %\section{Omitted proofs from Subsection~\ref{subsec:three-halves-upper-bound}}
% %\label{app:three-halves-upper-bound}
% %\input{app-three-halves-upper-bound}

% \input{app-advanced-upper-bounds}

\end{document}